\documentclass[runningheads]{llncs}
\usepackage[T1]{fontenc}
\usepackage{graphicx}
\usepackage{amsmath}
\usepackage{amssymb}
\usepackage{todonotes}
\usepackage[linesnumbered,boxed,lined,noend]{algorithm2e}
\usepackage{makecell}
\usepackage{comment}
\usepackage{hyperref}

\makeatletter
\newcommand\footnoteref[1]{\protected@xdef\@thefnmark{\ref{#1}}\@footnotemark}
\makeatother

\newtheorem{fact}[theorem]{Fact}
\newtheorem{observation}[theorem]{Observation}

\newcommand{\FAST}{\ensuremath{\mathsf{FAST}}}
\newcommand{\FVST}{\ensuremath{\mathsf{FVST}}}
\newcommand{\K}{\ensuremath{\kappa}}

\newcommand{\ADT}{\ensuremath{\mathsf{ADT}}}

\newcommand{\dspTriangle}[1]{\ensuremath{\mathbb{DTP}[#1]}}
\newcommand{\dspTrianglePlus}[1]{\ensuremath{\mathbb{DTP^+}[#1]}}
\newcommand{\dsTriangle}[1]{\ensuremath{\mathbb{DT}[#1]}}
\newcommand{\dsTrianglePlus}[1]{\ensuremath{\mathbb{DT^+}[#1]}}

\newcommand{\INITNP}[2]{\ensuremath{\mathtt{Initialize}(#1,#2)}}
\newcommand{\REVERSE}[2]{\ensuremath{\mathtt{Reverse}(#1,#2)}}
\newcommand{\FTRIANGLE}[1]{\ensuremath{\mathtt{FindTriangle(#1)}}}

\newcommand{\FINDFAST}[1]{\ensuremath{\mathtt{FindFAST(#1)}}}
\newcommand{\FINDFVST}[1]{\ensuremath{\mathtt{FindFVST(#1)}}}

\newcommand{\REMOVE}[1]{\ensuremath{\mathtt{Remove}(#1)}}
\newcommand{\RESTORE}[1]{\ensuremath{\mathtt{Restore}(#1)}}

\newcommand{\db}[1]{\ensuremath{\mathtt{Db}_{#1}}}
\newcommand{\LONG}{\ensuremath{\mathtt{LONG}}}
\newcommand{\mindt}[1]{\ensuremath{\mathtt{min}_{\mathsf{d}}(#1)}}
\newcommand{\maxdb}[1]{\ensuremath{\mathtt{max}_{\mathsf{Db}}(#1)}}
\newcommand{\induced}[2]{\ensuremath{#1[#2]}}
\newcommand{\nopref}[1]{\ensuremath{{#1}^{\mathtt{-pref}}}}
\newcommand{\noset}[2]{\ensuremath{{#1}^{-{#2}}}}

\newcommand{\emptydb}{\ensuremath{\mathtt{Empty}}}
\newcommand{\deltadb}[3]{\ensuremath{\mathtt{Affected}_{#1}(#2,#3)}}
\newcommand{\backdb}{\ensuremath{\mathtt{Back}}}

\newcommand{\dsBasic}[1]{\ensuremath{\mathbb{D}[#1]}}

\newcommand{\MIN}{\ensuremath{\mathtt{FindFirstAfterPrefix}()}}
\newcommand{\INCOMING}[2]{\ensuremath{\mathtt{GetInNeighboursOmittingPrefix}(#1, #2)}}

\newcommand{\REVERSEBND}[3]{\ensuremath{\mathtt{Reverse}(#1,#2,#3)}}

\newcommand{\ROOT}{\ensuremath{\mathtt{GetRoot}()}}
\newcommand{\LEFT}[1]{\ensuremath{\mathtt{GetLeft}(#1)}}
\newcommand{\RIGHT}[1]{\ensuremath{\mathtt{GetRight}(#1)}}
\newcommand{\RECT}[2]{\ensuremath{\mathtt{SumRectangle}(#1, #2)}}

\newcommand{\degrees}{\ensuremath{\mathtt{Degrees}}}

\newcommand{\dsSome}[1]{\ensuremath{\mathbb{S}[#1]}}

\newcommand{\dsRemove}[2]{\ensuremath{\mathbb{DREM}[#1,#2]}}

\newcommand{\dsRemovePromise}[2]{\ensuremath{\mathbb{DREMP}[#1,#2]}}

\newcommand{\dist}{\ensuremath{\mathtt{indist}}}

\newcommand{\lbackdb}{\ensuremath{\mathtt{LBack}}}

\newcommand{\rmSet}{\ensuremath{\mathtt{Removed}}}
\newcommand{\reducedDegree}[1]{\ensuremath{\mathtt{RDeg}[#1]}}
\newcommand{\reducedDegreeBefore}[1]{\ensuremath{\mathtt{RDeg}^{(b)}[#1]}}

\newcommand{\longArcs}[1]{\ensuremath{\mathtt{LArcs}[#1]}}

\newcommand{\reducedDegreeBucket}[1]{\ensuremath{\mathtt{RDb}[#1]}}
\newcommand{\reducedEmptydb}{\ensuremath{\mathtt{REmpty}}}
\newcommand{\tok}{\ensuremath{\mathtt{TOK}}}
\newcommand{\ctok}[1]{\ensuremath{\mathtt{CTok}(#1)}}
\newcommand{\ctokBefore}[1]{\ensuremath{\mathtt{CTok}^{(b)}(#1)}}

\newcommand{\vtok}[1]{\ensuremath{\mathtt{VTok}[#1]}}

\newcommand{\INITIALIZEDREM}[3]{\ensuremath{\mathtt{Initialize}(#1,#2,#3)}}
\newcommand{\REVERSEARCDREM}[1]{\ensuremath{\mathtt{ReverseArc}(#1)}}
\newcommand{\REMOVEVERTEX}[1]{\ensuremath{\mathtt{RemoveVertex}(#1)}}
\newcommand{\RESTOREVERTEX}[1]{\ensuremath{\mathtt{RestoreVertex}(#1)}}

\newcommand{\fixRdeg}[1]{\ensuremath{\mathtt{NewRd}(#1)}}

\newcommand{\INITIALIZEDREMP}[2]{\ensuremath{\mathtt{Initialize}(#1,#2)}}
\newcommand{\REVERSEARCDREMP}[1]{\ensuremath{\mathtt{ReverseArc}(#1)}}

\newcommand{\heavy}[2]{\ensuremath{\mathtt{Heavy}_{#2}(#1)}}
\newcommand{\degree}[2]{\ensuremath{deg_{#1}(#2)}}

\newcommand{\MINFVST}{\ensuremath{\mathtt{FindFirstAfterPrefixNoF}()}}
\newcommand{\INCOMINGFVST}[2]{\ensuremath{\mathtt{GetInNeighboursOmittingPrefixNoF}(#1, #2)}}
\newcommand{\FTRIANGLEFVST}[1]{\ensuremath{\mathtt{FindTriangleNoF(#1)}}}

\begin{document}
\title{Dynamic Parameterized Feedback Problems in Tournaments}
%
%
\author{Anna Zych-Pawlewicz\inst{1}
\and
Marek Żochowski\inst{1}
}
\authorrunning{Anna Zych-Pawlewicz and Marek Żochowski}
%
\institute{University of Warsaw}
\maketitle              
\begin{abstract}
In this paper we present the first dynamic algorithms for the
problem of {\sc{Feedback Arc Set}} in Tournaments ($\FAST$) and
the problem
of {\sc{Feedback Vertex Set}} in Tournaments ($\FVST$). Our algorithms
maintain a dynamic tournament on $n$ vertices altered by redirecting
the arcs,
and answer if the tournament admits a feedback arc set (or
respectively feedback vertex set) of size at most $K$, for some
chosen parameter $K$.
For dynamic $\FAST$ we offer two algorithms.
In the promise model, where we are guaranteed, that the size of
the
solution does not exceed $g(K)$ for some computable function $g$,
we give an $O(\sqrt{g(K)})$ update and $O(3^K K \sqrt{K})$ query
algorithm. In the general setting without any promise, we offer an
$O(\log^2 n)$ update and $O(3^K K \log^2 n)$ query time algorithm
for dynamic $\FAST$. For dynamic $\FVST$ we offer an algorithm
working in the promise model, which admits $O(g^5(K))$ update and
$O(3^K K^3 g(K))$ query time.
\end{abstract}


\section{Introduction and Related Work}\label{r:intro}
In this paper we study feedback problems in dynamic tournaments.
The problems we focus on here is the {\sc{Feedback Arc Set}}
problem and the
{\sc{Feedback Vertex Set}} problem. The feedback arc
(resp. vertex) set of
a given graph
is a set of arcs (resp. vertices) whose removal makes the graph
acyclic. In the classical static setting, given a graph and a parameter
$K$,
the {\sc{Feedback Arc Set}} problem (resp.
{\sc{Feedback Vertex Set}} problem),
asks if there exists a
feedback
arc (resp. vertex) set of size at most $K$. Both these problems
are
flag problems in the area of fixed parameter tractable
algorithms and
as
such received a lot of interest, see~\cite{article} for the
survey.

A tournament is a directed graph where every pair of vertices is
connected by exactly one arc. Feedback arc sets in tournaments
have applications in voting systems and rank aggregation,
and are
well studied from the combinatorial~\cite{Erds1965OnSO,FERNANDEZDELAVEGA1983328,REID1970225}  and
algorithmic~\cite{10.5555/1283383.1283426,10.1145/1250790.1250806,10.1145/1798596.1798608,DBLP:conf/icalp/AlonLS09,feige2009faster,FominFAST} points of view.
Unfortunately, the
{\sc{Feedback Arc Set}} problem in tournaments ($\FAST$) is
NP-hard~\cite{doi:10.1137/050623905}. The
fastest parametrized algorithm achieves
$2^{O(\sqrt{K})} n^{O(1)}$
running time~\cite{FominFAST}, where $n$ is the size of the input
tournament. In this paper, however, we explore a textbook
branching
algorithm for this problem running in $3^K n^{O(1)}$
time~\cite{FASTBranching}.

The {\sc{Feedback Vertex Set}} problem in tournaments ($\FVST$)
also has
many interesting applications, for instance in social choice theory
where it is
essential to the definition of a certain type of election
winners~\cite{RePEc:clt:sswopa:524}.
It is also NP-hard~\cite{10.5555/88208.88256}
and has been studied from various
angles~\cite{kumar_et_al:LIPIcs.STACS.2016.49,10.1145/3446969,10.1007/978-3-642-15775-2_23,mnich_et_al:LIPIcs:2016:6409}.
The fastest currently known parametrized algorithm for this
problem
runs in $O(1.6181^K + n^{O(1)})$
time~\cite{kumar_et_al:LIPIcs.STACS.2016.49}.

Recently, there is a growing interest in studying parametrized
problems in a dynamic setting. In this context, we typically
consider an instance $I$ of a problem of interest with an
associated parameter $K$. The instance is dynamic, i.e., it is
updated over time by problem specific updates, while for simplicity
we assume that the problem parameter $K$ does not change
through the entire process. The goal is to be able
to efficiently provide the answer to the problem
after each
update. The time needed for that is called the query time.
The desired update/query running time
may depend in any computable way on the parameter $K$, but
it should be sublinear in terms of the size of $I$. The typicall
update/query running times in this setting are $f(K)$,
$f(K) (\log n)^{O(1)}$, or sometimes even $f(K) n^{o(1)}$,
where $f$ is some computable function and $n$ is the size of $I$.
Since we allow $f$ to be exponential, this setting applies
to NP-hard problems as long as they are fixed-parameter
tractable in the static setting.

Parameterized dynamic data structures were first systematically
investigated by Iwata and Oka~\cite{IwataO14}, followed by
Alman et al.~\cite{Alman}. These works provided data
structures with update times $f(K)$ or
$f(K)\cdot (\log n)^{O(1)}$ for several classic problems such
as {\sc{Vertex Cover}}, {\sc{Cluster Vertex Deletion}},
{\sc{Hitting Set}}, {\sc{Feedback Vertex Set}},
or {\sc{Longest Path}}. Other recent advances include data
structures for maintaining various graph decompositions together
with runs of dynamic programming
procedures~\cite{10.5555/3458064.3458114,DvorakKT14,DvorakT13,korhonen2023dynamic,majewski_et_al:LIPIcs.STACS.2023.46}
and
treatment of parameterized string problems from the dynamic
perspective~\cite{olkowski_et_al:LIPIcs.STACS.2023.50}.

Alman et al. in their work~\cite{Alman} study the problem of
{\sc{Feedback Vertex Set}} in dynamic undirected graphs, where
the graph is altered
by edge additions and removals and the goal is to test if
the graph
has a
feedback vertex set of size at most $K$. For
this problem, Alman et al. propose a dynamic algorithm with
$2^{\mathcal{O}(K \log K)} \log^{\mathcal{O}(1)} n$ amortized update time
and $\mathcal{O}(K)$ query time. It is worth mentioning, that while the
{\sc{Feedback Vertex Set}} problem is
NP-hard on undirected graphs, the {\sc{Feedback Arc Set}}
problem is polynomial in this class of graphs and can be efficiently
maintained dynamically using dynamic connectivity algorithms~\cite{Holm}.
So an interesting question is whether these two problems admit efficient
dynamic algorithms in the class of directed graphs.
In this regard Alman et al.~\cite{Alman} show lower bounds
 for the
{\sc{Feedback
Vertex Set}} problem, which extend also
to the {\sc{Feedback Arc Set}}
problem. 
To be more precise, they show
that in this case the dynamic algorithm requires $\Omega(f(K) m^{\delta})$
update/query time for some fixed $\delta > 0$,
assuming RO hypothesis(see~\cite{10.1145/3395037}).
Thus, a natural question is whether we can have more
efficient dynamic algorithms
at least in tournaments. In this paper we give positive answers to this
question.

Our precise setting is the following.
In the dynamic $\FAST$ (respectively $\FVST$)
problem, the goal is to design a data structure, which maintains
a dynamic
tournament under an operation of reversing an arc,
and can verify upon a
query whether the maintained tournament has a feedback arc set
(respectively
feedback vertex set) of size at most $K$, where $K$ is
the problem parameter.\footnote{For the sake of comparison we
formulate
our results in this setting, however our data structures do not
need to
know the parameter $K$ apriori and the queries can work with
different
values of the parameter.} This setting is referred to as
the
\emph{full} model. A popular restriction
of this setting, introduced in generality
by Alman et. al.~\cite{Alman},
is called the \emph{promise} model.
The promise model applied to our setting provides
the data structure with a guarantee,
that the feedback arc set (or
respectively
feedback vertex set) remains of size bounded by $g(K)$ for some
computable
function $g$ for the entire process. Some algorithms provided by
Alman et al.~\cite{Alman} work
only in this restricted setting.

\paragraph{\textbf{Main results.}}
For dynamic $\FAST$ we offer two algorithms.
In the promise model we give an $O(\sqrt{g(K)})$ update and
$O(3^K K \sqrt{K})$ query algorithm. In the full model, we offer an
$O(\log^2 n)$ update and $O(3^K K \log^2 n)$ query time algorithm
for dynamic $\FAST$. For dynamic $\FVST$ we offer an algorithm
in the promise model, which admits $O(g^5(K))$ update and
$O(3^K K^3 g(K))$ query time. Our results, compared to the related
results of Alman et al.~\cite{Alman} are shown in the table below.
All our running times are worst case (i.e., not amortized).
   \begin{center}
        \begin{tabular}{ | c | c | c | c | }
          \hline
          \thead{Problem} & \thead{Undirected} & \thead{Directed} & \thead{Tournament} \\
          \hline
          \makecell{FAS} & \makecell{$\log^{\mathcal{O}(1)}n$ update\\
                                     $\log^{\mathcal{O}(1)}n$ query} &
          \makecell{no $f(K)m^{o(1)}$\\ algorithm\\ assuming RO} &
          \makecell{\texttt{full model}:\\$\mathcal{O}(\log^2{n})$ update \\ $\mathcal{O}(3^K K \log^2{n})$ query\\
          \texttt{promise model}:\\
          $O(\sqrt{g(K)})$ update\\
          $O(3^K K \sqrt{K})$ query
          } \\
          \hline
          \makecell{FVS} &
          \makecell{$2^{\mathcal{O}(K \log{K})} \log^{\mathcal{O}(1)}{n}$\\
          amortized update \\
          $\mathcal{O}(K)$ query} & \makecell{no $f(K) m^{o(1)}$\\ algorithm\\ assuming RO} & \makecell{\texttt{promise model:} \\ $\mathcal{O}(g^5(K))$ update \\ $\mathcal{O}(3^K K^3 g(K))$ query} \\
          \hline
        \end{tabular}\label{results_overview}
    \end{center}
In the next sections we provide a brief overview of our techniques
and the main
ideas how they are used to obtain our results. The rigorous proofs,
due to space limitation, are moved to the appendix.

\section{Triangle Detection}

Our ultimate goal is to use in the dynamic setting
the standard branching algorithms
(see Algorithm~\ref{fig:branchingFASToverview} and
Algorithm~\ref{fig:branchingFVSToverview}), which rely
on efficiently locating a triangle in a tournament.
Hence, the basis for our algorithms are the triangle detection data
structures presented in Section~\ref{r:triangles},
that might be of independent interest.
The data structures we provide are stated independently of
the $\FAST$ or $\FVST$ problems. They rely on a different parameter,
which is the maximum number of arc-disjoint triangles in the
tournament. Nevertheless, our later results rely on a close
connection between the maximum number of arc-disjoint triangles and
the minimum feedback arc set of a tournament. Throughout the paper,
for a given tournament $T$, we denote by $\FAST(T)$ the size of the
minimum feedback arc set in $T$, by $\ADT(T)$ the maximum number of
arc-disjoint triangles in $T$, and by $\FVST(T)$ the size of the
minimum feedback vertex set in $T$. The following fact holds.
\begin{fact}\label{arc_disjoint_triangles_and_fast}
    $\ADT(T) \leq \FAST(T) \leq 6(\ADT(T) + 1)$.
\end{fact}

\begin{proof}
    It is easy to see that $\ADT(T) \leq \FAST(T)$, as in each of the $\ADT(T)$ arc-disjoint triangles one arc must be taken to the feedback arc set of $T$. The proof of the second inequality can be found in~\cite{krithika2018parameterized} (Theorem 4).
\end{proof}

The triangle detection data structures maintain the dynamic tournament altered by reversing arcs, and allow queries for a triangle in the maintained tournament. The first data structure is given by Theorem~\ref{thm:DsTrianglePromiseOverview} below. It is later used to solve dynamic $\FAST$ in the promise model.

\begin{theorem}[Theorem~\ref{thm:DsTrianglePromise} in Section~\ref{r:triangles}]\label{thm:DsTrianglePromiseOverview}
For any integer $n \in \mathbb{N}$ there exists a data structure $\dspTriangle{n}$, that maintains a dynamically changing tournament $T$ on $n$ vertices\footnote{\label{note1}The data structure assumes that the vertices of $T$ are indexed by numbers in $[n]$.}
by
supporting the following operations:
\begin{enumerate}
    \item $\INITNP{T}{n}$ -- initializes the data structure with a given tournament $T$ on $n$ vertices, in time $\mathcal{O}(n^2)$
    \item $\REVERSE{v}{u}$ -- reverses an arc between vertices $v$ and $u$ in $T$, in $\mathcal{O}(\sqrt{\ADT(T)})$ time
    \item $\FTRIANGLE{}$ -- returns a triangle from $T$ or reports that there are no triangles, in time $\mathcal{O}(\ADT(T) \sqrt{\ADT(T)})$
\end{enumerate}
\end{theorem}

We note here that the above data structure does not need to know
the value of $\ADT(T)$. The same holds for the second data structure, which we present next.
The second triangle detection data structure gets rid of the dependence on the parameter in the update operation. This later allows us to use it to solve the dynamic $\FAST$ in the full model at the cost of introducing factors poly-logarithmic in the size of the tournament.

\begin{theorem}[Theorem~\ref{thm:DsTriangleFull} in Section~\ref{r:triangles}]\label{thm:DsTriangleFullOverview}
For any integer $n \in \mathbb{N}$ there exists a data structure $\dsTriangle{n}$, that maintains a dynamically changing tournament $T$ on $n$ vertices\footnoteref{note1} 
by
supporting the following operations:
\begin{enumerate}
    \item $\INITNP{T}{n}$ -- initializes the data structure with a given tournament $T$ on $n$ vertices, in time $\mathcal{O}(n^2)$
    \item $\REVERSE{v}{u}$ -- reverses an arc between vertices $v$ and $u$ in $T$, in time $\mathcal{O}(\log^2 n)$
    \item $\FTRIANGLE{}$ -- returns a triangle from $T$ or reports that there are no triangles, in time $\mathcal{O}(\ADT(T) \log^2 n)$
\end{enumerate}
\end{theorem}

Both data structures rely on the same information related to the
maintained dynamic tournament $T$. First of all, both data structures
partition vertices of $T$ into \emph{indegree buckets}, where the
indegree bucket $\db{T}[d]$ stores vertices of indegree $d$. By $\maxdb{T}$ we
denote the maximum size of an indegree bucket. The tournament is
acyclic if and only if all indegree buckets have size one
(see Fact~\ref{transitive_buckets}). Both data structures also
maintain
a set $\emptydb$ of empty indegree buckets,
which stores all the values of $d$ for which
$\db{T}[d]$ is empty. The promise data structure additionally
maintains
a set $\backdb$ of all \emph{back arcs}. An arc $uv$ is called
a back arc if $d_T(u) \geq d_T(v)$, where $d_T(x)$ stands for
the indegree of $x$ in $T$. It is relatively easy to maintain
indegree buckets and the set of empty buckets under arc reversals in
time $O(1)$ per update,
as each arc reversal affects the indegrees of exactly two vertices
(see Lemma~\ref{update_impl_basic} for the details). Maintaining
the set of back arcs $\backdb$, however, takes
$O(\sqrt{\ADT(T)})$ per update
(due to Lemma~\ref{back_arcs_changes_nok}).

To support detecting triangles,
we first define a prefix $P$ of the
tournament $T$
as the union of the indegree buckets
in the maximum prefix of indegree buckets such that
all indegree buckets in the prefix have size one (see
Definition~\ref{prefix_def} for a more formal definition).
Thus, the subtournament
$\induced{T}{P}$
of $T$ induced by $P$ is acyclic.
By $\nopref{T}$ we denote
the tournament $T$ without its prefix.

Both data structures follow the same approach. In order to find a
triangle, they first locate a vertex
$v$ of minimum indegree
in $\nopref{T}$.
We can prove, that $d_{\nopref{T}}(v)$ is bounded
in terms of $\ADT(T)$ (Corollary~\ref{min_deg_Tprim}). Once $v$ is
located, the data structures find an in-neighbor $u$ of vertex $v$ in
$\nopref{T}$,
and then they find a set $W$ of $d_{\nopref{T}}(v)$ in-neighbors
of $u$ in $\nopref{T}$. Since $d_{\nopref{T}}(v) \leq d_{\nopref{T}}(u)$,
vertex $u$ has at least $d_{\nopref{T}}(v)$ in-neighbors in $\nopref{T}$.
We are bound to find a triangle $uvw$ for
some $w \in W$,
because if all the edges between $W$ and $v$ were directed
towards $v$, that would imply that the indegree of $v$ in
$\nopref{T}$ is more than $d_{\nopref{T}}(v)$.

Thus, in order to support the $\FTRIANGLE{}$
method, it suffices to implement
two methods: one to locate a vertex
$v \in V(\nopref{T})$ minimizing
$d_{\nopref{T}}(v)$, and the second
one to list $l$ in-neighbors in
$\nopref{T}$ of a given vertex
$u \in \nopref{T}$.
Both these methods use the stored
structures such as indegree buckets
$\db{T}[d]$ and the sets $\emptydb$
and $\backdb$ in order to implement
these methods. For instance,
the method for locating vertex $v$
is the same for both data structures
and consists of iterating through
the set $\emptydb$ to find the first
index $d$ that is missing in
$\emptydb$ (and retrieving a vertex
from the corresponding indegree bucket, see
Lemma~\ref{empty_impl_promise} for
the details).
Thus, it takes
$O(|\emptydb|)$ time to locate $v$.
The method responsible for finding in-neighbors is different for the
two data structures. In the promise data structure $\dspTriangle{n}$,
this method (provided in
Lemma~\ref{incoming_impl_promise})
heavily relies on the fact that we can iterate over the set
of back arcs $\backdb$. In fact, for
$\dspTriangle{n}$ data structure
it
takes $O(|\backdb| + |\emptydb| + l)$ time to find $l$ in-neighbors
of a given vertex  in $\nopref{T}$ (see
Lemma~\ref{incoming_impl_promise}).
As mentioned before, maintaining set $\backdb$
takes $O(\sqrt{\ADT(T)})$ time per update, which we want to
avoid if there is no promise. Hence,
the $\dsTriangle{n}$ data structure uses balanced search trees to
look for in-neighbors, what implies poly logarithmic factors in the
update and query times.
For $\dsTriangle{n}$ data structure,
it takes
$O(l \cdot \log^2{n} + |\emptydb|)$ time to find $l$
in-neighbors
of a given vertex  in $\nopref{T}$ (see Lemma~\ref{incoming_full_impl} for
the details and the proof).
As a result of the above,
with regard to triangle detection,
most of our running times stem
from the fact, that we are able to
bound
$|\emptydb|$, $|\backdb|$ and $\maxdb{T}$ in terms of $\ADT(T)$
(Lemma~\ref{small_buckets}, Lemma~\ref{empty_small} and Lemma~\ref{few_back_arcs}).

\section{Dynamic $\FAST$ Algorithm}
We now give an idea how we use the triangle detection data structures for
the dynamic $\FAST$ problem. This is the topic of
Section~\ref{sec:dynamicFast}. Let us consider a standard branching
algorithm verifying whether $\FAST(T) \leq K$ shown in
Algorithm~\ref{fig:branchingFASToverview}
(see~\cite{FASTBranching} for correctness).
Observe, that Algorithm~\ref{fig:branchingFASToverview} can be executed
using both of the provided triangle detection data structures.

\begin{algorithm}
    \SetKwInOut{Input}{Algorithm}
	\SetKwInOut{Output}{Output}
    \Input{$\FINDFAST{K}$}
    \Output{Verify if $\FAST(T) \leq K$}
    \If{$T$ is acyclic}{\Return TRUE \;}
    \If{$K=0$}{ \Return FALSE \;}

    $uvw \gets \FTRIANGLE{}$\;

    \For{$xy \in \{ uv, vw, wu \}$}{
       $\REVERSE{x}{y}$ \;
           \If{$\FINDFAST{K-1}$}{
                   $\REVERSE{y}{x}$\;
                   \Return TRUE \;
            }
            $\REVERSE{y}{x}{}$\;

}
    \Return FALSE \;
    \caption{Pseudocode for $\FINDFAST{K}$}
    \label{fig:branchingFASToverview}
\end{algorithm}

We first show how to obtain the dynamic $\FAST$ algorithm in the
promise model. Let $K$ be the parameter, and let us assume that at
all times $\FAST(T) \leq g(K)$.
Then, by Fact ~\ref{arc_disjoint_triangles_and_fast},
also $\ADT(T) \leq g(K)$.
So, extending the data structure of
Theorem~\ref{thm:DsTrianglePromiseOverview}
with the $\FINDFAST{K}$ method gives a dynamic algorithm for
$\FAST$ in the promise model with $\mathcal{O}(\sqrt{g(K)})$ update
and $\mathcal{O}(3^K g(K) \sqrt{ g(K) })$ query time. With a bit
of extra technical work, we can use the data structure of Theorem~\ref{thm:DsTrianglePromiseOverview} to obtain a better result stated below.
\begin{theorem}[Theorem~\ref{thm:FastPromise} in Section~\ref{sec:dynamicFast}]\label{thm:FastPromiseOverview}
    The dynamic \emph{FAST} problem with parameter $K$ admits a data
    structure\footnoteref{note1} with initialization time
    $\mathcal{O}(n^2)$, worst-case update time
    $\mathcal{O}(\sqrt{g(K)})$ and
    worst-case query time $\mathcal{O}(3^K K \sqrt{K})$ under
    the
    promise that there is a computable function $g$, such that the
    maintained tournament $T$ always has a feedback arc set of size at
    most $g(K)$.
\end{theorem}

Let us now give an idea of how to use the data structure
of Theorem~\ref{thm:DsTriangleFullOverview} for the dynamic $\FAST$
in the full model. We can extend also this data structure with the
$\FINDFAST{K}$ method given in
Algorithm~\ref{fig:branchingFASToverview}. However, the problem is
that the running time of $\FTRIANGLE{}$ depends on $\ADT(T)$, and we
are not working in the promise model anymore. There is a remedy to
this. We can provide $\dsTriangle{n}$ with an extra method
$\FTRIANGLE{d}$, which  takes a parameter $d$, and essentially runs
the previous $\dsTriangle{n}.\FTRIANGLE{}$ method but only allows it
to perform $\mathcal{O}(d \log^2 n)$ steps. If this number of steps
is not sufficient to find a triangle, then we can deduce, that
$\ADT(T) > d$ and as a consequence also $\FAST(T) > d$. So the
$\dsTriangle{n}.\FTRIANGLE{d}$ is guaranteed to run in
$\mathcal{O}(d \log^2 n)$ time, but it either returns a triangle or
reports that there are more than $d$ arc-disjoint triangles. We can
now extend the data structure $\dsTriangle{n}$ of
Theorem~\ref{thm:DsTriangleFullOverview} with the $\FINDFAST{K}$
method given in Algorithm~\ref{fig:branchingFASToverview}, but we
replace line $uvw \gets \FTRIANGLE{}$ with $uvw \gets \FTRIANGLE{K}$.
If the procedure $\FTRIANGLE{K}$ fails to find a triangle, we can
safely return FALSE, as this means that $\FAST(T) \geq \ADT(T) > K$.
The procedure $\FTRIANGLE{K}$ takes $\mathcal{O}(K \log^2 n)$ time.
This way, we arrive at the following result.
\begin{theorem}[Theorem~\ref{thm:FastFull} in Section~\ref{sec:dynamicFast}]\label{thm:FastFullOverview}
    The dynamic \emph{FAST} problem with a parameter $K$ admits a
    data structure\footnoteref{note1} with initialization
    time $\mathcal{O}(n^2)$, worst-case update time
    $\mathcal{O}(\log^2{n})$ and worst-case query time
    $\mathcal{O}(3^K K \log^2{n})$.
\end{theorem}

\section{Dynamic $\FVST$ Algorithm}
In Section~\ref{sec:dynamicFVST} we present the dynamic data
structure for the $\FVST$ problem with a parameter $K$ in the
promise model. The main idea is the same as for the $\FAST$ problem.
We want to be able to quickly find a triangle in tournament $T$ and
perform a standard branching algorithm for $\FVST$, presented in
Algorithm~\ref{fig:branchingFVSToverview} (its correctness follows
from Fact~\ref{cycles_are_triangles}).
\begin{algorithm}
    \SetKwInOut{Input}{Algorithm}
	\SetKwInOut{Output}{Output}
    \Input{$\FINDFVST{K}$}
    \Output{Verify if $\FVST(T) \leq K$}
    \If{$T$ is acyclic}{\Return TRUE \;}
    \If{$K=0$}{ \Return FALSE \;}

    $uvw \gets \FTRIANGLE{}$\;

    \For{$x \in \{ u, v, w \}$}{
       $\REMOVE{x}$ \;
           \If{$\FINDFVST{K-1}$}{
                   $\RESTORE{x}$\;
                   \Return TRUE \;
            }
            $\RESTORE{x}$\;

}
    \Return FALSE \;
    \caption{Pseudocode for $\FINDFVST{K}$}
    \label{fig:branchingFVSToverview}
\end{algorithm}

The
branching algorithm (Algorithm~\ref{fig:branchingFVSToverview})
finds a triangle in
the tournament, and then
branches recursively with each of the triangle vertices removed.
To implement this algorithm, we
not only need
a method to find a triangle in the
maintained tournament,
but we also need to support vertex
removals and restorations, which are
significantly more complex than
edge reversals, as they change the indegree of up to $(n-1)$ vertices.
This loss of locality poses a number of problems,
including maintaining indegree buckets, maintaining
the set of empty degree buckets,
or even keeping track of the acyclicity of the tournament.
Observe also, that our parameter is now $\FVST(T)$, which
behaves entirely different than $\FAST(T)$. For instance,
our triangle
detection data structures rely on the fact, that a tournament $T$
with $\FAST(T) \leq K$ has a bounded number of back arcs
in terms of the parameter $K$. This, unfortunately, stops to be
the case for $\FVST(T)$: the number of back arcs can be actually
unbounded in terms of $\FVST(T)$. Thus, we cannot iterate over
back arcs when issuing the $\FTRIANGLE{} $ query, as the
$\dspTriangle{n}$ data structure does. Adapting $\dsTriangle{n}$ data
structure to this setting is even more problematic, as updating
the balanced search trees it relies on is inefficient under
the updates of removing or restoring a vertex.
\newline
\newline
\textbf{Vertex removals and restorations.}
We first deal with the problem of removing and restoring
vertices.
In Subsection~\ref{sub:drem} we introduce a new data structure
called $\dsRemove{n}{k}$. The parameter $k$ is
an internal parameter of the data stucture that we will
relate to the problem parameter $K$ later on.
The data structure maintains an $n$-vertex tournament
$T$ while supporting arcs reversals and vertex
removals/restorations.
In order to perform these operations
efficiently the data
structure stores an explicit representation of $T$ without any
vertices removed, the set of removed vertices $F$ and an implicit
representation of $\noset{T}{F}$, where $\noset{T}{F}$ stands for
the tournament $T$ with the set $F$ of vertices removed, i.e,
$\noset{T}{F}=\induced{T}{V(T) \setminus F}$.
Additionally, it
allows at most $k$ vertices to be
removed at the same time. This
data structure is covered by Lemma~\ref{DREM_def_lemma}.

As observed above, we need to be able to quickly
decrease/increase indegrees of many vertices, because
removal/restoration of a vertex can pessimistically affect all
other vertices.
In order to do so, the $\dsRemove{n}{k}$ data structure stores a
set of tokens, one for each removed vertex.
The vertices are partitioned into reduced indegree buckets
that are similar to
indegree buckets from Section~\ref{r:triangles}, and
 a token at position $x$ ,,decreases'' indegrees of
all vertices in the reduced indegree buckets with indices
$d \geq x$.
The difference with respect to Section~\ref{r:triangles}
is that the partitioning into buckets is not according to the
actual indegree in the tournament, but rather according to
a ``reduced degree'' $\reducedDegree{v}$, which is an
approximation of
$d_{\noset{T}{F}}(v)$.
To be more precise, the value of $d_{\noset{T}{F}}(v)$ can be
calculated by subtracting from $\reducedDegree{v}$
the number of tokens at positions
smaller or equal then
$\reducedDegree{v}$.
This is property is required by invariant $1$ of the $\dsRemove{n}{k}$
data structure as stated in Lemma~\ref{DREM_def_lemma}.
The implicit representation of $\noset{T}{F}$ is similar to the representation
that we used in Section~\ref{r:triangles}, that is we maintain
the reduced degrees, the partition of $V(\noset{T}{F})$ into
buckets according to their reduced degrees (i.e.,
reduced indegree buckets), and a
subset of such buckets that are empty.
Apart from $T$, $F$, tokens, and the
implicit representation of tournament
$\noset{T}{F}$,
the $\dsRemove{n}{k}$ data
structure stores the set of all $k$-long back arcs in $T$.
For an arc $uv$ we define its length with respect
to tournament $T$ as $l_T(uv)=|d_T(u)-d_T(v)|$.
An
arc is $k$-long if its length is at least $k$,
otherwise the arc is $k$-short.
In the $\dsRemove{n}{k}$ data structure, every vertex stores a list of
$k$-long back arcs (with regard to $T$) adjacent to it.

In order to remove or restore a vertex $v$,
we first update the stored information corresponding to
vertex $v$. We then
need to fix the reduced degrees of all affected vertices in
order to satisfy invariant $1$.
We first show a method $\fixRdeg{u}$ that
given any vertex $u$, quickly fixes its reduced degree.
We then fix the reduced degree of $v$. Next,
we fix reduced degrees of other vertices in order to fix
invariant $1$ globally.
To accomplish that, we first iterate over vertices having their
reduced
degrees
inside a small interval $[l, r]$ around the reduced degree of $v$,
and we fix all these vertices.
The interval is chosen in a way that all arcs between $v$ and
vertices with reduced degrees outside of $[l, r]$ are $k$-long
arcs with regard to $T$.
By considering all cases, we arrive
at a conclusion that
removing/restoring vertex $v$ does
not alter indegrees
$d_{\noset{T}{F}}(w)$ of vertices
$w$ whose reduced degrees are
smaller then
$l$, and alters
in the same way (either adds $1$ or
subtracts $1$) the indegrees
$d_{\noset{T}{F}}(w)$ of almost all
vertices $w$ whose reduced degrees
are larger then $r$.
Such a regular change to indegrees can be handled
by addition/removal of a token for $v$
(depending whether we remove or restore $v$).
However, there is a small caveat to the above.
The exceptions to this regular change of indegrees
are vertices connected to $v$ via a
$k$-long back arc.
For those vertices, we need to directly fix the reduced degree
of such a vertex.
Because of that, we finish the operation by iterating over
$k$-long back arcs adjacent to vertex $v$
and fixing the reduced degrees of the respective in-neighbors
of $v$. Thus note, that the running times of the $\dsRemove{n}{k}$
data structure depend on the number of $k$-long back arcs adjacent
to $v$,
on $k$ itself and
on the size of the reduced indegree buckets.
\newline
\newline
\textbf{Detecting triangles.}
Next, we need to implement a method
that allows finding triangles in
$\noset{T}{F}$. To
achieve that, we follow ideas from
Section~\ref{r:triangles}. As mentioned before,
the running time of the method $\FTRIANGLE{}$ supported by
$\dspTriangle{n}$ data structure from
Theorem~\ref{thm:DsTrianglePromise}
depends on the maximum size of an indegree bucket $\maxdb{T}$,
the size of the set of empty buckets $|\emptydb|$, the
minimum in-degree in $\nopref{T}$ and the
size of the set of back arcs $|\backdb|$. Among these
parameters, only $\maxdb{T}$ can be bounded
in terms of $\FVST(T)$ (see Lemma~\ref{small_buckets_fvst}).

Our approach to deal with this
issue is to reduce the number of
back arcs by removing
some vertices in $T$.
We again partition the set of back
arcs into
\emph{$k$-long} and \emph{$k$-short} back arcs.
While $k$-short back arcs are relatively easy to
handle, non-trivial ideas are needed to
handle $k$-long back arcs. However, we can bound
$|\emptydb|$ and the
minimum in-degree in $\nopref{T}$ in terms of $\FVST(T)$ and
the number
of $k$-long back arcs (Lemma~\ref{small_empty_fvst} and
Fact~\ref{minimal_degree}).
Hence, reducing the number of
$k$-long back
arcs
is the main issue that we need to solve.
To achieve this, we
define a $k$-long graph $G_{\LONG}$ of the
tournament $T$,
which is an undirected graph, where
vertices are connected via an edge in
$G_{\LONG}$ if they are connected via
a $k$-long back arc in the tournament $T$.
We also define the $k$-heavy set $\heavy{T}{k}$ of the
tournament $T$ as the set of vertices of
degree higher than $k$ in $G_{\LONG}$.
Now if $W \subseteq V(T)$ is a feedback vertex set in $T$ of size at
most $k$, then
$W$ is necessarily a vertex cover in $G_{\LONG}$: there is no other
way to get rid of a $k$-long back arc from $T$ than to remove its
endpoint (see Lemma~\ref{vertex_cover_fvst} for a more formal
argument). Due to standard vertex cover kernelization arguments,
$\heavy{T}{k} \subseteq W$ for any feedback vertex set $W$ of $T$
with $|W| \leq k$. Moreover, when we remove $\heavy{T}{k}$ from
$G_{\LONG}$, at most $k |W|$ edges remains in $G_{\LONG}$. This
simple but crucial observations are covered by
Lemma~\ref{vertex_cover_fvst}.  So the idea is to keep the set
$\heavy{T}{k}$ removed from the tournament, in order to keep the
number of $k$-long back arcs connecting the remaining
vertices small.
In particular, if
$\heavy{T}{k} \subseteq F$, where $F$ is the set of
removed vertices,
any vertex that is left in $\noset{T}{F}$ has at most $k$
adjacent $k$-long back arcs, which makes iterating through them
efficient. We want to emphasize, that the property of
an arc being $k$-long is considered here with regard to
tournament $T$ and it does not depend on $F$, so we can
maintain
$G_{\LONG}$ in form of adjacency lists efficiently.

To implement the above idea,
in Subsection~\ref{sub:dremp} we
introduce a wrapper data structure
around the $\dsRemove{n}{k}$ called
$\dsRemovePromise{n}{k}$.
It allows only one kind of updates, arc
reversals, and keeps the invariant that
the $k$-heavy set of the maintained
tournament is removed.
This not only allows us to efficiently implement
the methods for finding triangles,
but also ensures fast running times of
$\dsRemove{n}{k}$ operations in
the promise model.
The wrapper is defined in
Lemma~\ref{DREMP_def_lemma}. In
order to support the method
$\FTRIANGLE{}$, the data structure
needs to access the information
such as indegree
buckets or the set of
empty indegree buckets with
regard to $\noset{T}{F}$.
Fortunatelly, at a cost of significant
technical effort,
these can be essentially recomputed
from the implicit representation of
$\noset{T}{F}$ stored by
$\dsRemovePromise{n}{k}$.
Moreover, let us say that reduced length of a back
arc $uv$ is defined as
$rl(uv)=|\reducedDegree{u}-\reducedDegree{v}|$.
We extensively explore the fact that there
is a strong relation between $rl(uv)$, $l_T(uv)$ and
$l_{\noset{T}{F}}(uv)$ (see
Observation\ref{monotonic_tokenized_degrees}).
\newline
\newline
\textbf{Dynamic $\FVST$ algorithm in the promise model.}
In order to implement Algorithm~\ref{fig:branchingFVSToverview},
we use $\dsRemovePromise{n}{k}$ data structure for $k=g(K)$, where
$K$ is the problem parameter, to keep $\heavy{T}{k}$ removed
from the tournament, i.e., $\heavy{T}{k} \subseteq F$ at all times,
where $F$ is the set of currently removed vertices. The
implicit representation of $\noset{T}{F}$ similar to the one in
Section~\ref{r:triangles} allows us then to efficiently implement
the $\FTRIANGLE{}$ method. When we branch on the vertices
of the found triangle, we use the methods of the
$\dsRemove{n}{k}$ data structure (internally
maintained by $\dsRemovePromise{n}{k}$), in order to temporarily
remove these vertices from the tournament. This approach leads
to the following theorem.

\begin{theorem}[Theorem~\ref{thm:FvstPromise} in Section~\ref{sec:dynamicFVST}]\label{thm:FvstPromiseOverview}
    The dynamic $\FVST$ with a parameter $K$ admits a data
    structure\footnoteref{note1} with initialization time
    $\mathcal{O}(n^2)$, worst-case update time
    $\mathcal{O}(g(K)^5)$ and worst-case query time
    $\mathcal{O}(3^K K^3 g(K))$ under the promise that there is a
    computable function $g$, such that tournament $T$ always has a
    feedback vertex set of size at most $g(K)$.
\end{theorem}

\begin{credits}
\subsubsection{\ackname} This work was
funded
by National Science Centre, Poland, Grant 2017/26/D/ST6/00264.
\end{credits}

\bibliographystyle{splncs04}
\bibliography{references}

\begin{thebibliography}{10}
\providecommand{\url}[1]{\texttt{#1}}
\providecommand{\urlprefix}{URL }
\providecommand{\doi}[1]{https://doi.org/#1}

\bibitem{Alman}
Alman, J., Mnich, M., Williams, V.V.: Dynamic parameterized problems and
  algorithms. ACM Trans. Algorithms  \textbf{16}(4) (jul 2020).
  \doi{10.1145/3395037}, \url{https://doi.org/10.1145/3395037}

\bibitem{10.1145/3395037}
Alman, J., Mnich, M., Williams, V.V.: Dynamic parameterized problems and
  algorithms. ACM Trans. Algorithms  \textbf{16}(4) (jul 2020).
  \doi{10.1145/3395037}, \url{https://doi.org/10.1145/3395037}

\bibitem{doi:10.1137/050623905}
Alon, N.: Ranking tournaments. SIAM Journal on Discrete Mathematics
  \textbf{20}(1),  137--142 (2006). \doi{10.1137/050623905},
  \url{https://doi.org/10.1137/050623905}

\bibitem{DBLP:conf/icalp/AlonLS09}
Alon, N., Lokshtanov, D., Saurabh, S.: Fast {FAST}. In: Albers, S.,
  Marchetti{-}Spaccamela, A., Matias, Y., Nikoletseas, S.E., Thomas, W. (eds.)
  Automata, Languages and Programming, 36th International Colloquium, {ICALP}
  2009, Rhodes, Greece, July 5-12, 2009, Proceedings, Part {I}. Lecture Notes
  in Computer Science, vol.~5555, pp. 49--58. Springer (2009).
  \doi{10.1007/978-3-642-02927-1\_6},
  \url{https://doi.org/10.1007/978-3-642-02927-1\_6}

\bibitem{RePEc:clt:sswopa:524}
Banks, J.: Sophisticated voting outcomes and agenda control. Working
  Papers~524, California Institute of Technology, Division of the Humanities
  and Social Sciences (1984),
  \url{https://EconPapers.repec.org/RePEc:clt:sswopa:524}

\bibitem{10.5555/3458064.3458114}
Chen, J., Czerwi\'{n}ski, W., Disser, Y., Feldmann, A.E., Hermelin, D., Nadara,
  W., Pilipczuk, M., Pilipczuk, M., Sorge, M., Wr\'{o}blewski, B.,
  Zych-Pawlewicz, A.: Efficient fully dynamic elimination forests with
  applications to detecting long paths and cycles. In: Proceedings of the
  Thirty-Second Annual ACM-SIAM Symposium on Discrete Algorithms. p. 796–809.
  SODA '21, Society for Industrial and Applied Mathematics, USA (2021)

\bibitem{10.1145/1798596.1798608}
Coppersmith, D., Fleischer, L.K., Rurda, A.: Ordering by weighted number of
  wins gives a good ranking for weighted tournaments. ACM Trans. Algorithms
  \textbf{6}(3) (jul 2010). \doi{10.1145/1798596.1798608},
  \url{https://doi.org/10.1145/1798596.1798608}

\bibitem{DvorakKT14}
Dvo\v{r}{\'{a}}k, Z., Kupec, M., T\r{u}ma, V.: A dynamic data structure for
  {MSO} properties in graphs with bounded tree-depth. In: 22th Annual European
  Symposium on Algorithms, {ESA} 2014. Lecture Notes in Computer Science,
  vol.~8737, pp. 334--345. Springer (2014).
  \doi{10.1007/978-3-662-44777-2\_28},
  \url{https://doi.org/10.1007/978-3-662-44777-2\_28}

\bibitem{DvorakT13}
Dvo\v{r}{\'{a}}k, Z., T\r{u}ma, V.: A dynamic data structure for counting
  subgraphs in sparse graphs. In: 13th International Symposium on Algorithms
  and Data Structures, {WADS} 2013. Lecture Notes in Computer Science,
  vol.~8037, pp. 304--315. Springer (2013).
  \doi{10.1007/978-3-642-40104-6\_27},
  \url{https://doi.org/10.1007/978-3-642-40104-6\_27}

\bibitem{Erds1965OnSO}
Erd{\"o}s, P., Moon, J.W.: On sets of consistent arcs in a tournament. Canadian
  Mathematical Bulletin  \textbf{8},  269 -- 271 (1965),
  \url{https://api.semanticscholar.org/CorpusID:19010097}

\bibitem{feige2009faster}
Feige, U.: Faster fast(feedback arc set in tournaments) (2009)

\bibitem{FERNANDEZDELAVEGA1983328}
{Fernandez de la Vega}, W.: On the maximum cardinality of a consistent set of
  arcs in a random tournament. Journal of Combinatorial Theory, Series B
  \textbf{35}(3),  328--332 (1983).
  \doi{https://doi.org/10.1016/0095-8956(83)90060-6},
  \url{https://www.sciencedirect.com/science/article/pii/0095895683900606}

\bibitem{article}
Festa, P., Pardalos, P., Resende, M.: Feedback set problems. Encyclopedia of
  Optimization  \textbf{2} (06 1999). \doi{10.1007/978-0-387-74759-0_178}

\bibitem{FominFAST}
Fomin, F.V., Pilipczuk, M.: On width measures and topological problems on
  semi-complete digraphs. Journal of Combinatorial Theory, Series B
  \textbf{138},  78--165 (2019).
  \doi{https://doi.org/10.1016/j.jctb.2019.01.006},
  \url{https://www.sciencedirect.com/science/article/pii/S0095895619300073}

\bibitem{10.1007/978-3-642-15775-2_23}
Gaspers, S., Mnich, M.: Feedback vertex sets in tournaments. In: de~Berg, M.,
  Meyer, U. (eds.) Algorithms -- ESA 2010. pp. 267--277. Springer Berlin
  Heidelberg, Berlin, Heidelberg (2010)

\bibitem{Holm}
Holm, J., de~Lichtenberg, K., Thorup, M.: Poly-logarithmic deterministic
  fully-dynamic algorithms for connectivity, minimum spanning tree, 2-edge, and
  biconnectivity. J. ACM  \textbf{48}(4),  723–760 (jul 2001).
  \doi{10.1145/502090.502095}, \url{https://doi.org/10.1145/502090.502095}

\bibitem{CyclesTrianglesTournamentPhD}
Hüffner, F.: Algorithms and experiments for parameterized approaches to hard
  graph problems. Ph.D. thesis, Institut für Informatik,
  Friedrich–SchillerUniversität Jena (2007)

\bibitem{IwataO14}
Iwata, Y., Oka, K.: Fast dynamic graph algorithms for parameterized problems.
  In: 14th Scandinavian Symposium and Workshops on Algorithm Theory, {SWAT}
  2014. Lecture Notes in Computer Science, vol.~8503, pp. 241--252. Springer
  (2014). \doi{10.1007/978-3-319-08404-6\_21},
  \url{https://doi.org/10.1007/978-3-319-08404-6\_21}

\bibitem{10.1145/1250790.1250806}
Kenyon-Mathieu, C., Schudy, W.: How to rank with few errors. In: Proceedings of
  the Thirty-Ninth Annual ACM Symposium on Theory of Computing. p. 95–103.
  STOC '07, Association for Computing Machinery, New York, NY, USA (2007).
  \doi{10.1145/1250790.1250806}, \url{https://doi.org/10.1145/1250790.1250806}

\bibitem{korhonen2023dynamic}
Korhonen, T., Majewski, K., Nadara, W., Pilipczuk, M., Sokołowski, M.: Dynamic
  treewidth (2023)

\bibitem{krithika2018parameterized}
Krithika, R., Sahu, A., Saurabh, S., Zehavi, M.: The parameterized complexity
  of packing arc-disjoint cycles in tournaments (2018)

\bibitem{kumar_et_al:LIPIcs.STACS.2016.49}
Kumar, M., Lokshtanov, D.: {Faster Exact and Parameterized Algorithm for
  Feedback Vertex Set in Tournaments}. In: Ollinger, N., Vollmer, H. (eds.)
  33rd Symposium on Theoretical Aspects of Computer Science (STACS 2016).
  Leibniz International Proceedings in Informatics (LIPIcs), vol.~47, pp.
  49:1--49:13. Schloss Dagstuhl -- Leibniz-Zentrum f{\"u}r Informatik,
  Dagstuhl, Germany (2016). \doi{10.4230/LIPIcs.STACS.2016.49},
  \url{https://drops.dagstuhl.de/entities/document/10.4230/LIPIcs.STACS.2016.49}

\bibitem{10.1145/3446969}
Lokshtanov, D., Misra, P., Mukherjee, J., Panolan, F., Philip, G., Saurabh, S.:
  2-approximating feedback vertex set in tournaments. ACM Trans. Algorithms
  \textbf{17}(2) (apr 2021). \doi{10.1145/3446969},
  \url{https://doi.org/10.1145/3446969}

\bibitem{majewski_et_al:LIPIcs.STACS.2023.46}
Majewski, K., Pilipczuk, M., Soko{\l}owski, M.: {Maintaining CMSO2 Properties
  on Dynamic Structures with Bounded Feedback Vertex Number}. In: Berenbrink,
  P., Bouyer, P., Dawar, A., Kant\'{e}, M.M. (eds.) 40th International
  Symposium on Theoretical Aspects of Computer Science (STACS 2023). Leibniz
  International Proceedings in Informatics (LIPIcs), vol.~254, pp. 46:1--46:13.
  Schloss Dagstuhl -- Leibniz-Zentrum f{\"u}r Informatik, Dagstuhl, Germany
  (2023). \doi{10.4230/LIPIcs.STACS.2023.46},
  \url{https://drops.dagstuhl.de/opus/volltexte/2023/17698}

\bibitem{mnich_et_al:LIPIcs:2016:6409}
Mnich, M., Williams, V.V., V{\'e}gh, L.A.: {A 7/3-Approximation for Feedback
  Vertex Sets in Tournaments}. In: Sankowski, P., Zaroliagis, C. (eds.) 24th
  Annual European Symposium on Algorithms (ESA 2016). Leibniz International
  Proceedings in Informatics (LIPIcs), vol.~57, pp. 67:1--67:14. Schloss
  Dagstuhl--Leibniz-Zentrum fuer Informatik, Dagstuhl, Germany (2016).
  \doi{10.4230/LIPIcs.ESA.2016.67},
  \url{http://drops.dagstuhl.de/opus/volltexte/2016/6409}

\bibitem{olkowski_et_al:LIPIcs.STACS.2023.50}
Olkowski, J., Pilipczuk, M., Rychlicki, M., W\k{e}grzycki, K., Zych-Pawlewicz,
  A.: {Dynamic Data Structures for Parameterized String Problems}. In:
  Berenbrink, P., Bouyer, P., Dawar, A., Kant\'{e}, M.M. (eds.) 40th
  International Symposium on Theoretical Aspects of Computer Science (STACS
  2023). Leibniz International Proceedings in Informatics (LIPIcs), vol.~254,
  pp. 50:1--50:22. Schloss Dagstuhl -- Leibniz-Zentrum f{\"u}r Informatik,
  Dagstuhl, Germany (2023). \doi{10.4230/LIPIcs.STACS.2023.50},
  \url{https://drops.dagstuhl.de/opus/volltexte/2023/17702}

\bibitem{FASTBranching}
Raman, V., Saurabh, S.: Parameterized algorithms for feedback set problems and
  their duals in tournaments. Theoretical Computer Science  \textbf{351}(3),
  446--458 (2006). \doi{https://doi.org/10.1016/j.tcs.2005.10.010},
  \url{https://www.sciencedirect.com/science/article/pii/S0304397505006365},
  parameterized and Exact Computation

\bibitem{REID1970225}
Reid, K., Parker, E.: Disproof of a conjecture of erdös and moser on
  tournaments. Journal of Combinatorial Theory  \textbf{9}(3),  225--238
  (1970). \doi{https://doi.org/10.1016/S0021-9800(70)80061-8},
  \url{https://www.sciencedirect.com/science/article/pii/S0021980070800618}

\bibitem{10.5555/88208.88256}
Speckenmeyer, E.: On feedback problems in digraphs. In: Proceedings of the
  Fifteenth International Workshop on Graph-Theoretic Concepts in Computer
  Science. p. 218–231. WG '89, Springer-Verlag, Berlin, Heidelberg (1990)

\bibitem{10.5555/1283383.1283426}
van Zuylen, A., Hegde, R., Jain, K., Williamson, D.P.: Deterministic pivoting
  algorithms for constrained ranking and clustering problems. In: Proceedings
  of the Eighteenth Annual ACM-SIAM Symposium on Discrete Algorithms. p.
  405–414. SODA '07, Society for Industrial and Applied Mathematics, USA
  (2007)

\end{thebibliography}
\newpage
\appendix

\section{Dynamic data structures for finding triangles}\label{r:triangles}

In this section we present in detail our two data structures
for finding a triangle in a dynamic tournament. Both data
structures allow updating a dynamic tournament $T$
by
arc reversal and querying for a triangle in $T$. The running times
of
the data structures depend mainly on the maximum number of
arc-disjoint
triangles in $T$ denoted by $\ADT(T)$. The data structures do not
require any access to the actual value of $\ADT(T)$, however they
are only efficient when this parameter is small. We offer two
alternative data structures: one, where the running time of both
update and query depends purely
on $\ADT(T)$ and one, where only the query depends on
$\ADT(T)$, at the additional cost of introducing factors
poly-logarithmic in the size of the tournament. We believe that
these data structures are of independent interest, as finding
triangles in graphs is a building block for many more complex
algorithms.
The main results of this section are stated in the following two
theorems.

\begin{theorem}[Theorem~\ref{thm:DsTrianglePromiseOverview} in the overview]\label{thm:DsTrianglePromise}
For any integer $n \in \mathbb{N}$ there exists a data structure $\dspTriangle{n}$, that maintains a dynamically changing tournament $T$ on $n$ vertices\footnote{\label{note1}The data structure assumes that the vertices of $T$ are indexed by numbers in $[n]$.}
by
supporting the following operations:
\begin{enumerate}
    \item $\INITNP{T}{n}$ -- initializes the data structure with a given tournament $T$ on $n$ vertices, in time $\mathcal{O}(n^2)$
    \item $\REVERSE{v}{u}$ -- reverses an arc between vertices $v$ and $u$ in $T$, in time $\mathcal{O}(\sqrt{\ADT(T)})$
    \item $\FTRIANGLE{}$ -- returns a triangle in $T$ or reports that there are no triangles, in time $\mathcal{O}(\ADT(T) \sqrt{\ADT(T)})$
\end{enumerate}
\end{theorem}

\begin{theorem}[Theorem~\ref{thm:DsTriangleFullOverview} in the overview]\label{thm:DsTriangleFull}
For any integer $n \in \mathbb{N}$ there exists a data structure $\dsTriangle{n}$, that maintains a dynamically changing tournament $T$ on $n$ vertices\footnoteref{note1} 
by
supporting the following operations:
\begin{enumerate}
    \item $\INITNP{T}{n}$ -- initializes the data structure with a given tournament $T$ on $n$ vertices, in time $\mathcal{O}(n^2)$
    \item $\REVERSE{v}{u}$ -- reverses an arc between vertices $v$ and $u$ in $T$, in time $\mathcal{O}(\log^2 n)$
    \item $\FTRIANGLE{}$ -- returns a triangle in $T$ or reports that there are no triangles, in time $\mathcal{O}(\ADT(T) \log^2 n)$
\end{enumerate}
\end{theorem}

In Subsection~\ref{sec:prelimtools} we introduce preliminary
combinatorial properties of tournaments and we
conclude with a description of a preliminary data structure that we
extend in the follow up sections.
In Subsection~\ref{sec:triaQ} we introduce the core idea
of how to handle the queries which is used for both data structures.
Then, in Subsection~\ref{sec:DStrianglepromise} we describe the
data structure of Theorem~\ref{thm:DsTrianglePromise} and in
Subsection~\ref{sec:DStriangleFull} we describe the data structure
of Theorem~\ref{thm:DsTriangleFull}.

\subsection{Preliminary tools}\label{sec:prelimtools}

A tournament $T=(V,E)$ is a directed graph where each
pair of distinct vertices in $V$ is connected by a directed edge in $E$. We assume that there are no loops and no multi-edges. We
often denote the set of vertices of $T$ as $V(T)$ and the
set of arcs of $T$ as $E(T)$.  By in-neighbours of $v \in V(T)$ we mean the vertices $w \in V(T)$ such that $wv \in E(T)$. An out-neighbour of $v \in V(T)$ is a vertex $w \in V(T)$ such that $vw \in E(T)$.
A triangle in $T$ is a directed cycle on
three vertices. We start this section with the following well known fact.

\begin{fact}[see~\cite{CyclesTrianglesTournamentPhD} for the proof]\label{cycles_are_triangles}
    A tournament is acyclic if and only if it contains no triangles.
\end{fact}

Another well known criterion for a tournament $T$ to be acyclic is
related to the in-degrees of vertices in $T$. For a vertex
$v \in V(T)$ we denote its in-degree (i.e., the number of its in-neighbours) as  $d_T(v)$. We also denote
the minimum in-degree in a tournament $T$ as
$\mindt{T}=\min_{v \in V(T)} d_T(v)$. It is natural to partition
the vertices according to their in-degrees: we define such
partition below. We often refer to the corresponding equivalence
classes as \emph{degree buckets}. For a natural number
$n \in \mathbb{N}$ we use the notation
$[n] = \{0, 1, \ldots, n - 1\}$.
For clarity, we define $\max(\emptyset) = \min(\emptyset) = 0$
(maximum and minimum over an empty set is zero), thus for an empty
tournament $T$ we get $\FAST(T)=\ADT(T)=\mindt{T}=0$.

\begin{definition}[Degree buckets]
    For an $n$-vertex tournament $T$ and $d \in [n]$ we define the $d$'th degree bucket as $\db{T}[d] = \{v \in V(T): d_T(v) = d\}$.
    We denote the maximum size of a degree bucket of $T$ as
    $\maxdb{T}= \max_{d \in [n]} |\db{T}[d]|$.
\end{definition}

The following fact is folklore in the literature
(see~\cite{feige2009faster} for the proof).

\begin{fact}\label{transitive_buckets}
    An $n$-vertex tournament $T$ is acyclic if and only
    if $|\db{T}[d]|=1$ for all $d \in [n]$.
\end{fact}

As we shall soon see in Lemma~\ref{back_edge_triangle},
the arcs
directed in the direction opposite to the in-degree order on the
set of vertices are the reason why there are triangles in the
tournament.
We refer to these arcs
as \emph{back arcs}, whereas all the other arcs are referred to as
\emph{forward}.

\begin{definition}[Back arc]\label{back_edge_def}
    Let $T$ be a tournament. An arc $uv \in E(T)$ is a
    \emph{back arc} if $d_T(u) \geq d_T(v)$. Otherwise we call $uv$
    a \emph{forward} arc.
\end{definition}

Further, for a tournament $T$ and a set $X \subseteq V(T)$,
by $\induced{T}{X}$ we denote a subgraph of $T$ induced by $X$,
i.e., a subgraph of $T$ containing the vertices of $X$ and the
arcs connecting them.
Our algorithm for finding a triangle in a dynamically changing
tournament is based on the following lemma.

\begin{lemma}\label{back_edge_triangle}
    Let $T$ be a tournament and $uv \in E(T)$ be a back arc.
    Then $uv$ belongs to some triangle in $T$.
\end{lemma}

\begin{proof}
    We prove Lemma~\ref{back_edge_triangle} by induction on
    $|V(T)|$. In the base case when $|V(T)|=2$ the claim trivially
    holds as the only arc is not a back arc.
    Let us move on to the induction step and assume $|V(T)|=n$.
    Note that $d_T(u) > 0$ because $1 \leq d_T(v) \leq d_T(u)$.
    Let $w$ be an in-neighbour of $u$. If $vw \in E(T)$, then
    $\induced{T}{\{w, u, v\}}$ is a triangle. Otherwise, consider
    tournament $T'$ obtained from $T$ by removing $w$. Note that
    $|V(T')|=n-1$. Also note that $d_{T'}(v) = d_T(v) - 1$
    and $d_{T'}(u) = d_T(u) - 1$, so $uv$ is still a back arc.
    Thus, using the inductive hypothesis, there is a triangle
    containing $uv$ in $T'$. Clearly, any triangle in $T'$ is also
    a triangle in $T$.
\end{proof}

\begin{definition}[Prefix of a tournament]\label{prefix_def}
    Let $T$ be an $n$-vertex tournament for $n > 0$. The length $p$ of
    the prefix of $T$ is defined as follows
    \begin{equation*}
    p =
    \begin{cases}
      n & \text{if } \db{T}[d] = 1 \text{ for all } d \in [n]\\
      \min_{d \in [n]} {\db{T}[d] \neq 1} & \text{otherwise.}
    \end{cases}
    \end{equation*}
    The set of vertices $P = \bigcup_{d = 0}^{p - 1} \db{T}[d]$ is
    called \emph{the prefix} of $T$.
    Note that $p=|P|$. Moreover, we denote $T$
    without its prefix by
    $\nopref{T}=\induced{T}{V(T) \setminus P}$.
\end{definition}

\begin{lemma}\label{prefix_transitive}
    Let $T$ be a non-empty tournament on $n$ vertices, $P$ be its
    prefix. Then all arcs connecting
    $P$ and $V \setminus P$ are forward in $T$ (i.e., directed from
    $P$ to $V \setminus P$). Additionally
    $\induced{T}{P}$ has no back arcs. Moreover,
    if $P \neq V$ then
    $\db{T}[|P|] = \emptyset$.
\end{lemma}

\begin{proof}
    For the first two statements, we show that any
    back arc $uv \in E(T)$ satisfies $u,v \notin P$. So let
    us fix a back arc $uv \in E(T)$ such that $d_T(v)$ is minimum
    over all back arcs. We will show that $d_T(v) > |P|$.
    Indeed, if $v \in P$, then there is $d_T(v)$
    vertices of
    $P$ with in-degree smaller than $d_T(v)$, which all
    contribute to the in-degree of $v$. In addition to that vertex
    $w$, with
    $d_T(w) \geq d_T(v)$, also
    contributes to the in-degree of $v$.
    So there is at least $d_T(v)+1$ vertices contributing to
    the indegree of $v$, what gives a contradiction with the
    assumption that $v \in P$.

    To prove that $\db{T}[|P|] = \emptyset$ if $|P| < n$,
    let us assume for contradiction that $|\db{T}[|P|]| > 1$ and
    let
    $x,y \in \db{T}[|P|]$, $x \neq y$. Then, $d_T(x)=d_T(y)=|P|$.
    However, all vertices of $P$ contribute to indegrees of
    both $x$ and $y$, based on the first claim proved above.
    Moreover, either $x$ contributes to the indegree of $y$, or
    $y$ contributes to the indegree of $x$. Thus, one of
    the vertices $x,y$ has indegree strictly greater than $|P|$.

\end{proof}

The following corollary is a direct implication of
Lemma~\ref{prefix_transitive}.

\begin{corollary}\label{t_prime_no_0}
    Let $T$ be a tournament. Then in the tournament $\nopref{T}$
    there are no vertices of in-degree $0$.
\end{corollary}


%

The remainder of Subsection~\ref{sec:prelimtools} is devoted to
bounding the tournaments characteristics in terms of
parameter $\ADT(T)$.

\begin{lemma}[Small buckets]\label{small_buckets}
    For a tournament $T$ is holds that\\ $\maxdb{T} \leq 8(\sqrt{\ADT(T) + 1} + 1)$.
\end{lemma}

\begin{proof}
    For an empty tournament $T$ the claim holds trivially. So
    let us assume that $T$ is non-empty, let us fix
    any $d \in [n]$ and let $  $.

    If $n = 0$ then the claim holds as $\maxdb{T}=0$.
    So let us assume $n>0$, fix $d \in [n]$ and let $k=\ADT(T)$.
    For the sake of contradiction, let us assume that
    $|\db{T}[d]| > 8(\sqrt{k + 1} + 1)$.
    By Fact~\ref{arc_disjoint_triangles_and_fast}, there exists a
    feedback arc set $F$ of $T$, such that $|F| \leq 6k + 6$.
    Let $\overline{T}$ be a tournament $T$ with arcs of $F$
    reversed and let $a = \lfloor\frac{|\db{T}[d]| - 1}{2}\rfloor$.
    Note, that every reversal of an arc changes in-degrees of
    exactly two vertices by one. Thus, by
    Fact~\ref{transitive_buckets}:
    \[12k + 12 \geq 2|F| \geq \sum_{v \in \db{T}[d]} |d_{\overline{T}}(v) - d_T(v)| \geq 2 \sum_{j = 1}^{a} j = a(a + 1) \geq 16k + 16\]
    contradiction.
\end{proof}

\begin{definition}[Empty set of a tournament]\label{empty_promise}
    Let $T$ be an $n$ vertex tournament.
    A set $\emptydb=\{d \in [n] : \db{T}[d] = \emptyset\}$ is
    called the \emph{empty set} of $T$.
\end{definition}

\begin{lemma}[Small empty set]\label{empty_small}
    Let $T$ be a tournament and $\emptydb$ be the empty set of $T$. Then $|\emptydb| \leq 12(\ADT(T) + 1)$.
\end{lemma}

\begin{proof}
    By Fact~\ref{arc_disjoint_triangles_and_fast}, there exists a feedback arc set $F$ of $T$ of size at most $6(\ADT(T) + 1)$.
    Every reversal of an arc changes the degrees of exactly $2$ vertices. Thus, by reversing one arc we can decrease $|\emptydb|$ by at most $2$.
    After reversing all arcs from $F$ we arrive at an acyclic tournament. Thus, by Fact~\ref{transitive_buckets}, the size of $\emptydb$ must be $0$ after such modifications.
\end{proof}

\begin{fact}\label{minimal_degree}
    Let $T$ be a tournament and $\emptydb$ be the empty set of $T$. Then $\mindt{T} \leq |\emptydb|$ (recall that $\mindt{T}$ is the minimum in-degree in $T$).
\end{fact}

\begin{proof}
    Clearly $|\db{T}[d]| = 0$ for $d \in [\mindt{T}]$. Thus, $\mindt{T} \leq |\mathtt{Empty}|$.
\end{proof}

The following corollary is a direct consequence of Fact~\ref{minimal_degree}, Lemma~\ref{empty_small} and Lemma~\ref{prefix_transitive}.

\begin{corollary}\label{min_deg_Tprim}
Let $T$ be a tournament and let $P$ be the prefix of $T$. Then it
holds that
$\mindt{\nopref{T}} \leq 12(\ADT(T)+1)$
(recall that $\nopref{T}=\induced{T}{V(T) \setminus P}$).
\end{corollary}

In what follows it is useful to define a set of vertices that can be in some way affected by reversing an arc $uv$.

\begin{definition}[Affected buckets]\label{affected_buckets_def}
Let $T$ be a tournament and $u,v \in V(T)$. We define the union of affected buckets as
\begin{flalign*}
\deltadb{T}{u}{v} & =\db{T}[d_T(u)-1] \cup \db{T}[d_T(u)] \cup \db{T}[d_T(u)+1] &\\
     & \; \; \cup \db{T}[d_T(v)-1] \cup \db{T}[d_T(v)] \cup \db{T}[d_T(v)+1].&
\end{flalign*}
In the above union we consider the non-existent degree buckets as empty.
\end{definition}

We now bound the change in the set of back arcs caused by reversing arc $uv$.

\begin{lemma}[Back arcs change]\label{back_arcs_changes_nok}
    Let $T$ be a tournament, 
    $uv \in E(T)$ and $\backdb$ be the set of all back arcs of $T$. Let $\overline{T}$ be a tournament $T$ with arc $uv$ reversed and $\overline{\backdb}$ be
    the set of all back arcs of $\overline{T}$.
    Then, for all $xy \in \backdb$ $\triangle$ $\overline{\backdb}$
    one of the two following conditions hold:
    \begin{enumerate}
     \item $x \in \{ u,v \}$ and $y \in \deltadb{T}{u}{v}$
     \item $y \in \{ u,v \}$ and $x \in \deltadb{T}{u}{v}$.
    \end{enumerate}
\end{lemma}

\begin{proof}
    The only vertices which in-degree changes by the $uv$ arc reversal are $u$ and $v$. Moreover, their in-degree
    can change at most by $1$. Thus, any other vertex $w \in V(T)$ that is an endpoint of an arc in $\backdb$ $\triangle$ $\overline{\backdb}$ satisfies
    $|d_T(w)-d_T(v)| \leq 1$ or $|d_T(w)-d_T(u)| \leq 1$. The claim follows.
\end{proof}

\begin{lemma}[Small back arcs change]\label{back_arcs_changes}
    Let $T$ be a tournament, $uv \in E(T)$ and $\backdb$ be the set of all back arcs of $T$. Let $\overline{T}$ be a tournament $T$ with arc $uv$ reversed and $\overline{\backdb}$ be the set of all back arcs of $\overline{T}$.
    Then $|\backdb$ $\triangle$ $\overline{\backdb}| \leq 48(\sqrt{\ADT(T) + 1} + 1)$.
\end{lemma}

\begin{proof}
    By the arguments of Lemma~\ref{back_arcs_changes_nok}, $|\backdb$ $\triangle$ $\overline{\backdb}| \leq |\deltadb{T}{u}{v}| - 2 + 1 \leq |\deltadb{T}{u}{v}|$.
    By Lemma~\ref{small_buckets}, this yields at most $6 \cdot 8(\sqrt{\ADT(T) + 1} + 1) = 48(\sqrt{\ADT(T) + 1} + 1)$ arcs.
\end{proof}

\begin{lemma}[Few back arcs in total]\label{few_back_arcs}
    Let $T$ be a tournament and $\backdb$ be the set of all back arcs of $T$. Then $|\backdb| \leq 288(\ADT(T) + 1)(\sqrt{6\ADT(T) + 7} + 1)$.
\end{lemma}

\begin{proof}
    Let $F$ be a minimum feedback arc set of $T$.
    We prove by induction on $|F|$, that $|\backdb| \leq 48 \sum_{j = 2}^{|F| + 1} (\sqrt{j} + 1)$.
    In the base case $|F| = 0$ and $T$ is an acyclic tournament, so by Lemma~\ref{back_edge_triangle} we get $|\backdb| = 0$.
    In the inductive step, we assume $|F| > 0$. Let $e \in F$.
    Let us reverse $e$ obtaining tournament $\overline{T}$ and let $\overline{F}$ be a minimum feedback arc set of $\overline{T}$ and $\overline{\backdb}$ be the set of all back arcs in $\overline{T}$.
    Note, that $|\overline{F}| = |F| - 1$, because $F$ is a minimum feedback arc set of $T$.
    By Lemma~\ref{back_arcs_changes} and Fact~\ref{arc_disjoint_triangles_and_fast}, $|\backdb$ $\triangle$ $\overline{\backdb}| \leq 48(\sqrt{|F| + 1} + 1)$.
    Thus, using the inductive hypothesis,
    $|\backdb| \leq 48(\sqrt{|F| + 1} + 1) + |\overline{\backdb}| \leq 48(\sqrt{|F| + 1} + 1) + 48 \sum_{j = 2}^{|\overline{F}| + 1} (\sqrt{j} + 1) = 48\sum_{j = 2}^{|F| + 1} (\sqrt{j} + 1)$.
    Hence, $|\backdb| \leq 48|F|(\sqrt{|F| + 1} + 1)$.
    By Fact~\ref{arc_disjoint_triangles_and_fast}, $|F| \leq 6(\ADT(T) + 1)$. Thus, $|\backdb| \leq 288(\ADT(T) + 1)(\sqrt{6\ADT(T) + 7} + 1)$.
\end{proof}

We conclude this subsection with the description of a basic data
structure maintaining information about a dynamic tournament. The
data structures presented further are extensions of this basic data
structure. Before we proceed with this, we need to introduce an
even simpler structure that maintains a set $S \subseteq [m]$ for
some natural number $m$ and is referred to as an \emph{array set}
over $[m]$.

\begin{definition}[Array set over {$[m]$}]
    Let $m \in \mathbb{N}$. We call a data structure an \emph{array set over} $[m]$ if it implements a set $S \subseteq [m]$ in the following way:
    \begin{enumerate}
        \item A doubly linked list containing all elements of $S$
        \item An array storing for each $i \in [m]$ a pointer to the list node representing $i$ or null if $i \notin S$
    \end{enumerate}
\end{definition}

Such data structure takes $\mathcal{O}(m)$ space, requires
$\mathcal{O}(m)$ initialization time and supports add, delete, find
operations in $\mathcal{O}(1)$ time and iteration over $S$ in
$\mathcal{O}(|S|)$ time.

We move on to the description of the basic data structure. All our
data structures assume that the vertices of the input $n$-vertex
tournament $T$ are
provided to the data structure with their unique index from $[n]$.
This assumption is reflected in the footnotes.

\begin{lemma}\label{update_impl_basic}
    For every $n \in \mathbb{N}$ there exists a data structure
    $\dsBasic{n}$ that for a dynamic tournament $T$ on $n$
    vertices\footnoteref{note1} 
    altered by arc reversals, maintains the following information:
    \begin{enumerate}
        \item Adjacency matrix of $T$
        \item An array of size $n$ of in-degrees of vertices in $V(T)$
        \item An array of size $n$ of degree buckets of $T$, where\\ each bucket $\db{T}[i], i \in [n]$ is represented by a list
        \item An array of size $n$ of pointers to the bucket lists:\\ each vertex stores the pointer to its position in the list representing its bucket
        \item Empty set $\emptydb$ of $T$ as an array set over $[n]$
        \item An array of size $n + 1$ counting the number of degree buckets of a given size, for all possible sizes
        \item Integer \maxdb{T}
    \end{enumerate}
    Data structure $\dsBasic{n}$ offers the following operations:
    \begin{itemize}
     \item $\INITNP{T}{n}$ -- initializes data structure with tournament $T$, in time $\mathcal{O}(n^2)$
     \item $\REVERSE{v}{u}$ -- reverses arc between vertices $v$ and $u$ in $T$, in time $\mathcal{O}(1)$
    \end{itemize}
\end{lemma}

\begin{proof}
    It is easy to see how to initialize data structures from points $1-7$ in time $\mathcal{O}(n^2)$ by reading all the arc of $T$ and performing a simple processing.
    Let us describe the update.
    Updating data structures from points $1.$ and $2.$ is straightforward.
    Note that reversing an arc between $v$ an $u$ changes only degrees of $v$ and $u$.
    Thus, updating degree buckets can be done in time $\mathcal{O}(1)$.
    Having this, we can easily update $\emptydb$ and the array in Point $6$ also in constant time.
    Finally, the size of any degree bucket changes by at most two, so there is a constant number of candidates for the new value of $\maxdb{T}$.
    Thus, $\maxdb{T}$ can be updated in $\mathcal{O}(1)$ time using the array of Point $6$.
\end{proof}

\subsection{Triangle Query Implementation}
\label{sec:triaQ}

In this section we present techniques used to perform query (find
a triangle) that are common for both structures: the one of
Theorem~\ref{thm:DsTrianglePromise} and the one of
Theorem~\ref{thm:DsTriangleFull}.
They extend the structure of Lemma~\ref{update_impl_basic} by adding
the functionality of finding a triangle. The procedures provided in
this subsection, similarly as the data structure of
Lemma~\ref{update_impl_basic}, operate on an $n$-vertex tournament
$T$ and they assume access to the indexing of $V(T)$ with the numbers
in $[n]$.

We implement the functionality of finding a triangle using two crucial subprocedures (working on a tournament $T$):
\begin{enumerate}
    \item $\MIN$ -- returns a vertex $v$ of minimum in-degree in $\nopref{T}$, together with $d_{\nopref{T}}(v)$ (if $\nopref{T} = \emptyset$ then it returns $null$)

    \item $\INCOMING{v}{l}$ -- returns $l \in \mathbb{N}$ in-neighbours of $v \in V(\nopref{T})$ in $\nopref{T}$ (or as many as there exist)
\end{enumerate}

The main idea is to find a vertex $v \in V(\nopref{T})$ with
in-degree $d_{\nopref{T}}(v)$ bounded by a function of $\ADT(T)$.
Due to Corollary~\ref{min_deg_Tprim} this can be achieved with
$\MIN$ subprocedure.
After we find $v$, using $$\INCOMING{v}{1}$$ invocation we find a
vertex $u$ which is an in-neighbour of $v$. Using
$$\INCOMING{u}{d_{\nopref{T}}(v)}$$ invocation, we iterate over
in-neighbours of $u$ in $\nopref{T}$.
After inspecting $d_{\nopref{T}}(v)$ in-neighbours of $u$ we are
bound to find a triangle by a similar argument as in the proof of
Lemma~\ref{back_edge_triangle}.
The pseudocode for the $\FTRIANGLE{}$ subroutine is presented in Algorithm~\ref{fig:common_query_triangles} and details are covered by Theorem~\ref{common_query_triangles}.
One needs to keep in mind, that the procedure $\FTRIANGLE{}$, similarly as the two subprocedures listed above, operates on a tournament $T$ that shall be clear from the context.
The $\FTRIANGLE{}$ method is useful in both the promise and the
full models. It returns a triangle or states that there are none.

\begin{algorithm}
    \SetKwInOut{Input}{Operation}
	\SetKwInOut{Output}{Output}
    \Input{\FTRIANGLE{}}
    \Output{A triangle in tournament $T$ or \texttt{NONE} if none exists}
    \BlankLine

    $v \leftarrow \MIN$\;
    \If{$v$ is $null$} {
         \Return \texttt{NONE}\;
     }

    $\{u\} \leftarrow \INCOMING{v}{1}$\;
    $N \leftarrow \INCOMING{u}{d_{\nopref{T}}(v)}$\;

    \For{$w \in N$}{
        \If{$vw \in E(T)$}{
            \Return wuv\;
        }
    }
    \caption{Pseudocode for the operation of finding a triangle in $T$}
    \label{fig:common_query_triangles}
\end{algorithm}

\begin{theorem}\label{common_query_triangles}
    Let $T$ be a tournament. Let us assume access to the information related to $T$ maintained by the data structure of Lemma~\ref{update_impl_basic}. Let us
    also assume that we are given the following two subprocedures:
    \begin{enumerate}
        \item $\MIN$ -- returns a vertex $v$ of minimum in-degree in $\nopref{T}$, together with $d_{\nopref{T}}(v)$ (if $\nopref{T} = \emptyset$ then it returns $null$),
        in time $t_1(T)$
        \item $\INCOMING{v}{l}$ -- returns $l \in \mathbb{N}$ in-neighbours of $v \in V(\nopref{T})$ in $\nopref{T}$ (or as many as there exist), in time $t_2(l,T)$ for a non-decreasing (with respect to $l$) function $t_2$
    \end{enumerate}
    Then it is possible to implement method $\FTRIANGLE{}$ that returns a triangle in $T$ or states that there are no triangles, in time $\mathcal{O}(t_1(T) + t_2(\mindt{\nopref{T}}, T) + \mindt{\nopref{T}})$.
\end{theorem}

\begin{proof}
    By Lemma~\ref{prefix_transitive}, every triangle in $T$ is entirely contained inside $\nopref{T}$.
    By Observation~\ref{t_prime_no_0} and Fact~\ref{transitive_buckets}, if $\nopref{T}$ is not empty then there exists a triangle in $\nopref{T}$.
    Consider the procedure $\FTRIANGLE{}$ whose pseudocode is given in Figure~\ref{fig:common_query_triangles}.
    If $\MIN$ subprocedure returns $null$ then it means that $\nopref{T}$ is empty and thus, by the arguments above, $T$ has no triangles.
    Otherwise, we obtain vertex $v$ and, by the arguments above, there is a triangle entirely contained inside $\nopref{T}$.
    The call to the subprocedure takes $t_1(T)$ time.
    Vertex $v$, by Observation~\ref{t_prime_no_0}, satisfies $0 < d_{\nopref{T}}(v) = \mindt{\nopref{T}}$. Thus, subprocedure $\INCOMING{v}{1}$ must return an in-neighbour $u$ of $v$. This takes $t_2(1, T)$ time. Note that $uv$ is a back arc in $\nopref{T}$, since $d_{\nopref{T}}(v) = \mindt{\nopref{T}}$.
    Then, $\INCOMING{u}{d_{\nopref{T}}(v)}$ returns a list $N$ of in-neighbours of $u$ such that $|N| = \mindt{\nopref{T}}$, since $d_{\nopref{T}}(v) \leq d_{\nopref{T}}(u)$. This takes $t_2(\mindt{\nopref{T}},T)$ time.
    As a final step, we iterate over $N$ and check which vertex (possibly many) from it forms a triangle together with $v$ and $u$. Such a vertex must exist by the same argument as in the proof of Lemma~\ref{back_edge_triangle}. This last step takes $\mathcal{O}(\mindt{\nopref{T}})$ time. The assumption that $t_2$ is non-decreasing gives the total upper bound on the running time of the entire procedure.
\end{proof}

Corollary~\ref{min_deg_Tprim} immediately gives the following.

\begin{corollary}\label{common_query_triangles_corollary}
    Let $T$ be a tournament. Let us assume access to the information related to $T$ maintained by the data structure of Lemma~\ref{update_impl_basic}. Let us
    also assume that we are given the following two subprocedures:
    \begin{enumerate}
        \item $\MIN$ -- returns a vertex $v$ of minimum in-degree in $\nopref{T}$, together with $d_{\nopref{T}}(v)$; if $\nopref{T} = \emptyset$ then it returns $null$,
        in time $t_1(T)$
        \item $\INCOMING{v}{l}$ -- returns $l \in \mathbb{N}$ in-neighbours of $v \in V(\nopref{T})$ in $\nopref{T}$ (or as many as there exist), in time $t_2(l,T)$ for a non-decreasing (with respect to $l$) function $t_2$
    \end{enumerate}
    Then it is possible to implement method $\FTRIANGLE{}$ that returns a triangle from $T$ or states that there are no triangles, in time $\mathcal{O}(t_1(T) + t_2(12(\ADT(T) + 1),T)  + \ADT(T))$.
\end{corollary}

The remainder of this section is devoted to bounding the functions
$t_1$ and $t_2$ of Theorem~\ref{common_query_triangles} in terms
of $\ADT(T)$.
The subprocedure $\MIN$ is the same for both the promise and the
full model and we present it below. The subprocedure
$\INCOMING{v}{l}$ depends on the model so its implementation is
deferred to the respective subsections.

\begin{lemma}\label{empty_impl_promise}
    Let $T$ be an $n$ vertex tournament, and $\emptydb$ be the empty set of $T$. Let us assume access to the information related to $T$ maintained by the
    data structure of Lemma~\ref{update_impl_basic}.
%
    Then $\MIN$ can be implemented in time $\mathcal{O}(|\emptydb|)$.
\end{lemma}

\begin{proof}
    If $n = 0$ then return $null$.
    Else if $|\emptydb| = 0$ then return $null$, by Lemma~\ref{prefix_transitive}.
    Otherwise, let $d$ be the minimal element in $\emptydb$.
    By the definition of the prefix, every element in $\emptydb$ is not smaller than $|P|$, where $P$ is the prefix of $T$.
    Hence, by Lemma~\ref{prefix_transitive}, it is enough to find the smallest integer $d' \in [n]$ larger than $d$, but not present in $\emptydb$.
    This can be easily done by iterating over integers from $d + 1$ up to $n$ and checking if a given integer belongs to $\emptydb$ in constant time.
    Notice, that we will check at most $|\emptydb|$ integers before finding a not-empty bucket and that we will always find such $d'$, by Lemma~\ref{prefix_transitive}.
    In the end, return a vertex $v$ from $\db{T}[d']$ and $d_\nopref{T}(v)$.
    Notice that, $d_\nopref{T}(v) = d_T(v) - |P|$, by Lemma~\ref{prefix_transitive}.
    Given $\emptydb$ as an array set, we can calculate $|P|$ in time $\mathcal{O}(|\emptydb|)$ by finding a minimum element in $\emptydb$, by Lemma~\ref{prefix_transitive}.
\end{proof}

Lemma~\ref{empty_small} directly leads to the following corollary.

\begin{corollary}\label{empty_impl_promise_corollary}
    Under the assumptions of Lemma~\ref{empty_impl_promise}, $\MIN$ can be implemented in time $\mathcal{O}(\ADT(T))$.
\end{corollary}

\subsection{Triangle Data Structure in the Promise Model}\label{sec:DStrianglepromise}

The goal of this subsection is to prove
Theorem~\ref{thm:DsTrianglePromise}. In order to do so we need to
provide the implementation of $\INCOMING{v}{l}$ in the promise
model. We also need to extend the data structure $\dsBasic{n}$ of
Lemma~\ref{update_impl_basic} to the data structure
$\dspTriangle{n}$ of Theorem~\ref{thm:DsTrianglePromise}, to
maintain all information needed to implement $\INCOMING{v}{l}$ in
the promise model. The subprocedure's pseudocode is given in
Algorithm~\ref{fig:incoming_promise_triangles}.

\begin{algorithm}
    \SetKwInOut{Input}{Operation}
	\SetKwInOut{Output}{Output}
    \Input{$\INCOMING{v}{l}$}
    \Output{A list of $l$ (or as many as there exist) incoming neighbours of $v$ in $\nopref{T}$}
    \BlankLine

    $N \leftarrow \emptyset$\;

    \tcp{First group}
    \For{$e \in \backdb$}{
        \If{$e=uv$}{
            $N$ $+=$ $u$\;
        }
    }

    \If{$|N| \geq l$}{
        \Return Arbitrary $l$ elements of $N$\;
    }

    $p \leftarrow \min(\emptydb)$\;
    \tcp{Second group}
    \For{$i \gets p \text{ \bf to } d_{T}(v) - 1$}{
        \For{every $w \in \db{T}[i]$}{
            \If{$wv \in E(T)$}{
                $N$ $+=$ $w$\;
            }

            \If{$|N| \geq l$}{
                \Return Arbitrary $l$ elements of $N$\;
            }
        }
    }

    \Return $N$\;
    \caption{Pseudocode for $\INCOMING{v}{l}$ subprocedure in the promise model}
    \label{fig:incoming_promise_triangles}
\end{algorithm}

\begin{lemma}\label{incoming_impl_promise}
    Let $T$ be a tournament, $\emptydb$ be its empty set
    and $\backdb$ be the set of back arcs in $T$.
    Assume access to the information on $T$ maintained by the data structure of Lemma~\ref{update_impl_basic} and assume that $\backdb$ is given as a list.
    Then $\INCOMING{v}{l}$ can be
    implemented in time $\mathcal{O}(|\backdb| + |\emptydb| + l)$.
\end{lemma}

\begin{proof}
    The pseudocode is given in Algorithm~\ref{fig:incoming_promise_triangles}. Let us describe how to iterate over all incoming neighbours of $v$ in $\nopref{T}$. From the algorithm it will be obvious that we can stop it at any moment after finding $l$ of them.

    The search for in-neighbours in $\nopref{T}$ can be divided into two parts: vertices of not smaller and smaller in-degrees than $d_T(v)$. Notice that if a vertex $u$ belongs to the first group, the arc $uv$ is a back arc in $T$. Hence, it is enough to iterate over $\backdb$ and store found incoming neighbours of $v$. All found neighbours must belong to $\nopref{T}$, by Lemma~\ref{prefix_transitive}.

    It remains to find incoming neighbours from the second group. Let $p$ be the length of the prefix.
    Given $\emptydb$ as an array set, we can calculate $p$ in time $\mathcal{O}(|\emptydb|)$ by finding a minimum element in $\emptydb$, by Lemma~\ref{prefix_transitive}.
    Next, let us iterate from $i := p$ up to $d_{T}(v) - 1$.
    For every $w \in \db{T}[i]$, check if $wv \in E(T)$.
    If so, we found a new in-neighbour in $\nopref{T}$, because $d_T(w) \geq p$.
    If not, we found a back arc $vw \in \backdb$.
    However, this can happen only $|\backdb|$ number of times.
    Similarly, only $|\emptydb|$ number of times $\db{T}[i]$ can be empty.
    To sum up, we end up finding $l$ in-neighbours of $v$ (or as many as there exist) in $\nopref{T}$ in time $\mathcal{O}(|\backdb|+|\emptydb|+l)$.
\end{proof}

Lemma~\ref{few_back_arcs} and Lemma~\ref{empty_small} immediately lead to the following corollary.

\begin{corollary}\label{incoming_impl_promise_corollary}
    Under the same assumptions as of Lemma~\ref{incoming_impl_promise},
    the procedure\\ $\INCOMING{v}{l}$ takes
    $\mathcal{O}(\ADT(T)\sqrt{\ADT(T) + 1} + l)$ time.
\end{corollary}

This concludes the implementation of the query subprocedure.
We are now ready to prove Theorem~\ref{thm:DsTrianglePromise}.

\paragraph*{Proof of Theorem~\ref{thm:DsTrianglePromise}}
\begin{proof}

Let $\dspTriangle{n}$ be the data structure $\dsBasic{n}$ extended by:
\begin{itemize}
 \item Storing the set of all back arcs $\backdb$ as a list and a square matrix of size $n^2$ storing pointers to the list node representing a given arc of $\backdb$ or null.
 \item The method \FTRIANGLE{} of Corollary~\ref{common_query_triangles_corollary} using the implementation of\\ $\INCOMING{u}{l}$ subprocedure given by Corollary~\ref{incoming_impl_promise_corollary} and the implementation of $\MIN$ given by Corollary~\ref{empty_impl_promise_corollary}.
\end{itemize}
Theorem~\ref{thm:DsTrianglePromise} is a direct consequence of Lemma~\ref{back_arcs_changes_nok},
Lemma~\ref{small_buckets},
Corollary~\ref{common_query_triangles_corollary},
Corollary~\ref{incoming_impl_promise_corollary} and Corollary~\ref{empty_impl_promise_corollary}.
\end{proof}

In the remainder of this subsection we provide an extension of the data structure $\dspTriangle{n}$ of Theorem~\ref{thm:DsTrianglePromise}, which prevents the data structure from running the operations for too long if the promise is not satisfied. Essentially, if an operation on $\dspTriangle{n}$ is taking a long time, this means that the promise is not satisfied.
Checking whether the promise is satisfied is hard.
However, we can easily check whether the properties of the tournament which follow from the promise and which we use in the data structure are satisfied.
Thus, the extended data structure first performs some preliminary checks and only if they succeed it continues with the regular methods' implementations.
Note, that if the preliminary checks do not affect the data structure, this approach guarantees that in case of method's failure, the data structure remains intact, which is important in case of updates. We formalize this idea in the next theorem.

\begin{theorem}\label{thm:DsTrianglePromiseInternal}
For every $n \in \mathbb{N}$, there exists a data structure $\dspTrianglePlus{n}$ maintaining a dynamic tournament $T$ on $n$ vertices\footnoteref{note1}, by supporting the following operations:
\begin{enumerate}
    \item $\INITNP{T}{n}$ -- initializes the data structure with a given tournament $T$ on $n$ vertices, in $\mathcal{O}(n^2)$ time
    \item $\REVERSEBND{v}{u}{d}$ -- for $d \in \mathbb{N}$, reverses an arc between vertices $v$ and $u$ in $T$ or reports that $\ADT(T) \geq d + 1$, in time $\mathcal{O}(\sqrt{d})$
    \item $\FTRIANGLE{d}$ -- for $d \in \mathbb{N}$, returns a triangle from $T$ or reports that there are no triangles or reports that $\ADT(T) \geq d + 1$, in time $\mathcal{O}(d \sqrt{d})$
\end{enumerate}
\end{theorem}

\begin{proof}
    The data structure $\dspTrianglePlus{n}$ maintains the same information as the data structure $\dspTriangle{n}$ of Theorem~\ref{thm:DsTrianglePromise}. As the theorem states, we can initialize it in $\mathcal{O}(n^2)$ time.

    By Lemma~\ref{back_arcs_changes_nok} and
    Lemma~\ref{update_impl_basic}, the structure
    $\dspTriangle{n}$ takes $\mathcal{O}(|\deltadb{T}{u}{v}|)$ time to reverse an arc. The data structure $\dspTrianglePlus{n}$, before attempting reversal, additionally verifies if $\maxdb{T} \leq 8 (\sqrt{d + 1}+1)$.
    If this is the case then $|\deltadb{T}{u}{v}| \leq 48 (\sqrt{d + 1}+1)$ and it carries on the $\REVERSE{u}{v}$ operation as in $\dspTriangle{n}$. Otherwise, it reports that there is more than $d$ arc-disjoint triangles, which is correct by Lemma~\ref{small_buckets}.

    The data structure $\dspTrianglePlus{n}$ performs an analogous verification before attempting to find a triangle.
    If verifies whether $|\backdb| \leq 288(d + 1)(\sqrt{6d + 7} + 1)$ and whether $|\emptydb| \leq 12(d + 1)$ in constant time.
    If this is not the case, it reports that there are more than $d$ arc-disjoint triangles, which is correct by Lemma~\ref{few_back_arcs} and Lemma~\ref{empty_small}.
    Note, that $\mindt{\nopref{T}}$ is upper bounded by $|\emptydb| \leq 12(d + 1)$, by Fact~\ref{minimal_degree} applied to $\nopref{T}$ and Lemma~\ref{prefix_transitive}.
    Hence, $\dspTrianglePlus{n}$ carries on $\FTRIANGLE{}$ procedure as defined for $\dspTriangle{n}$ data structure in the required time, by Theorem~\ref{common_query_triangles}, Lemma~\ref{incoming_impl_promise} and Lemma~\ref{empty_impl_promise}.
\end{proof}

\subsection{Triangle Data Structure in the Full Model}\label{sec:DStriangleFull}

The goal of this subsection is to prove
Theorem~\ref{thm:DsTriangleFull}.
Similarly as in the previous subsection, we extend the data
structure $\dsBasic{n}$ of Lemma~\ref{update_impl_basic} to the
data structure $\dsTriangle{n}$ of Theorem~\ref{thm:DsTriangleFull},
supporting the operation $\FTRIANGLE{}$, whose pseudocode is given
in
Algorithm~\ref{fig:common_query_triangles}.
This subroutine in turn needs subprocedures
$\MIN$ and $\INCOMING{v}{l}$ to operate.
The implementation of subprocedure $\MIN$ introduced in
Lemma~\ref{empty_impl_promise} works in both models. Thus, we are
only missing the implementation of the subprocedure
$\INCOMING{v}{l}$ suitable for the full model.
The description heavily relies on two new data structures defined next.

\begin{definition}[Degrees structure]\label{degrees_def}
    Let $T$ be a non-empty tournament on $n$ vertices\footnoteref{note1}. Let $M$ be a binary $n \times n$ matrix, such that $M[d, v] = 1$ if and only if $d_T(v) = d$. Let $S$ be a segment tree spanned over $[n]$ and let $\mathcal{I}$ be the set of intervals (segments) associated with the nodes of $S$.
    A data structure is called a \emph{degrees structure} for $T$ if it supports the following operations:
    \begin{enumerate}
        \item $\ROOT$ -- returns a \emph{pointer} to the root of $S$.
        \item $\LEFT{p}$ and $\RIGHT{p}$ -- given a \emph{pointer} $p$ to a node in $S$ returns a \emph{pointer} to a child (left and right, respectively) of $p$ in $S$.
        \item $\RECT{p}{d}$ -- given a \emph{pointer} $p$ representing $[a, b] \in \mathcal{I}$ and $d \in \mathbb{Z}$ returns the sum of elements from $M$ contained in a sub-rectangle $[0, d] \times [a, b]$.
    \end{enumerate}
    and supports changing the in-degree of an arbitrary vertex from $T$.
\end{definition}

\begin{fact}\label{degrees_impl}
    Let $T$ be a non-empty tournament on $n$ vertices\footnoteref{note1}. Then there exists a degrees structure for $T$ performing $\ROOT$, $\LEFT{p}$, $\RIGHT{p}$ operations in constant time and $\RECT{p}{d}$ operation in $\mathcal{O}(\log{n})$ time. The change of a vertex in-degree is supported in $\mathcal{O}(\log^2{n})$ time. Moreover, the structure takes $\mathcal{O}(n \log{n})$ space and can be initialized in $\mathcal{O}(n \log^2{n})$ time.
\end{fact}

\begin{proof}
    The sought data structure is a standard segment tree $S$ spanned over $[n]$. In each node $p$ of $S$, corresponding to interval $I_p$, we store an AVL tree $A_p$, that stores the in-degrees of vertices in $I_p$.

    The first three operations are straightforward to implement. The last query operation, namely $\RECT{p}{d}$, is implemented easily by querying $A_p$ for the number of stored elements not greater than $d$. In order to change an in-degree of a vertex, we simply remove it from $\mathcal{O}(\log{n})$ AVL trees and then insert new value in $\mathcal{O}(\log{n})$ AVL trees in total time $\mathcal{O}(\log^2{n})$.
    To prove the space complexity bound it is enough to notice that every element (in-degree of a vertex) is stored on $\mathcal{O}(\log{n})$ AVL trees. Initialization in the required time is straightforward -- we simply insert vertices' in-degrees one by one into the structure.
\end{proof}

\begin{definition}[Adjacency tree]\label{adjacency_tree_def}
    Let $T$ be a non-empty tournament on $n$ vertices\footnoteref{note1}. An \emph{adjacency tree} $T_v$ of $v$ is a standard static segment tree spanned over $[n]$, such that the value in the $x$-th leaf is $1$ if $xv \in E(T)$ and $0$ otherwise. The tree should also support changing the leaf value in time $\mathcal{O}(\log{n})$ and maintain in each node $p$ the sum of values stored in all leaves covered by the interval corresponding to $p$.
\end{definition}
The data structure of Definition~\ref{adjacency_tree_def} is a
standard summing segment tree.
Equipped with a degrees structure and the adjacency trees, we are
ready to show how to implement the $\INCOMING{v}{l}$ subprocedure.

\begin{algorithm}
    \SetKwInOut{Input}{Algorithm}
	\SetKwInOut{Output}{Output}
    \Input{$\INCOMING{v}{l}$}
    \Output{A list of $l$ (or as many as there exist) incoming
    neighbours of $v$ in $\nopref{T}$}
    $N \leftarrow \emptyset$\;
    \For{$i \gets 1$ to $l$}{
        $x \leftarrow \mathtt{SingleNeighbour}(v)$\;
        \If{$x \neq null$}{
            $N$ $+=$ $x$\;
            $T_v.Set(x, 0)$\;
        }
    }

    \For{every $x \in N$}{
        $T_v.Set(x, 1)$\;
    }

    \Return $N$\;

    \BlankLine
    \SetKwInOut{Input}{Operation}
	\SetKwInOut{Output}{Output}
    \Input{$\mathtt{SingleNeighbour}(v)$}
    \Output{A single incoming neighbour of $v$ in $\nopref{T}$,
    if it exists}

    \tcp{Let $t_{[a, b]}$ be the number of in-neighbours of $v$ in $\nopref{T}$ from interval $[a, b]$. Note, that $t_{[a, b]} = s_{[a, b]} - p_{[a, b]}$, where $s_{[a, b]}, p_{[a, b]}$ are defined in the proof of Lemma~\ref{incoming_full_impl}.}

    $ptr \leftarrow \degrees.\ROOT$\;
    $pos \leftarrow [0, n - 1]$\;

    \tcp{BST-like search}
    \While{$|pos| \neq 1$}{
        $l \leftarrow t_{pos.left}$\;
        $r \leftarrow t_{pos.right}$\;

        \If{$l > 0$}{
            $(pos, ptr) \leftarrow (pos.left, \degrees.\LEFT{ptr})$\;
        }

        \ElseIf{$r > 0$}{
            $(pos, ptr) \leftarrow (pos.right, \degrees.\RIGHT{ptr})$\;
        }

        \Else{
            \Return $null$\;
        }
    }

    \If{$t_{pos} > 0$}{
        \Return $pos.first$ \tcp*[r]{pos is a one-element interval}
    }
    \Return $null$\;
    \caption{Pseudocode for $\INCOMING{v}{l}$ in the full model}
    \label{fig:incoming_full_triangles}
\end{algorithm}

\begin{lemma}\label{incoming_full_impl}
    Let $T$ be a non-empty tournament on $n$ vertices\footnoteref{note1}. Let $\emptydb$ be the empty set of $T$. Assume the access to the following data structures:
    \begin{enumerate}
        \item The degrees structure $\degrees$ for $T$ from Fact~\ref{degrees_impl}
        \item $\emptydb$ as an array set over $[n]$
        \item Adjacency trees $\{T_u\}_{\{u \in V(T)\}}$
    \end{enumerate}
    Then subprocedure $\INCOMING{v}{l}$ can be implemented in
    time $\mathcal{O}(l \cdot \log^2{n} + |\emptydb|)$.
\end{lemma}

\begin{proof}
    Given $\emptydb$ as an array set, we compute the length $p$ of the prefix $P$ of $T$ in time $\mathcal{O}(|\emptydb|)$ by finding a minimum element in $\emptydb$, by Lemma~\ref{prefix_transitive}.
    This allows us to use the $\degrees$ data structure
    to find out, for any interval $[a,b]$ represented by a single node in the segment
    tree $T_v$, the value $p_{[a,b]} = P \cap [a,b]$, which is the number of vertices of $P$ in the interval $[a,b]$.
    It is done by a single invocation of method $\degrees.\RECT{\cdot}{\cdot}$ in $\mathcal{O}(\log{n})$ time.
    Let us denote by $s_{[a,b]}$ the number of in-neighbours in $T$ of $v$ within the interval $[a,b]$.
    If $[a, b]$ is represented by a single node in $T_v$ then $s_{[a,b]}$ can be found just
    by querying such a node.

    In order to find a single in-neighbour of $v$ in $\nopref{T}$ (lines $10$-$23$ of Algorithm~\ref{fig:incoming_full_triangles}), we follow the path from
    the root of $T_v$ to a leaf. At every tree node that
    we visit, we enter the child node
    corresponding to $[a,b]$ if $s_{[a,b]}-p_{[a,b]} > 0$.
    Since all the vertices of $P$ are in-neighbors of $v$ in $T$, by Lemma~\ref{prefix_transitive}, the simple subtraction works here. We end up finding the desired neighbour $x$ in $\mathcal{O}(\log^2 n)$ time. We then temporarily remove $x$ from the in-neighbours in $T_v$, in order to find a different in-neighbour within the next search (lines $1$-$9$ of
    Algorithm~\ref{fig:incoming_full_triangles}). Each such removal takes $\mathcal{O}(\log^2 n)$ time, hence the total time for all removals is bounded by $\mathcal{O}(l \log^2 n)$.
    In the end, we put found in-neighbours back into $T_v$ in $\mathcal{O}(l \log{n})$ time.
\end{proof}

Lemma~\ref{empty_small} gives the following corollary.

\begin{corollary}\label{incoming_full_impl_corollary}
    Under the assumptions of Lemma~\ref{incoming_full_impl}, one can implement the procedure
    $\INCOMING{v}{l}$ in time $\mathcal{O}(l \cdot \log^2{n} + |\ADT(T)|)$.
\end{corollary}

This concludes the implementation of the query subprocedure.
We are ready to prove Theorem~\ref{thm:DsTriangleFull}.

\paragraph*{Proof of Theorem~\ref{thm:DsTriangleFull}}

\begin{proof}
For $n = 0$, we define $\dsTriangle{0} = \dsBasic{0}$, and note that the implementation of the required methods is trivial.
For $n > 0$,  we define $\dsTriangle{n}$ to be the data structure $\dsBasic{n}$ maintaining $T$, extended by:
\begin{itemize}
 \item The degrees structure $\degrees$ for $T$ (see Fact~\ref{degrees_impl})
 \item Adjacency trees $\{T_u\}_{\{u \in V(T)\}}$
 \item The method \FTRIANGLE{} of Corollary~\ref{common_query_triangles_corollary} using the implementation of\\ $\INCOMING{u}{l}$ subprocedure given by Corollary~\ref{incoming_full_impl_corollary} and the implementation of $\MIN$ given by Corollary~\ref{empty_impl_promise_corollary}.
\end{itemize}
The theorem statement is a direct consequence of
Fact~\ref{degrees_impl}, Corollary~\ref{common_query_triangles_corollary}, Corollary~\ref{incoming_full_impl_corollary} and Corollary~\ref{empty_impl_promise_corollary}.
\end{proof}

We conclude this subsection with an extension of the data structure of Theorem~\ref{thm:DsTriangleFull}.

\begin{lemma}\label{update_full_impl_plus}
The data structure $\dsTriangle{n}$ can be extended into data structure $\dsTrianglePlus{n}$ via method $\FTRIANGLE{d}$, which for a given parameter $d \in \mathbb{N}$ either finds a triangle in the maintained tournament $T$, or reports that there are no triangles, or reports that $\ADT(T) \geq d + 1$.
Method $\dsTrianglePlus{n}.\FTRIANGLE{d}$ works in $\mathcal{O}(d \log^2 n)$ time.
\end{lemma}

\begin{proof}
    For $n = 0$ the implementation of $\dsTrianglePlus{n}.\FTRIANGLE{d}$ is trivial.
    Thus, we focus on describing the implementation for $n > 0$.
    Let $\emptydb$ be the empty set of $T$. To perform the query, we first check if $|\emptydb| \leq 12(d + 1)$.
    If $|\emptydb| > 12(d + 1)$ then we return that $T$ contains more than $d$ arc-disjoint triangles, which is correct by Lemma~\ref{empty_small}.
    Note, that $\mindt{\nopref{T}} \leq |\emptydb| \leq 12(d + 1)$, by Lemma~\ref{prefix_transitive} and Fact~\ref{minimal_degree} applied to $\nopref{T}$.
    Thus, we can safely run $\dsTrianglePlus{n}.\FTRIANGLE{}$ as defined in Theorem~\ref{thm:DsTriangleFull} obtaining the required time complexity, by Theorem~\ref{common_query_triangles}, Lemma~\ref{empty_impl_promise} and Lemma~\ref{incoming_full_impl}.
\end{proof}

\section{Dynamic $\FAST$ via Triangle Data Structure}\label{sec:dynamicFast}

In this section we present our dynamic algorithms for the $\FAST$
problem. We offer two algorithms, a more efficient one working in
the promise model, and a less efficient one in the full model.
To provide our results, we first show how to extend a data structure
$\dsSome{n} \in \{ \dspTrianglePlus{n}, \dsTrianglePlus{n} \}$ to
support the method $\FINDFAST{k}$, which implements
a well known branching algorithm for the static $\FAST$ problem
(see \cite{FASTBranching}, Theorem~$5$) given in
Algorithm~\ref{fig:branchingFAST}\footnote{It recursively  finds a triangle in the
tournament and tries to reverse each of its arcs.}
The argument $k$ in $\FINDFAST{k}$ can be an arbitrary natural
number (unrelated to the problem parameter $K$ that is fixed for
the lifetime of the data structure), however,
the query complexity depends on $k$.

\begin{algorithm}
    \SetKwInOut{Input}{Algorithm}
	\SetKwInOut{Output}{Output}
    \Input{$\FINDFAST{k}$ using $\dspTrianglePlus{n}$ data structure}
    \Output{Verify if $\FAST(T) \leq k$}
    \If{$T$ is acyclic}{\Return TRUE \;}
    \If{$k=0$}{ \Return FALSE \;}

    $uvw \gets \dspTrianglePlus{n}.\FTRIANGLE{k}$\;
    \If(// Too many arc-disjoint triangles){$uvw=\bot$} {
        \Return FALSE \;
    }

    \For{$xy \in \{ uv, wu, vw \}$}{
       \If{$\dspTrianglePlus{n}.\REVERSEBND{x}{y}{k}$}{
           \If{$\dspTrianglePlus{n}.\FINDFAST{k-1}$}{
                   $\dspTrianglePlus{n}.\REVERSEBND{x}{y}{4k + 4}$\;
                   \Return TRUE \;
            }
            $\dspTrianglePlus{n}.\REVERSEBND{x}{y}{4k + 4}$\;
        }
        \Else{ \Return FALSE \; }
}
    \Return FALSE \;
    \caption{Pseudocode for $\FINDFAST{k}$ for $\dspTrianglePlus{n}$ structure, for $k \in \mathbb{N}$}
    \label{fig:branchingFAST}
\end{algorithm}

\begin{theorem}[dynamic $\FAST$ promise, Theorem~\ref{thm:FastPromiseOverview} in the overview]\label{thm:FastPromise}
    Dynamic \emph{FAST} problem with parameter $K$ admits a data structure\footnoteref{note1} with initialization time $\mathcal{O}(n^2)$, worst-case update time $\mathcal{O}(\sqrt{g(K)})$ and
    worst-case query time $\mathcal{O}(3^K K\sqrt{K})$ under the promise that there is a computable function $g$, such that the maintained tournament $T$ always has a feedback arc set of size at most $g(K)$.
\end{theorem}

\begin{proof}
Without loss of generality, we assume $0 \leq K \leq g(K)$.
We use $\dspTrianglePlus{n}$ data structure from Theorem~\ref{thm:DsTrianglePromiseInternal} extended with the $\FINDFAST{\cdot}$ method presented in Algorithm~\ref{fig:branchingFAST}.
The data structure can be initialized in $\mathcal{O}(n^2)$ time, by Theorem~\ref{thm:DsTrianglePromiseInternal}.
For the update operation we call
$\REVERSE{v}{u}$ procedure inherited from $\dspTriangle{n}$ data
structure.
Since the size of the minimum feedback arc set is upper
bounded by $g(K)$, by Fact~\ref{arc_disjoint_triangles_and_fast},
the maximum number of arc-disjoint triangles in $T$ is also always
upper bounded by $g(K)$. Thus, the update operation works in time
$\mathcal{O}(\sqrt{g(K)})$.

For the query operation that verifies whether $\FAST(T) \leq K$, we
invoke the procedure $\FINDFAST{K}$ supported by $\dspTrianglePlus{n}$.
A few explanations regarding the correctness are in place before we
start the analysis. In the pseudocode, $\FINDFAST{k}$ is a
recursive procedure defined for any natural parameter $k \geq 0$.
In the recursive call, we first verify if $T$ is already acyclic,
and if it is we correctly return TRUE.
By Fact~\ref{transitive_buckets}, this can be done by simply
checking whether $|\emptydb| = 0$.
If $T$ has at least one cycle then we check if $k = 0$. If that is the case then we correctly return FALSE.
Next, we test for a triangle using method $\FTRIANGLE{k}$
supported by $\dspTrianglePlus{n}$ data structure.
We know that $T$ is not acyclic, thus the only reason for
$\FTRIANGLE{k}$ method not to return a triangle is that there are
more than $k$ arc-disjoint triangles.
We assume that in this case the method returns $\bot$.
If the invocation returns $\bot$,
by Fact~\ref{arc_disjoint_triangles_and_fast}, this means that
$\FAST(T) > k$, so we can safely return FALSE.
Otherwise we find a triangle, and we branch on its arcs trying to
reverse each of them. With each arc reversal we check if the data
structure managed to actually reverse an arc, and we assume that if
it failed then the call to
$\REVERSEBND{x}{y}{k}$ method (also supported by $\dspTrianglePlus{n}$
data structure) returns FALSE.
If that is the case, then again by
Fact~\ref{arc_disjoint_triangles_and_fast} we can safely return
FALSE. Otherwise we invoke a recursive call
$\FINDFAST{k-1}$.
When returning from the recursion, we invoke
$\REVERSEBND{x}{y}{4k + 4}$ to reverse the
changes introduced to the maintained tournament before the
recursive call.
We show that this invocation never fails.
For that, let us track the changes applied to the tournament by the recursion.
Let $T'$ be the tournament before the call
$\REVERSEBND{x}{y}{k}$ and $T''$ be the tournament $T'$ with arc between $x$ and $y$ reversed as a result of
a successful call.
Notice, that only $\REVERSEBND{\cdot}{\cdot}{\cdot}$ method
modifies the tournament. Thus, using an inductive argument, we
assume that after returning from the recursive call the tournament is again equal to $T''$.
Note, that the invocation $\REVERSEBND{x}{y}{4k + 4}$ on $T''$ fails only if $\maxdb{T''} > 8(\sqrt{4k + 5} + 1)$, as shown in the proof of Theorem~\ref{thm:DsTrianglePromiseInternal}.
However, the invocation of $\REVERSEBND{x}{y}{k}$ succeeded, which
means that $\maxdb{T'} \leq 8(\sqrt{k + 1} + 1)$.
It is easy to see, that $|\maxdb{T'} - \maxdb{T''}| \leq 2$, thus the invocation $\dspTrianglePlus{n}.\REVERSEBND{x}{y}{4k + 4}$ always succeeds, since $8(\sqrt{k + 1} + 1) + 2 \leq 8(\sqrt{4k + 5} + 1)$, for $k \in \mathbb{N}$.
Hence, this branching correctly verifies if $\FAST(T) \leq k$, by the arguments in \cite{FASTBranching}.

Let us now analyse the time complexity of $\FINDFAST{k}$ invocation. The
method $\FINDFAST{k}$ calls itself recursively at most $3^k$ times.
Apart from trivial checks, each recursive call invokes
$\mathcal{O}(1)$ calls to $\FTRIANGLE{k'}$ and
$\REVERSEBND{x}{y}{k'}$ for $k' \leq 4k + 4$. By
Theorem~\ref{thm:DsTrianglePromiseInternal}, the calls
$\FTRIANGLE{k'}$ and $\REVERSEBND{x}{y}{k'}$ take
$\mathcal{O}(k \sqrt{k})$ and $\mathcal{O}(\sqrt{k})$ time,
respectively.
Thus, $\FINDFAST{k}$ invocation takes
$\mathcal{O}(3^k k \sqrt{k} )$ time.
Hence, the total running time of the query operation is
$\mathcal{O}(3^K K \sqrt{K} )$ as claimed.
\end{proof}

\begin{theorem}[dynamic $\FAST$ full, Theorem~\ref{thm:FastFullOverview} in overview]\label{thm:FastFull}
    Dynamic \emph{FAST} problem with a parameter $K$ admits a data structure\footnoteref{note1} with initialization time $\mathcal{O}(n^2)$, worst-case update time $\mathcal{O}(\log^2{n})$ and worst-case query time $\mathcal{O}(3^K K \log^2{n})$.
\end{theorem}

\begin{proof}
Without loss of generality, we assume $K \geq 0$.
We extend the data structure $\dsTrianglePlus{n}$ with the method
$\FINDFAST{\cdot}$, which works similarly to the one shown in
Algorithm~\ref{fig:branchingFAST} for $\dspTrianglePlus{n}$.
The pseudocode of $\FINDFAST{\cdot}$ for $\dsTrianglePlus{n}$
data structure is the same as for $\dspTrianglePlus{n}$,
if we plug in $\dsTrianglePlus{n}$ data structure instead of
$\dspTrianglePlus{n}$, with one difference.
In place of $\dspTrianglePlus{n}.\REVERSEBND{x}{y}{k}$ and
$\dspTrianglePlus{n}.\REVERSEBND{x}{y}{4k + 4}$ we call
$\dsTrianglePlus{n}.\REVERSE{x}{y}$, as $\dsTrianglePlus{n}$ takes
two parameters for the arc reversal operation and always succeeds.
With that difference, the correctness analysis is analogous to the
one in the proof of Theorem~\ref{thm:FastPromise}.

Regarding the running time of the
$\FINDFAST{K}$ invocation via $\dsTrianglePlus{n}$ data structure,
again there are $3^K$ recursive calls, only now each call takes
$\mathcal{O}(K \log^2 n)$ time by Lemma~\ref{update_full_impl_plus} and Theorem~\ref{thm:DsTriangleFull}.
We conclude the proof of the theorem by noting that,
by Theorem~\ref{thm:DsTriangleFull}, the updates can be performed
in $\mathcal{O}(\log^2{n})$ time and the initialization of the
data structure takes $\mathcal{O}(n^2)$ time.
\end{proof}

\section{Dynamic $\FVST$ in the Promise Model}\label{sec:dynamicFVST}

In this section we present our main result regarding the dynamic
parametrized $\FVST$ problem stated in the theorem below.

\begin{theorem}[dynamic $\FVST$ promise, Theorem~\ref{thm:FvstPromiseOverview} in the overview]\label{thm:FvstPromise}
    Dynamic $\FVST$ problem with a parameter $K$ admits a
    data structure\footnoteref{note1} with initialization time
    $\mathcal{O}(n^2)$, worst-case update time
    $\mathcal{O}(g(K)^5)$ and worst-case query time
    $\mathcal{O}(3^K g(K) K^3)$ under the promise that there is a
    computable function $g$, such that tournament $T$ always has a
    feedback vertex set of size at most $g(K)$.
\end{theorem}

Our approach to the above result is to
run upon each query a standard branching
algorithm  given in
Algorithm~\ref{fig:branchingFVSToverview} (whose
correctness follows
from Fact~\ref{cycles_are_triangles}).
The
branching algorithm finds a triangle in
the tournament, and then
branches recursively with each of the triangle vertices removed.
The recursion is stopped when $K$ vertices are removed, where $K$
is the parameter of the problem.
To implement the branching algorithm, we
not only need
a method to find a triangle in the
maintained tournament,
but we also need to support vertex
removals and restorations, which are
significantly more complex than
edge reversals.

We deal with the latter problem first.
In Subsection~\ref{sub:drem} we introduce a new data structure
called $\dsRemove{n}{k}$. It maintains an $n$-vertex tournament
$T$ while supporting arcs reversals and vertex
removals/restorations. In order to perform these operations
efficiently the data
structure stores an explicit representation of $T$ without any
vertices removed, the set of removed vertices $F$ and an implicit
representation of
$\induced{T}{V(T) \setminus F}$.
Additionally, it
allows at most $k$ vertices to be
removed at the same time. This
data structure is covered by Lemma~\ref{DREM_def_lemma}.

Next, we need to implement a method
that allows finding triangles in
$\induced{T}{V(T) \setminus F}$. To
achieve that, we follow ideas from
Section~\ref{r:triangles}. Recall,
that the running times of the
$\dspTriangle{n}$ data structure for
finding triangles from
Theorem~\ref{thm:DsTrianglePromise}
depends on $\ADT(T)$ which is strongly
related to $\FAST(T)$. Unfortunately,
this values do not relate to $\FVST(T)$,
as we can get rid of many back arcs
by deleting one vertex. To be more precise,
the running times mainly depend on
the size of the empty set
$\emptydb$, the size of the set of back
arcs
$\backdb$ and the maximum size
of a degree bucket $\maxdb{T}$. While
$\maxdb{T}$ is bounded in terms of
$\FVST(T)$, the same does not hold for
the other two sets. In particular,
there is no reasonable bound on the
set of back
arcs
$\backdb$. To deal with this problem,
we partition the set of back arcs into
$k$-long and $k$-short back arcs.
While $k$-short back arcs are easy to
handle, non-trivial ideas are needed to
handle $k$-long back arcs.
A back arc is $k$-long if the
difference of the indegrees of its two
endpoints is at least $k$. Then we
define a $k$-long graph $G_{\LONG}$ of a
tournament,
which is an undirected graph, where
vertices are connected via an edge in
$G_{\LONG}$ if they are connected via
a $k$-long back arc in the tournament.
We also define the $k$-heavy set of a
tournament as the set of vertices of
degree higher than $k$ in $G_{\LONG}$.
We observe that the vertices of the
$k$-heavy set
need to be in \textbf{every} feedback
vertex set of size
at most $k$, so we can safely keep them
removed from the tournament. Once we
remove them, there is few $k$-long back arcs
left.

To implement this idea,
in Subsection~\ref{sub:dremp} we
introduce a wrapper data structure
around the $\dsRemove{n}{k}$ called
$\dsRemovePromise{n}{k}$.
It allows only one kind of updates, arc
reversals, and keeps the invariant that
the $k$-heavy set of the maintained
tournament is removed.
This not only allows us to later implement
the methods for finding triangles,
but also ensures fast running times of
$\dsRemove{n}{k}$ operations in
the promise model.
The wrapper is defined in
Lemma~\ref{DREMP_def_lemma}.

Finally
Subsection~\ref{sub:fvst_branching} is
devoted to proving
Theorem~\ref{thm:FvstPromise} using
previously developed data
structures to implement the branching
algorithm described above.

\subsection{Preliminary tools}\label{sub:fvst_prelim}

In this subsection we introduce basic tools used to tackle the
dynamic $\FVST$ problem. We start by introducing the indegree
distance between vertices, which in turn allows for
partitioning back arcs of a
tournament into $k$-long and $k$-short back arcs in
Definition~\ref{k_long_arcs_def}. The concept of $k$-long back arcs
turns out to be very useful later on.
In this section, for a
tournament $T$ and a set $F \subseteq V(T)$, we use the notation
$\noset{T}{F} := \induced{T}{V(T) \setminus F}$.

\begin{definition}[In-degree distance]
Let $T$ be a tournament.
We define the in-degree distance between vertices $u, v
\in V(T)$ as
$\dist_T(u,v)=d_T(u)-d_T(v)$.
\end{definition}

\begin{observation}\label{degrees_after_removal}
    Let $T$ be a tournament and $F \subseteq V(T)$.\\
    Then for all $u, v \in V(\noset{T}{F})$ $|\dist_T(u,v) - \dist_{\noset{T}{F}}(u,v)| \leq |F|$.
\end{observation}
\begin{proof}
    Removal of $v$ from a tournament decreases in-degree of every other vertex by either $0$ or $1$.
\end{proof}

\begin{definition}[$k$-long and $k$-short arcs]\label{k_long_arcs_def}
    Let $T$ be a tournament and $k \in \mathbb{N}$.
    An arc $uv \in E(T)$ is called a $k$-\emph{long arc}
    if $|\dist_T(u,v)| \geq k$. Otherwise, it is called a $k$-\emph{short arc}.
\end{definition}

Similarly as in Subsection~\ref{sec:prelimtools}, in order to be able to
search for triangles, we need bounds on some basic parameters
of the tournament, but this time with respect to the size of
the minimum feedback vertex set $\FVST(T)$.
We start by bounding the size of degree buckets and the size of
the empty
set of a tournament $T$ in terms of the both $\FVST(T)$
and the number of $k$-long back arcs.

\begin{lemma}\label{small_buckets_fvst}
    Let $T$ be a tournament. Then $\maxdb{T} \leq 2 \FVST(T) + 1$.
\end{lemma}

\begin{proof}Let $|V(T)|=n.$
    Without loss of generality, we can assume $0 \leq \FVST(T) < n$.
    Consider arbitrary $d \in [n]$. Let $F$ be the feedback vertex set of minimum size equal to $\FVST(T)$.  
    Let $S = \db{T}[d] \setminus F$.
    Each removal of a vertex from a tournament decreases the in-degree of every other vertex by either $0$ or $1$. Thus, for all $v \in S$ it holds that $d - |F| \leq d_{T^{-F}}(v) \leq d$.
    By Fact~\ref{transitive_buckets}, all vertices from $S$ have different degrees in $T^{-F}$, so $|S| \leq |F|+1$.
    Hence, $|\db{T}[d]| \leq |S| + |F| \leq 2|F| + 1$.
\end{proof}

\begin{lemma}\label{small_empty_fvst}
    Let $T$ be a tournament on $n$ vertices, $k \in \mathbb{N}$ and $\lbackdb$ be the set of all $k$-long back arcs in $T$.
    Let $\emptydb$ be the empty set of $T$. Then $|\emptydb| \leq \FVST(T) \cdot (2k + |\lbackdb| + 2 \FVST(T) + 5)$.
\end{lemma}

\begin{proof}
For a vertex $v \in V(T)$ let
$\emptydb^{-v}$ denote the empty set of $\noset{T}{\{ v \}}$
and let
$L_v = \{u \in V(T)$ : $uv \in \lbackdb$ or $vu \in \lbackdb \}$.
Let also $l_v = d_{T}(v) - k - 1$ and $r_v = d_{T}(v) + k + 1$.
First, we prove the following claim.
\begin{claim}\label{clm:empty_rek}
Let $F$ be a feedback vertex set of size $\FVST(T)$. Let
us assume that $F$ is not empty and let $v \in F$ be any vertex in $F$. Then
$|\emptydb| \leq |[l_v - 1, r_v + 1]| + |\emptydb^{-v}| + |L_v| = 2k + 5 + |\emptydb^{-v}| + |L_v|$.
\end{claim}
\begin{proof}[Proof of Claim~\ref{clm:empty_rek}]
Let us consider an in-degree $d \in [n] \setminus [l_v-1,r_v+1]$ such that
$\db{T}[d] = \emptyset$, or in other words,
$d$ represents some empty degree bucket outside of $[l_v-1,r_v+1]$.
Our goal is to injectively assign $d$ to $f(d) \in [n]$ in a way that
\begin{enumerate}
  \item either $\db{\noset{T}{v}}[f(d)]$ is empty
  \item or $\db{T}[f(d)] \cap L_v$ is non empty.
\end{enumerate}

Let us first consider an easier case when $d < l_v-1$.
If $\db{\noset{T}{v}}[d]$ is also empty, we set $f(d)=d$. Otherwise, $\db{T}[d+1]$ contains a vertex of $L_v$ and we
set $f(d)=d+1$.
It is easy to see that the assignment $f$ is an
injection for $d < l_v-1$.

We now consider the case when $d > r_v+1$. Let
$d''=\min_{i > d} \{i: \db{T}[i]=\emptyset $ or
$ \db{T}[i] \cap L_v \neq \emptyset \}$ (here we artificially assume that $\db{T}[n]=\emptyset$).
If $\db{T}[d'']=\emptyset$,
    that means that in all buckets $\db{T}[i]$ for $i \in [d+1,d'']$ there are no vertices of $L_v$. This implies that all vertices in the buckets $\db{T}[i]$ for $i \in [d+1,d'']$ are out-neighbours of $v$, and hence they decrease their in-degree by one as a result of removing $v$. So the buckets $\db{T}[i]$ for $i \in [d+1,d'']$ get all shifted one to the left when we remove $v$ from $T$. This implies that $\db{\noset{T}{v}}[d''-1]=\db{T}[d'']=\emptyset$, so we set $f(d)=d''-1$.
    If, on the other hand, $\db{T}[d''] \cap L_v \neq \emptyset$, we set $f(d)=d''$. Again, it is easy to see that the constructed function is an injection.
\end{proof}
We now move on to proving Lemma~\ref{small_empty_fvst}
by induction on $\FVST(T)$. The basis
of the induction is when $\FVST(T) = 0$. Then,
$T$ is acyclic and by Fact~\ref{transitive_buckets}, $|\emptydb| = 0$.

In the induction step, we assume $0 < \FVST(T) < n$.
Let thus $F$ be the feedback vertex set for $T$ of size
$\FVST(T)$, what implies that $F$ is not empty. Let us
fix an arbitrary vertex $v \in F$.
By applying the inductive hypothesis to $\noset{T}{v}$, $F \setminus \{v\}$ and $k + 1$, we obtain $$|\emptydb^{-v}| \leq (|F| - 1)(2k + 2 + |\lbackdb^{-v}| + 2(|F| - 1) + 5),$$ where $\lbackdb^{-v}$ is the set of all $(k + 1)$-long back arcs in $\noset{T}{v}$.
    By Observation~\ref{degrees_after_removal} applied to set $F=\{ v \}$, $\lbackdb^{-v} \subseteq \lbackdb$, hence we get the following:
\begin{flalign*}
\emptydb & \leq 2k + 5 + |\emptydb^{-v}| + |L_v| \leq&\\
     & \leq 2k + 5 + (|F| - 1) \cdot (2k + |\lbackdb| + 2|F| + 5) + |\lbackdb| =&\\
     & =|F| \cdot (2k + |\lbackdb| + 2|F| + 5 - 2) \leq&\\
     & \leq |F| \cdot (2k + |\lbackdb| + 2|F| + 5)
     \end{flalign*}
\end{proof}

To conclude the preliminaries, we define the $k$-long graph of a tournament, which is used later in the section.

\begin{definition}[$k$-long graph]\label{k_long_back_arcs_graph_def}
    Let $T$ be a tournament, $k \in \mathbb{N}$ and
    $\lbackdb \subseteq E(T)$ be the set of all $k$-long back
    arcs in $T$. An undirected graph $G_{\LONG}$ is called the
    $k$-\emph{long graph} of $T$ if $V(G_{\LONG}) = V(T)$ and
    $\{v, u\} \in E(G_{\LONG}) \iff uv \in \lbackdb \vee vu \in \lbackdb$.
\end{definition}

\subsection{Data Structure for Vertex Removals}\label{sub:drem}

In this subsection we provide a data structure $\dsRemove{n}{k}$
that maintains an $n$ vertex tournament $T$ and supports both
types of updates: arc reversal and vertex deletion/restoration
(adding previously deleted vertex back to $T$).
The vertex updates are more difficult than arc reversals,
because they affect in-degrees of up to $(n - 1)$ vertices,
while the edge reversals affect indegrees of $2$.
In order to perform vertex updates efficiently,
the data structure allows at most $k$ vertices to be removed at
the same time.
It stores an explicit representation of $T$ without any vertices
removed, a dynamically changing set of removed vertices $F$ and an
implicit representation of tournament $\noset{T}{F}$.

As observed above, we need to be able to quickly
decrease/increase in-degrees of many vertices, because
removal/restoration of a vertex can pessimistically affect all
other vertices.
In order to do so, the $\dsRemove{n}{k}$ data structure stores a
set of tokens, one for each removed vertex.
The vertices are partitioned into buckets that are similar to
degree buckets from Section~\ref{r:triangles}, and
 a token at position $x$ ,,decreases'' in-degrees of
all vertices in the buckets with indices $d \geq x$.
The difference with respect to Section~\ref{r:triangles}
is that the partitioning into buckets is not according to the
actual indegree in the tournament, but rather according to
a ``reduced degree''. The reduced degree of $v$ is an
approximation of
$d_{\noset{T}{F}}(v)$.
To be more precise, the value of $d_{\noset{T}{F}}(v)$ can be
calculated by subtracting from $v$'s
reduced degree the number of tokens at positions
smaller or equal then
$v$'s reduced degree.
This is formalized by invariant $1$ in Lemma~\ref{DREM_def_lemma}.

Apart from $T$, $F$ and tokens, the $\dsRemove{n}{k}$ data
structure stores an implicit representation of tournament
$\noset{T}{F}$ and the set of all $k$-long back arcs in $T$ in a
form of a $k$-long graph of $T$.
The implicit representation is similar to the  representation
that we used in Section~\ref{r:triangles}, that is we maintain
the reduced degrees, the partition of $V(\noset{T}{F})$ into
degree buckets (with respect to their reduced degrees) and a
subset of such buckets that are empty. We move on to describing
the data structure for vertex removals in more detail.

\begin{lemma}[Definition of $\dsRemove{n}{k}$]\label{DREM_def_lemma}
For every $n \in \mathbb{N}$ and $k \in [n + 1]$ there exists a data structure $\dsRemove{n}{k}$ that for a dynamic tournament $T$ on $n$ vertices\footnoteref{note1} altered by arc reversals and a dynamic set of vertices $F \subseteq V(T)$, such that $|F| \leq k$, maintains $T$, $F$ and $\noset{T}{F}$ by maintaining the following information:
\begin{enumerate}
 \item All information stored by $\dsBasic{n}$ for tournament $T$ (see Lemma~\ref{update_impl_basic})
 \item $\rmSet$ -- set $F$ of currently removed vertices as an array set over $[n]$
 \item $\reducedDegree{}$ -- the array storing, for each vertex $v \in V(\noset{T}{F})$, a value called the reduced degree ranging from $0$ to $n - 1$; for vertices $v \in F$ the value is undefined
 \item $\reducedDegreeBucket{}$ -- the partition of vertices in $V(\noset{T}{F})$ into $n$ buckets based on the values in $\reducedDegree{}$.
 Reduced degree buckets $\reducedDegreeBucket{}$ are stored analogously as the standard degree buckets in $\dsBasic{n}$ (points $3$. and $4$. in Lemma~\ref{update_impl_basic})
 \item $\reducedEmptydb$ -- the array set over $[n]$ of indices $i \in [n]$ such that $\reducedDegreeBucket{i}=\emptyset$
 \item $\tok$ -- the array set over $[5n]$ storing positions of tokens
 \item $\vtok{}$ -- the array storing an assignment of each vertex $v \in F$ to the position of its token; for $v \in V(\noset{T}{F})$ the value is undefined
  \item $\longArcs{}$ -- the array storing, for each vertex $v \in V(T)$, the set of its neighbours in the $k$-long graph of $T$ as an array set over $[n]$
\end{enumerate}

Let $\ctok{d}$ for $d \in \mathbb{Z}$ denote the number of token positions in $\tok$ that are not greater than $d$.
The data structure satisfies the following important invariants:
 \begin{enumerate}
   \item For all $v \in V(\noset{T}{F})$ it holds that $d_{\noset{T}{F}}(v) = \reducedDegree{v} - \ctok{\reducedDegree{v}}$
   \item the assignment stored by $\vtok{}$ is a bijective function from $F$ to the set of token positions $\tok$, i.e., every vertex $v \in F$ is assigned a unique token position in $\tok$, and every position in $\tok$ is assigned to a unique vertex $v \in F$. Moreover, the positions in $\tok$ are pairwise different.
   \item For all $v \in F$ it holds that $|\vtok{v} - d_{T}(v)| \leq 6k - 3$
\end{enumerate}

Data structure $\dsRemove{n}{k}$ supports the following operations:
\begin{enumerate}
 \item $\INITIALIZEDREM{T}{k}{F}$ -- initializes the data structure with tournament $T$, parameter $k$ and set of removed vertices $F$, in time $\mathcal{O}(n^2)$
 \item $\REVERSEARCDREM{v, u}$ -- reverses an arc between vertices $v, u \in V(T)$ in $T$ in time $\mathcal{O}((k + 1)^3 \cdot (\maxdb{T} + 1))$. Additionally, the operation returns the list of all newly added/removed $k$-long back arcs in $T$ of size at most $12 \maxdb{T} + 1$, where $T$ is the maintained tournament before the method invocation
 \item $\REMOVEVERTEX{v}$ -- adds vertex $v \in V(\noset{T}{F})$ to the set $F$ of removed vertices, in time $\mathcal{O}((k + 1) \cdot (|F|+1)^2 \cdot (\maxdb{T} + 1)  + (|F|+1)(|\longArcs{v}| + 1) )$
 \item $\RESTOREVERTEX{v}$ -- removes vertex $v \in F$ from the set $F$ of removed vertices, in time $\mathcal{O}((k + 1) \cdot (|F|+1)^2 \cdot (\maxdb{T} + 1)   + (|F| + 1)(|\longArcs{v}| + 1) )$
\end{enumerate}
\end{lemma}

In the above lemma, to be formally correct, in the running times
depending on multiple
parameters the factors are increased by one so that the running
time does not drop to zero if one of the parameters is zero.

Before we prove Lemma~\ref{DREM_def_lemma},
we introduce some useful observations.

\begin{observation}\label{monotonic_tokenized_degrees}
    Consider data structure $\dsRemove{n}{k}$ maintaining tournament $T$ and set $F$, where $k \in [n + 1]$.
    Then, for all $v, u \in V(\noset{T}{F})$ and $l \in \mathbb{N}$
    it holds that
    \begin{enumerate}
    \item if $\reducedDegree{u} \leq \reducedDegree{v}$ then $d_{\noset{T}{F}}(u) \leq d_{\noset{T}{F}}(v)$,
    \item if $\reducedDegree{v} - \reducedDegree{u} \geq l + |F|$ then $d_{\noset{T}{F}}(v) - d_{\noset{T}{F}}(u) \geq l$ and
    \item if $\reducedDegree{v} - \reducedDegree{u} \geq l + 2|F|$ then $d_{T}(v) - d_{T}(u) \geq l$.
    \end{enumerate}
\end{observation}

\begin{proof}
Point 1. follows from the first and the second invariant.
Point 2. follows from the first and the second invariant in the following way:
\begin{flalign*}
l + |F| & \leq \reducedDegree{v} - \reducedDegree{u} =&\\
     & = \reducedDegree{v} - \ctok{\reducedDegree{v}} -
     (\reducedDegree{u} - \ctok{\reducedDegree{u}}) &+\\
     &+ \ctok{\reducedDegree{v}} - \ctok{\reducedDegree{u}}
    =&\\
     & =d_{\noset{T}{F}}(v) - d_{\noset{T}{F}}(u) + \ctok{\reducedDegree{v}} - \ctok{\reducedDegree{u}} \leq&\\
     & \leq d_{\noset{T}{F}}(v) - d_{\noset{T}{F}}(u) + |F|
     \end{flalign*}
Point 3. follows from combining Point 2. with Observation~\ref{degrees_after_removal}.
\end{proof}

\begin{observation}\label{tokenized_buckets_size_fvst}
    Consider the $\dsRemove{n}{k}$ data structure maintaining
    tournament $T$ and set $F$, where $k \in [n + 1]$. Then for
    all $d \in [n]$ it holds that $|\reducedDegreeBucket{d}| \leq \maxdb{T} \cdot (|F| + 1)$.
\end{observation}

\begin{proof}
From the first invariant we immediately get $\reducedDegreeBucket{d} \subseteq \db{\noset{T}{F}}[d - \ctok{d}]$.
Further, $\db{\noset{T}{F}}[d - \ctok{d}] \subseteq
\{u \in V(\noset{T}{F})$ : $d_{T}(u) \in [d - \ctok{d}, d - \ctok{d} + |F|]\}$
as each removed vertex can reduce the indegree of another vertex
by at most one.
Hence, $|\db{\noset{T}{F}}[d - \ctok{d}]| \leq \maxdb{T} \cdot (|F| + 1)$.
\end{proof}

The next lemma helps determining the reduced degrees of
vertices
of a given in-degree in tournament $\noset{T}{F}$.
The lemma is presented in a generalized version, because it is
also used in situations where we do not have access to the valid
$\dsRemove{n}{k}$ maintaining the given tournament and set.

\begin{lemma}\label{tokenized_ops_fvst}
    Let $n \in \mathbb{N}$ and $K \subseteq [n]$ be a set of integers represented by an array set over $[n]$.
    Let $t_i = |\{t \in K$ : $t \leq i\}|$ for $i \in \mathbb{Z}$.
    Then, for any integer $d \in \mathbb{Z}$, it is possible to list all $i \in \mathbb{Z}$ satisfying $i-t_i = d$ in time $\mathcal{O}(|K|)$. Additionally, all such $i$ form a non-empty interval entirely contained inside $[d, d + |K|]$.
\end{lemma}

\begin{proof}
We start by calculating $j=\min ( \{i \in \mathbb{Z}: i-t_i=d \})$. This can be done in time $\mathcal{O}(|K|)$, since
$d \leq j \leq d + |K|$. Starting from $i=j$, we increase $i$ by one as long as $t_{i+1}=t_i+1$. It is easy to see that we obtain this way the desired interval of values $i$ such that
$i-t_i = d$. Moreover, we can increase $i$ at most $|K| - (j - d)$ times.
\end{proof}

We are now ready to give the proof of Lemma~\ref{DREM_def_lemma}.
In vertex removal/restoration operation, the most problematic
part is keeping invariant $1$ (the relation between reduced degrees, tokens and in-degrees in $\noset{T}{F}$) satisfied, because removing/restoring a vertex can affect in-degrees of up to $n - 1$ other vertices.
Both operations follow the same general approach to solve this issue.

We first update the stored information regarding
vertex $v$ (the one we remove/restore). We then
need to fix the reduced degrees of all affected vertices in
order to satisfy invariant $1$.
We first show a method $\fixRdeg{u}$ that
given any vertex $u$, quickly fixes its reduced degree.
We then fix the reduced degree of $v$. Next,
we fix reduced degrees of other vertices in order to fix
invariant $1$ globally.
To accomplish that, we first iterate over vertices having their
reduced
degrees
inside a small interval $[l, r]$ around the reduced degree of $v$,
and we fix all these vertices.
The interval is chosen in a way that all arcs between $v$ and
vertices with reduced degrees outside of $[l, r]$ are $k$-long
arcs with regard to $T$.
By considering all cases, we arrive at a conclusion that
removing/restoring vertex $v$ does not alter in-degrees in
$\noset{T}{F}$ of vertices to the left of $[l, r]$ and alters
in the same way (either adds $1$ or subtracts $1$) in-degrees
in $\noset{T}{F}$ of almost all vertices to the right of $[l, r]$.
Such a regular change to in-degrees can be handled
by addition/removal of a token for $v$
(depending whether we remove or restore $v$).
However, there is a small caveat to the above.
The exceptions to this regular change of indegrees
are vertices connected to $v$ via a
$k$-long back arc.
For those vertices, we need to directly fix the reduced degree
of such a vertex.
Because of that, we finish the operation by iterating over
$\longArcs{v}$ and fixing the reduced degrees of yielded vertices.

\begin{algorithm}
    \SetKwInOut{Input}{Operation}
	\SetKwInOut{Output}{Output}
    \Input{$\dsRemove{n}{k}.\REMOVEVERTEX{v}$}
    \BlankLine

    $l \leftarrow \reducedDegree{v} - k - 2|\rmSet| - 1$\;
    $r \leftarrow \reducedDegree{v} + k + 2|\rmSet| + 1$\;
    $rt \leftarrow$ minimal $i \geq r$, such that $i + 1 \notin \tok$\;
    Add token $rt + 1$ to $\tok$\;
    Set $\vtok{v} = rt+1$\;
    Add $v$ to $\rmSet$, remove $v$ from $\reducedDegreeBucket{\reducedDegree{v}}$, and update $\reducedEmptydb$ accordingly\;

    \For{$u \in \reducedDegreeBucket{[l, rt] \cap [n]}$}{
       Move $u$ from $\reducedDegreeBucket{\reducedDegree{u}}$ to $\reducedDegreeBucket{\fixRdeg{u}}$\;
       Update $\reducedEmptydb$ accordingly\;
       Set $\reducedDegree{u}=\fixRdeg{u}$\;
    }

    \For{$u \in \longArcs{v}$}{
       \If{$u \notin \rmSet$}{
           Move $u$ from $\reducedDegreeBucket{\reducedDegree{u}}$ to $\reducedDegreeBucket{\fixRdeg{u}}$\;
           Update $\reducedEmptydb$ accordingly\;
           Set $\reducedDegree{u}=\fixRdeg{u}$\;
       }
    }

    \caption{Pseudocode for $\REMOVEVERTEX{v}$ method of the
    $\dsRemove{n}{k}$ data structure}
    \label{fig:remove_vertex_fvst}
\end{algorithm}

\begin{algorithm}
    \SetKwInOut{Input}{Operation}
	\SetKwInOut{Output}{Output}
    \Input{$\dsRemove{n}{k}.\RESTOREVERTEX{v}$}
    \BlankLine

    $t \leftarrow \vtok{v}$\;
    Remove token position $t$ from $\tok$ \;
    Remove $v$ from $\rmSet$\;
    Set $\reducedDegree{v}=\fixRdeg{v}$ \;
    Add $v$ to $\reducedDegreeBucket{\reducedDegree{v}}$ and update $\reducedEmptydb$ accordingly \;
    $l \leftarrow \min(\reducedDegree{v} - k - 2|\rmSet| - 1,t)$\;
    $r \leftarrow \max(\reducedDegree{v} + k + 2|\rmSet| + 1,t)$\;

    \For{$u \in \reducedDegreeBucket{[l, r] \cap [n]}$}{
       Move $u$ from $\reducedDegreeBucket{\reducedDegree{u}}$ to $\reducedDegreeBucket{\fixRdeg{u}}$\;
       Update $\reducedEmptydb$ accordingly\;
       Set $\reducedDegree{u}=\fixRdeg{u}$\;
    }

    \For{$u \in \longArcs{v}$}{
       \If{$u \notin \rmSet$}{
           Move $u$ from $\reducedDegreeBucket{\reducedDegree{u}}$ to $\reducedDegreeBucket{\fixRdeg{u}}$\;
           Update $\reducedEmptydb$ accordingly\;
           Set $\reducedDegree{u}=\fixRdeg{u}$\;
       }
    }

    \caption{Pseudocode for $\RESTOREVERTEX{v}$ method of
    the $\dsRemove{n}{k}$ data structure}
    \label{fig:restore_vertex_fvst}
\end{algorithm}

\begin{proof}[Proof of Lemma~\ref{DREM_def_lemma}]
We start by defining an internal procedure $\fixRdeg{w}$ for
$w \in V(\noset{T}{F})$, which computes the reduced degree of a vertex $w$ based on the current list of token positions.
We first compute
$d_{\noset{T}{F}}(w)$ using $d_T(w)$, adjacency matrix of $T$ and set $\rmSet$, in time $\mathcal{O}(|F|)$.
We can now compute the value $\reducedDegree{v}$ satisfying $d_{\noset{T}{F}}(w) = \reducedDegree{v} - \ctok{\reducedDegree{v}}$ using Lemma~\ref{tokenized_ops_fvst} for $K = \tok$, in $\mathcal{O}(|\tok|)$ time.
    By Lemma~\ref{tokenized_ops_fvst}, the computed value satisfies  $\reducedDegree{v} \in [d_{\noset{T}{F}}(w), d_{\noset{T}{F}}(w) + |\tok|] \subseteq [n]$.
To sum up, procedure $\fixRdeg{w}$ takes $\mathcal{O}(|F|)$ time.
\newline
\newline
\textbf{Initialization.}\\ We move on to describing operation $\INITIALIZEDREM{T}{k}{F}$. This can be done similarly as in Lemma~\ref{update_impl_basic}, with the difference that now we also need to initialize the information related to tokens, reduced degrees and $k$-long back arcs. We start by initializing $\dsBasic{n}$ and setting $\rmSet := F$. Then, for every vertex $v \in F$ we add a token at position $t_v$, such that $|t_v - d_T(v)| \leq |F|$. We can find $t_v$ in time $\mathcal{O}(|F|)$.
Next, for every $v \in V(\noset{T}{F})$ set $\reducedDegree{v} = \fixRdeg{v}$.
The remaining data and can be easily calculated in $\mathcal{O}(n^2)$ time.
The first invariant is satisfied by the choice of reduced degrees.
The second invariant is satisfied trivially.
The third invariant follows from the assumption $|F| \leq k$: if $k = 0$ then it is trivially satisfied, else for all $v \in F$ it hods that $|\vtok{v} - d_T(v)| \leq |F| \leq k$.
\newline
\newline
\textbf{Arc reversal.}\\
We now move on to describing $\REVERSEARCDREM{v, u}$ operation.
Let $vu \in E(T)$ be the arc we want to reverse and let $\overline{T}$ be a tournament $T$ with this arc reversed.
The basic information of $\dsBasic{n}$ can be maintained by
Lemma~\ref{update_impl_basic}.
We update such basic information as the first step.

Observe, that for vertices $w \in V(\noset{T}{F}) \setminus \{ u,v \}$ it holds that
$d_{\noset{T}{F}}(w)=d_{\noset{\overline{T}}{F}}(w)$,
so the first invariant for these vertices remains satisfied.
Similarly,
for vertices $w \in F \setminus \{ u,v \}$
the third invariant remains satisfied. So the only
vertices for which the invariants first and third may not
hold after reversal of $vu$ are $v$ and $u$.
We first fix the third invariant in the following way.
For $w \in \{v, u\} \cap F$ we modify $\vtok{w}$ so that $|\vtok{w} - d_{\overline{T}}(w)| \leq |F|$ (let us denote the old value of $\vtok{w}$ as $t_w$). This can be done in $\mathcal{O}(|F|)$ time, since the number of all tokens is bounded by $|F|$. This fixes the third invariant for all vertices, so it remains to fix the first invariant. Notice, that the vertices
for which the first invariant may be broken are $v,u$ and
the vertices in reduced degree buckets $\reducedDegreeBucket{d}$ for
$d \in [t_v,\vtok{v}] \cup [\vtok{v}, t_v] \cup
[t_u,\vtok{u}] \cup [\vtok{u},t_u]$ (here, $[x,y] = \emptyset$ if $y < x$). For all such vertices $x$ we fix the first invariant by setting $\reducedDegree{x}=\fixRdeg{x}$ and updating the
reduced degree buckets $\reducedDegreeBucket{}$ and $\reducedEmptydb$ accordingly.
Since we moved two tokens by $\mathcal{O}(|F|)$ positions,
by Observation~\ref{tokenized_buckets_size_fvst}, this totals to $\mathcal{O}((k + 1)^3 \cdot (\maxdb{T} + 1))$ time, because $\fixRdeg{x}$ takes
$\mathcal{O}(|F|)$ time.

It remains to update $\longArcs{}$. Observe that all the arcs $L$ that start or stop to be $k$-long back arcs due to reversal of $vu$ contain either $v$ or $u$ as one of their endpoints. If $v$ is one of the endpoints of an arc in $L$, then the other
endpoint is either $u$ or belongs to $\db{T}[d]$ for $d \in [d_T(v)-k-1,d_T(v)-k+1] \cup [d_T(v)+k-1,d_T(v)+k+1]$.
The situation is similar for $u$.
Thus, there is a limited number of candidates for arcs in $L$ and we can compute $L$ in
$\mathcal{O}(\maxdb{T})$ time.
Then, updating $\longArcs{}$ is straightforward.
It is easy to see that the bound $|L| \leq 12 \maxdb{T} + 1$ required by the definition of $\REVERSEARCDREM{v, u}$ operation is satisfied.
\newline
\newline
\textbf{Vertex Removal.}\\
    Let us now move to the first vertex update operation: $\REMOVEVERTEX{v}$, which adds vertex $v$ to the set of removed vertices $F$. Here, by $F$ we denote the set of removed vertices together with $v$, whereas $F \setminus \{ v \}$ denotes the
    set of removed vertices before removing $v$.
    To remove $v$, we first let $l = \reducedDegree{v} - k - 2|F \setminus \{ v \} | - 1$ and $ r = \reducedDegree{v} + k + 2|F \setminus \{ v \}| + 1$.
    We compute $rt = \min_{i \geq r}(i + 1 \notin \tok)$ in time $\mathcal{O}(|\tok|)$ in a straightforward manner.
We now add $v$ to $\rmSet$ representing $F$, we insert its token $rt+1$ to $\tok$, and we set $\vtok{v} := rt + 1$.
We then remove $v$ from $\reducedDegreeBucket{\reducedDegree{v}}$
and update $\reducedEmptydb$ accordingly.
We can bound $rt$ from above as follows
\begin{flalign}
rt & \leq r+|F|-1 =&\\
     & = \reducedDegree{v} + k + 2(|F|-1) + 1 + |F|-1
    \leq &\\
     & = \reducedDegree{v} + k + 3|F|-2
    \leq &\\
     & \leq n-1 + n + 3n -2 \leq&\\
     & \leq 5n-3,
     \end{flalign}
where Equation (5.5) implies that $rt+1$ is a valid token position.

To fix the reduced degrees, we iterate over $i \in [l, rt] \cap [n]$ and for every $u \in \reducedDegreeBucket{i}$,
we fix its reduced degree by setting $\reducedDegree{u}=\fixRdeg{u}$ and updating reduced degree buckets $\reducedDegreeBucket{}$ and $\reducedEmptydb$, in $\mathcal{O}(|F|)$ total time per vertex.
    Then, for every $u \in \longArcs{v}$ if $u \notin \rmSet$ we fix reduced degree of $u$ analogously. The pseudocode is given in Algorithm~\ref{fig:remove_vertex_fvst}.
    In total, the $\REMOVEVERTEX{v}$ operation takes $\mathcal{O}( (k +  |F| + 1) \cdot (\maxdb{T} + 1) \cdot (|F| + 1)^2 + (|\longArcs{v}| + 1) \cdot (|F| + 1) ) = \mathcal{O}((k + 1) \cdot (|F|+1)^2 \cdot (\maxdb{T} + 1) + (|F| + 1)(|\longArcs{v}| + 1) )$ time, where we bound $|\reducedDegreeBucket{i}|$ using Observation~\ref{tokenized_buckets_size_fvst}.

    We now prove that the invariants of the $\dsRemove{n}{k}$ data structure are satisfied. The second invariant is clearly satisfied, as we added one new token corresponding to $v$ at
    a free and valid position, while the other tokens did not change their positions. Let us move on to the third invariant, which could have been affected by the above operations only for vertex $v$, as the other tokens did not change their positions.
    To prove the third invariant, let $d$ be the value of $\reducedDegree{v}$ before modifications and $t_d = \ctok{d}$, but without considering the new token at position $rt + 1$. Let $F$ be the set of removed vertices after $v$ is removed, i.e., $v \in F$. Using the bound on $rt$ from Equation (5.3) we obtain:
\begin{flalign*}
|d_T(v) - \vtok{v}| & = |d_T(v) - (rt + 1)| \leq &\\
     & |d_T(v) - d_{\noset{T}{(F \setminus \{ v \})}}(v)| + |d_{\noset{T}{(F \setminus \{ v \})}}(v) - (rt + 1)| \leq &\\
     & |F| - 1 + |d - t_d - (rt + 1)| \leq &\\
     & |F| - 1 + |rt-d|+|-t_d-1| \leq &\\
     & |F| - 1 + (k+3|F|-2) +(t_d+1) = & \\
     &  k + 4|F| - 2 +  t_d \leq &\\
     & 5|F| + k - 3 \leq 6k - 3.
     \end{flalign*}

It remains to prove the first invariant. Let us consider a vertex
$u \in V(\noset{T}{F})$ and let $d_u$ be the reduced degree of
$u$ before any modifications related to removing $v$. If $d_u \in
[l,rt]$, then the algorithm fixes the reduced degree of $u$ when iterating through $i \in [l,rt]$. Let us now consider the case when
$d_u < l$. Then, by Observation~\ref{monotonic_tokenized_degrees} (Point 3), $d_T(v) - d_T(u) \geq k$.
Thus, if $vu \in E(\noset{T}{F \setminus \{ v \} })$ then $u \in \longArcs{v}$, and as a consequence the reduced degree of $u$ gets fixed by iterating through $\longArcs{v}$. If, on the other hand, $uv \in E(\noset{T}{F \setminus \{ v \} })$, then the reduced
degree of $u$ needs not to be fixed, as the new token is at position larger than $d_u$. It remains to consider the case when $d_u > rt$. Similarly as above, by Observation~\ref{monotonic_tokenized_degrees} (Point 3), $uv \in E(\noset{T}{F \setminus \{ v \} })$ implies
$u \in \longArcs{v}$ and the reduced degree of $u$ is fixed by the algorithm. If, on the other hand, $vu \in E(\noset{T}{F \setminus \{ v \} })$, then the new token increases the value of $\ctok{\reducedDegree{u}}$ by $1$, so the reduced degree of $u$ satisfies the first invariant and does not need to be updated.
\newline
\newline
\textbf{Vertex Restoration.}\\
    We now move on to describing the last operation: $\RESTOREVERTEX{v}$. The pseudocode for this operation is shown in Algorithm~\ref{fig:restore_vertex_fvst}.
    The procedure $\RESTOREVERTEX{v}$ first removes $v$'s token $t=\vtok{v}$ from $\tok$.
    Then, it removes vertex $v$ from set $\rmSet$ representing the set of removed vertices $F$. Then, it restores the reduced
    degree of $v$ by setting $\reducedDegree{v}=\fixRdeg{v}$,
    adding $v$ to $\reducedDegreeBucket{\reducedDegree{v}}$ and
    updating $\reducedEmptydb$ accordingly if necessary.
    Then, the procedure recomputes the reduced degrees of the
    other vertices in the following way, analogous to the one for the vertex removal operation. Here, by $F$ we denote the set $F$ after the restoration of $v$, i.e., $v \notin F$.
    Let $l = \min(\reducedDegree{v} - k - 2|F| - 1, t)$ and $r = \max(\reducedDegree{v} + k + 2|F| + 1, t)$.
    The procedure iterates over $i \in [l, r] \cap [n]$ and for each $u \in \reducedDegreeBucket{i}$ it fixes its reduced degree by setting $\reducedDegree{u}=\fixRdeg{u}$, moving $u$ to $\reducedDegreeBucket{\fixRdeg{u}}$ and updating $\reducedEmptydb$ accordingly, in total time $\mathcal{O}(|F|)$ per vertex. Finally, it iterates over $u \in \longArcs{v}$ and fixes the reduced degree of $u$ analogously, only if $u \notin \rmSet$.
    Observe, that
\begin{flalign*}
    |l - r| & \leq 2k+4|F|+2 + |t - \reducedDegree{v}| \leq &\\
    & 4|F| + 2k + 2 + |t - d_{\noset{T}{F}}(v)| + |F| \leq &\\
    & 5|F| + 2k + 2 + |t - d_T(v)| + |d_T(v) - d_{\noset{T}{F}}(v)| \leq &\\
    & 6|F| + 2k + 2 + 6k - 3 = 6|F| + 8k - 1 \\
\end{flalign*}
    Hence, by Observation~\ref{tokenized_buckets_size_fvst}, the procedure takes $\mathcal{O}((|F| + k + 1) \cdot (\maxdb{T} + 1) \cdot (|F| + 1)^2 + (|\longArcs{v}| + 1) \cdot (|F| + 1)) = \mathcal{O}((k + 1) \cdot (|F|+1)^2 \cdot (\maxdb{T} + 1)   + (|F| + 1)(|\longArcs{v}| + 1) )$ time.

    We now prove that the invariants of a $\dsRemove{n}{k}$ data structure are satisfied.
    The second and the third invariant are trivially satisfied.
    We move on to proving the remaining first invariant.
    For a vertex $u \in V(\noset{T}{F \cup \{ v \} })$ let $\reducedDegreeBefore{u}$ denote its reduced degree before any modifications related to vertex restoration took place. Similarly, let $\ctokBefore{d}$ denote the number of tokens up until position $d$ including token $t$, i.e., before  removing $t$ from $\tok$.
    It is straightforward that the first invariant holds for vertices $u$ such that $\reducedDegreeBefore{u} \in [l,r]$.

    So let us consider a
    vertex $u$ such that $\reducedDegreeBefore{u} \notin [l,r]$.
    By the definition of $l$ and $r$ it holds that $|\reducedDegree{v}-\reducedDegreeBefore{u}| \geq k+2|F|+2$.
    The analysis now splits into two cases, one when
    $\reducedDegreeBefore{u} < l$ and one when $\reducedDegreeBefore{u} > r$.

    Let us first consider the case when $\reducedDegreeBefore{u} < l$. In that case, $\reducedDegreeBefore{u} < \reducedDegree{v}$. Moreover, by the first invariant, the difference can be bounded from above as follows.
    \begin{flalign*}
    \reducedDegree{v}-\reducedDegreeBefore{u} & = d_{\noset{T}{F}}(v)+\ctok{\reducedDegree{v}} - d_{\noset{T}{F \cup \{ v \}}}(u)-\ctokBefore{\reducedDegreeBefore{u}} & \\
    & \leq d_{\noset{T}{F}}(v) - d_{\noset{T}{F \cup \{ v \}}}(u)  +|F| & \\
    & \leq d_{\noset{T}{F}}(v)- d_{\noset{T}{F}}(u)+|F|+1.
    \end{flalign*}
    As a consequence, we obtain that $d_{\noset{T}{F}}(v)- d_{\noset{T}{F}}(u) \geq k+|F|$ implying that
    $d_T(v)-d_T(u) \geq k$ by Observation~\ref{degrees_after_removal}.
    So if $vu \in E(T)$, then $u \in \longArcs{v}$ and the reduced degree of $u$ is correctly fixed by the
    algorithm. On the other hand, if $uv \in E(T)$, then
    $d_{\noset{T}{F}}(u)=d_{\noset{T}{F \cup \{ v \}}}(u)$ and
    $ \ctokBefore{\reducedDegreeBefore{u}} = \ctok{\reducedDegreeBefore{u}}$, so $\reducedDegreeBefore{u}$
    satisfies the first invariant also after restoring $v$, and thus it needs not to be modified.

    Let us now consider the case when $\reducedDegreeBefore{u} > r$. In that case, $\reducedDegreeBefore{u} > \reducedDegree{v}$. Moreover, by the first invariant, the difference can be bounded from above as follows.
    \begin{flalign*}
    \reducedDegreeBefore{u}-\reducedDegree{v} & = d_{\noset{T}{F \cup \{ v \}}}(u)+\ctokBefore{\reducedDegreeBefore{u}} - d_{\noset{T}{F}}(v) -\ctok{\reducedDegree{v}} & \\
    &\leq d_{\noset{T}{F \cup \{ v \}}}(u) -  d_{\noset{T}{F}}(v)+|F|+1 & \\
    &\leq d_{\noset{T}{F}}(u) - d_{\noset{T}{F}}(v)+|F|+2.
    \end{flalign*}
    As a consequence, we obtain that $d_{\noset{T}{F}}(u)- d_{\noset{T}{F}}(v) \geq k+|F|$ implying that
    $d_T(u)-d_T(v) \geq k$.
    So if $uv \in E(T)$, then $u \in \longArcs{v}$ and the reduced degree of $u$ is correctly fixed by the
    algorithm. On the other hand, if $vu \in E(T)$, then
    $d_{\noset{T}{F}}(u)=d_{\noset{T}{F \cup \{ v \}}}(u)-1$ and
    $ \ctokBefore{\reducedDegreeBefore{u}} = \ctok{\reducedDegreeBefore{u}}+1$, so again $\reducedDegreeBefore{u}$ satisfies the first invariant after
    the restoration of $v$ and thus needs no modifications.
\end{proof}

\subsection{Ensuring Fast Running Times of the $\dsRemove{n}{k}$
Methods in the Promise Model}\label{sub:dremp}

In this subsection we introduce the $\dsRemovePromise{n}{k}$ data structure.
It turns out that if we remove a bounded number of vertices contributing many $k$-long back arcs from the tournament, the remaining tournament has bounded number of $k$-long back arcs.
Most importantly, the removed vertices need to be in \textbf{every} small feedback vertex set of $T$ and thus we can safely remove them and operate on the smaller subtournament.
Such vertices are defined as the $k$-heavy set of a tournament in Definition~\ref{k_heavy_vertices_def}.

Next, we define the decomposition of a tournament $T$ that removes its $k$-heavy set.
The $\dsRemovePromise{n}{k}$ maintains such a decomposition by maintaining the $\dsRemove{n}{k}$ and allowing only initialization and $\REVERSEARCDREMP{u, v}$ operation.
The idea behind such a wrapper is to use it during updates (arc reversals) and use its maintained $\dsRemove{n}{k}$ during queries in the proof of Theorem~\ref{thm:FvstPromise}.
As long as we do not restore vertices removed by $\dsRemovePromise{n}{k}$, this approach ensures fast running times of $\dsRemove{n}{k}$ operations during such queries.

\begin{definition}[$k$-heavy set]\label{k_heavy_vertices_def}
    Let $T$ be a tournament, $k \in \mathbb{N}$ and $G_{\LONG}$
    be the $k$-long graph of $T$.
    The set
    $\heavy{T}{k} := \{v \in V(G_{\LONG})$ $:$
    $\degree{G_{\LONG}}{v} > k\}$
    is called the $k$-\emph{heavy set} of $T$.
\end{definition}

The lemma below explains the meaning and the importance
of the $k$-\emph{heavy set} with regard to $\FVST$ problem.
In essence, any feedback vertex set of size at most
$k$ contains $\heavy{T}{k}$ and if we remove the
vertices $\heavy{T}{k}$, only a small number of $k$-long
back arcs remains.
The lemma is stated in a slightly more generalized form,
because it is used in the next subsection as well.

\begin{lemma}\label{vertex_cover_fvst}
    Let $T$ be a tournament, let $k \in \mathbb{N}$ and let
    $G_{\LONG}$ be the $k$-long graph of $T$.
    Let $k \geq l \in \mathbb{N}$ and let $F$ be a feedback vertex
    set of $T$ of size at most $l$. Then $F$ is a vertex cover of
    $G_{\LONG}$, $\heavy{T}{k} \subseteq F$ and
    $|E(G_{\LONG}[V(G_{\LONG}) \setminus \heavy{T}{k}])| \leq k \cdot l$.
\end{lemma}

\begin{proof}
    For the sake of contradiction, assume there exists
    $e = \{v, u\} \in E(G_{\LONG})$, such that $v, u \notin F$.
    Without loss of generality, assume $vu \in E(T)$.
    Then, $\dist_T(v,u) \geq k \geq l$.
    Thus, $\dist_{\noset{T}{F}}(v,u) \geq 0$, by
     Observation~\ref{degrees_after_removal}. Hence, $vu \in E(\noset{T}{F})$ is a still a back arc in $\noset{T}{F}$.
    This contradicts Lemma~\ref{back_edge_triangle}.
    The rest of the statement follows from standard kernelization
    arguments for vertex cover described in detail in~\cite{Alman}.
    These arguments state that
    the vertices of degrees larger than $k$ must belong to a vertex
    cover
    of size at most $l \leq k$. If we remove vertices of degrees
    larger than $k$, we are left with at most $k l$ edges,
    as each vertex of the vertex cover can cover at most $k$
    edges.
\end{proof}

\begin{definition}[$k$-long back arcs decomposition]\label{k_long_back_arcs_decomposition}
    Let $T$ be a tournament and $k \in \mathbb{N}$. A tuple $(T,  \heavy{T}{k}, \noset{T}{\heavy{T}{k}})$ is called the $k$-\emph{long back arcs decomposition} of $T$.
\end{definition}

\begin{lemma}[Definition of $\dsRemovePromise{n}{k}$]\label{DREMP_def_lemma}
    For every $n \in \mathbb{N}$ and $k \in [n + 1]$ there exists a data structure $\dsRemovePromise{n}{k}$ that for a dynamic tournament $T$ on $n$ vertices\footnoteref{note1} altered by arc reversals maintains the $k$-long back arcs decomposition $(T,  \heavy{T}{k}, \noset{T}{\heavy{T}{k}})$ of $T$, under the promise that tournament $T$ always has a feedback vertex set of size at most $k$. The data structure supports the following operations: 
    \begin{enumerate}
        \item $\INITIALIZEDREMP{T}{k}$ -- initializes the data structure with tournament $T$ and parameter $k$, in time $\mathcal{O}(n^2)$
        \item $\REVERSEARCDREMP{v, u}$ -- reverses arc between vertices $v, u \in V(T)$ in $T$, in time $\mathcal{O}(k^5)$
    \end{enumerate}
\end{lemma}

\begin{algorithm}
    \SetKwInOut{Input}{Operation}
	\SetKwInOut{Output}{Output}
    \Input{$\dsRemovePromise{n}{k}.\REVERSEARCDREMP{v, u}$}
    \BlankLine

    $dL \leftarrow \dsRemove{n}{k}.\REVERSEARCDREM{v, u}$ \tcp*[r]{List of newly added/removed $k$-long back arcs in $T$}
    $dH \leftarrow \heavy{T}{k}$ $\triangle$ $\heavy{\overline{T}}{k}$ \tcp*[r]{Calculated using $dL$}
    $dH_{res} \leftarrow$ vertices that need to be removed from $\heavy{T}{k}$ to obtain $\heavy{\overline{T}}{k}$\;
    $dH_{rem} \leftarrow$ vertices that need to be added to $\heavy{T}{k}$ to obtain $\heavy{\overline{T}}{k}$\;

    \For{$w \in dH_{res}$}{
        $\dsRemove{n}{k}.\RESTOREVERTEX{w}$\;
    }

    \For{$w \in dH_{rem}$}{
        $\dsRemove{n}{k}.\REMOVEVERTEX{w}$\;
    }

    \caption{Pseudocode for $\REVERSEARCDREMP{v, u}$ operation of the $\dsRemovePromise{n}{k}$ maintaining the $k$-long back arcs decomposition $(T, \heavy{T}{k}, \noset{T}{\heavy{T}{k}})$ of tournament $T$}
    \label{fig:update_impl_fvst}
\end{algorithm}

\begin{proof}
    The data structure internally maintains $\dsRemove{n}{k}$ data structure of Lemma~\ref{DREM_def_lemma} for $T$ and $F=\heavy{T}{k}$. By Lemma~\ref{vertex_cover_fvst}, set $\heavy{T}{k}$ has size at most $k$ at all times.
    By Lemma~\ref{DREM_def_lemma} data structure $\dsRemove{n}{k}$ can be initialized in $\mathcal{O}(n^2)$ time.
    This implies the implementation of $\dsRemovePromise{n}{k}.\INITIALIZEDREMP{T}{k}$ operation in $\mathcal{O}(n^2)$ time, since we can easily calculate set $\heavy{T}{k}$ in time $\mathcal{O}(n^2)$.

    Let us now move to the $\REVERSEARCDREMP{v, u}$ operation. The pseudocode is presented in Algorithm~\ref{fig:update_impl_fvst}.
    Without loss of generality, let $uv \in E(T)$ be the arc we want to reverse.
    Let $(\overline{T},  \heavy{\overline{T}}{k}, \noset{\overline{T}}{\heavy{\overline{T}}{k}})$ be the $k$-long back arcs decomposition of $\overline{T}$, which is a tournament $T$ with arc $uv$ reversed.
    Invoke the $\dsRemove{n}{k}.\REVERSEARCDREM{v, u}$ operation and let $dL$ be the returned list of all newly added/removed $k$-long back arcs in $T$ of size at most $12 \maxdb{T} + 1$.
    By Lemma~\ref{small_buckets_fvst}, the invocation takes $\mathcal{O}(k^4)$ time.
    Notice, that $|dL| \leq 12\maxdb{T} + 1 \leq 12(2k + 1) + 1 = 24k + 13$, by Lemma~\ref{small_buckets_fvst}.
    Next, calculate list $dH := \heavy{T}{k}$ $\triangle$ ${\heavy{\overline{T}}{k}}$.
    This can be done by iterating over $e \in dL$ and checking the degrees in the $k$-long graph $\overline{G}$ of $\overline{T}$ (stored as sizes of updated array sets $\longArcs{\cdot}$) for both endpoints of $e$. 
    Then, the list $dH$ is partitioned into two lists $dH_{res}$ and $dH_{rem}$, depending on whether the vertex is in $\heavy{T}{k}$ (to be removed from the heavy set) or in ${\heavy{\overline{T}}{k}}$ (to be added to the heavy set).
    First, the algorithm iterates over $w \in dH_{res}$ and
    invokes $\dsRemove{n}{k}.\RESTOREVERTEX{w}$. Then, the
    algorithm iterates over $w \in dH_{rem}$ and invokes
    $\dsRemove{n}{k}.\REMOVEVERTEX{w}$.
    In this way, the invariant $|F| \leq k$ does not
    get violated.

    By Lemma~\ref{DREM_def_lemma}, every such operation (for vertex $w$) takes $\mathcal{O}((k + 1)^3 \cdot (\maxdb{\overline{T}} + 1) + (|\longArcs{w}| + 1) \cdot (k + 1)) = \mathcal{O}(k^4 + (|\longArcs{w}| + 1) \cdot (k + 1))$ time, by Lemma~\ref{small_buckets_fvst}.
    Notice, that if $w \in dH_{res}$ then $|\longArcs{w}| = \degree{\overline{G}}{w} \leq k$. On the other hand, if $w \in dH_{rem}$ then $|\longArcs{w}| = \degree{\overline{G}}{w} \leq k + |dL| \leq k + 24k + 13 = 25k + 13$, because $w \notin \heavy{T}{k}$.
    Thus, a single vertex restoration/removal takes $\mathcal{O}(k^4)$ time.
    Hence, the whole $\REVERSEARCDREMP{u, v}$ operation takes $\mathcal{O}(k^5)$ time, because $|dH| \leq k$.
\end{proof}

\subsection{Data Structure for Dynamic $\FVST$ Problem}\label{sub:fvst_branching}

In this subsection we present the main result of
Section~\ref{sec:dynamicFVST}, i.e., we prove
Theorem~\ref{thm:FvstPromise} by showing the dynamic parameterized
algorithm determining if a dynamic tournament admits a feedback
vertex set of size bounded by the parameter.
Our approach is to use the $\dsRemovePromise{n}{g(K)}$ data
structure (where $K$ is the parameter of the problem) to maintain
the $g(K)$-long back arcs decomposition of
the maintained tournament $T$. This data structure internally
maintains the data structure $\dsRemove{n}{g(K)}$, which allows
vertex removals and restorations. This is crucial for queries.
Upon a query, our approach is to run a standard branching
algorithm shown in Algorithm~\ref{fig:branchingFVSToverview}
for finding if
$\FVST(T) \leq K$. This algorithm finds a triangle in the
tournament, and branches trying to remove one of the triangles
vertices. This requires vertex removals and restorations (provided
by $\dsRemove{n}{g(K)}$), but also a procedure for detecting
triangles. To support triangle detection we follow the general
approach from Section~\ref{sec:triaQ}, but we have to adapt the
procedures to be able to find triangles in $\noset{T}{F}$, where
$F$ is the set of vertices currently removed. The procedures
presented in Section~\ref{sec:triaQ} need access to the empty set
of the tournament where they look for a triangle, which is now
$\noset{T}{F}$. So before we move on to describing the procedures,
we relate the empty set of $\noset{T}{F}$ with the set
$\reducedEmptydb$ maintained by $\dsRemove{n}{g(K)}$.

\begin{lemma}\label{tokenized_empty_fvst}
    Let $\dsRemove{n}{k}$ be the data structure of Lemma~\ref{DREM_def_lemma} maintaining tournament $T$ and set $F$, where $k \in [n + 1]$.
    Let $\noset{\emptydb}{F}$ be the empty set of $\noset{T}{F}$.
    Then $|\noset{\emptydb}{F}| \leq |\reducedEmptydb| \leq |\noset{\emptydb}{F}| + 4|F|$.
    Moreover, set $\noset{\emptydb}{F}$ can be calculated in time $\mathcal{O}((|\reducedEmptydb| + 1) \cdot (|F| + 1))$.
\end{lemma}

\begin{proof}
    By Lemma~\ref{tokenized_ops_fvst}, for every $d \in \mathbb{Z}$, there is a non-empty interval $[l_d, r_d] \subseteq \mathbb{Z}$, such that $[l_d, r_d] = \{ i \in \mathbb{Z}: i-\ctok{i}=d \}$. By Lemma~\ref{tokenized_ops_fvst}, $[l_d,r_d] \subseteq [d,d+|F|]$.
    By the first invariant, for $d \in [n-|F|]$ it holds that
    $\db{\noset{T}{F}}[d] = \bigcup_{i \in [l_d, r_d]} \reducedDegreeBucket{i}$, i.e., the vertices of in-degree $d$ in $\noset{T}{F}$ are precisely the vertices which reduced degrees range from $l_d$ to $r_d$.

    Consider $i \in \reducedEmptydb \subseteq [n]$. Since $-1 \leq i-\ctok{i} \leq n-1$ for $i \in [n]$, one of the following four cases holds:
    \begin{enumerate}
     \item $i-\ctok{i}=-1$: by Lemma~\ref{tokenized_ops_fvst}
     it holds that $[l_{-1},r_{-1}] \subseteq [-1,-1+|F|]$, so there is at most $|F|$ such values of $i \in [n]$
     \item $i-\ctok{i} \in [n-|F|,n-1]$: since $i - \ctok{i} \leq i$, there is at most $|[n-|F|,n-1]|=|F|$ such values of $i \in [n]$
     \item $i-\ctok{i} \in [n-|F|]$ and $i-\ctok{i} \neq j-\ctok{j}$ for all $j \in [n], j \neq i$: in that case for $d=i-\ctok{i}$ we get the interval $[l_d=i,r_d=i]$ of size one, thus $\db{\noset{T}{F}}[d] = \reducedDegreeBucket{i}$ implying that $d \in \noset{\emptydb}{F}$, and in consequence there is at most $|\noset{\emptydb}{F}|$ such values of $i \in [n]$
     \item $i-\ctok{i} \in [n-|F|]$ and $i-\ctok{i}=(i+1)-\ctok{i+1}$ or $i-\ctok{i}=(i-1)-\ctok{i-1}$: the condition $i-\ctok{i}=(i+1)-\ctok{i+1}$ holds if and only if there is a token at position $i+1$, in the other case the situation is similar. Thus, there is in total at most $2|F|$ such values of $i \in [n]$
    \end{enumerate}

It is easy to see that if $i-\ctok{i}\neq(i+1)-\ctok{i+1}$ and $i-\ctok{i}\neq(i-1)-\ctok{i-1}$ then $i-\ctok{i} \neq j-\ctok{j}$ for all $j \in [n], j \neq i$, because of the second invariant of $\dsRemove{n}{k}$. Thus, we considered all possible cases.

This in total gives the bound of
$|\reducedEmptydb| \leq |\noset{\emptydb}{F}| + 4|F|$.

    The inequality $|\noset{\emptydb}{F}| \leq |\reducedEmptydb|$ is straightforward.
    To calculate $\noset{\emptydb}{F}$, we iterate over $i \in \reducedEmptydb$ and if $i - \ctok{i} \in [n-|F|]$ then we calculate interval $[l_{i-\ctok{i}}, r_{i-\ctok{i}}]$ using Lemma~\ref{tokenized_ops_fvst}.
    Next, if $i = l_{i-\ctok{i}}$ then we check whether $[l_{i-\ctok{i}}, r_{i-\ctok{i}}] \subseteq \reducedEmptydb$.
    In total, the algorithm takes $\mathcal{O}((|\reducedEmptydb| + 1) \cdot (|F| + 1))$.
\end{proof}

Together with Lemma~\ref{small_empty_fvst} of Subsection~\ref{sub:fvst_prelim} (bounding the size of the empty set) the above lemma gives the following corollary.

\begin{corollary}\label{corollary_empty_fvst}
    Let $\dsRemove{n}{k}$ be the data structure of Lemma~\ref{DREM_def_lemma} maintaining tournament $T$ and set $F$, where $k \in [n + 1]$.
    Let $G$ be the $k$-long graph of $T$.
    Then $|\reducedEmptydb| \leq \FVST(\noset{T}{F}) \cdot (2(k + |F|) + |E(\induced{G}{V(\noset{T}{F})})| + 2\FVST(\noset{T}{F}) + 5) + 4|F|$.
\end{corollary}

\begin{proof}
    Let $\noset{\emptydb}{F}$ be the empty set of $\noset{T}{F}$.
    Let $\noset{\lbackdb}{F}$ be the set of all $(k + |F|)$-long back arcs in $\noset{T}{F}$.
    By Lemma~\ref{small_empty_fvst}, $$|\noset{\emptydb}{F}| \leq \FVST(\noset{T}{F}) \cdot (2(k + |F|) + |\noset{\lbackdb}{F}| + 2\FVST(\noset{T}{F}) + 5).$$

    By Observation~\ref{degrees_after_removal} it holds that $|\noset{\lbackdb}{F}| \leq |E(\induced{G}{V(\noset{T}{F})})|$.
    Thus, $$|\noset{\emptydb}{F}| \leq \FVST(\noset{T}{F}) \cdot (2(k + |F|) + |E(\induced{G}{V(\noset{T}{F})})| + 2\FVST(\noset{T}{F}) + 5).$$
    By Lemma~\ref{tokenized_empty_fvst}, $|\reducedEmptydb| \leq |\noset{\emptydb}{F}| + 4|F|$.
    Hence, the claim follows.
\end{proof}

In order to implement the branching algorithm
to determine if $\FVST(T) \leq K$, we introduce an analogue theorem to  Theorem~\ref{common_query_triangles} of Section~\ref{r:triangles}.
It shows how to implement method $\FTRIANGLEFVST{}$ which, given $\dsRemove{n}{k}$ maintaining tournament $T$ and set $F \subseteq V(T)$, finds a triangle in tournament $\noset{T}{F}$ or states that there are no triangles.
Similarly as in Section~\ref{r:triangles}, we define two subprocedures:
\begin{enumerate}
    \item $\MINFVST$ -- returns a vertex $v$ of minimum in-degree in $\nopref{(\noset{T}{F})}$, together with $d_{\nopref{(\noset{T}{F})}}(v)$; if $\nopref{(\noset{T}{F})} $ is empty, then it returns $null$,
    in time $t_1(T,F)$
    \item $\INCOMINGFVST{v}{l}$ -- returns $l \in \mathbb{N}$ in-neighbours of $v \in V(\nopref{(\noset{T}{F})})$ in $\nopref{(\noset{T}{F})}$ (or as many as there exist), in time $t_2(l,T,F)$ for a non-decreasing (with respect to $l$) function $t_2$
\end{enumerate}
Tournament $T$ and set $F$ shall be clear from the context.

\begin{theorem}\label{common_query_triangles_fvst}
    Let $\dsRemove{n}{k}$ be the data structure of Lemma~\ref{DREM_def_lemma} maintaining tournament $T$ and set $F$, where $k \in [n + 1]$.
    Let us assume that we are given the following two subprocedures:
    \begin{enumerate}
        \item $\MINFVST$ -- returns a vertex $v$ of minimum in-degree in $\nopref{(\noset{T}{F})}$, together with $d_{\nopref{(\noset{T}{F})}}(v)$; if $\nopref{(\noset{T}{F})}$ is empty, then it returns $null$,
        in time $t_1(T,F)$
        \item $\INCOMINGFVST{v}{l}$ -- returns $l \in \mathbb{N}$ in-neighbours of $v \in V(\nopref{(\noset{T}{F})})$ in $\nopref{(\noset{T}{F})}$ (or as many as there exist), in time $t_2(l,T,F)$ for a non-decreasing (with respect to $l$) function $t_2$
    \end{enumerate}
    Then it is possible to implement method $\FTRIANGLEFVST{}$, that returns a triangle in $\noset{T}{F}$ or states that there are no triangles, in time $\mathcal{O}(t_1(T,F) + t_2(\mindt{\nopref{(\noset{T}{F})}}, T,F) + \mindt{\nopref{(\noset{T}{F})}})$.
\end{theorem}

\begin{proof}
    The pseudocode is identical to the one presented in Algorithm~\ref{fig:common_query_triangles} for $\FTRIANGLE{}$ method.
    The proof is analogous to the proof of Theorem~\ref{common_query_triangles}.
\end{proof}

\begin{lemma}\label{empty_impl_promise_fvst}
    Let $\dsRemove{n}{k}$ be the data structure of Lemma~\ref{DREM_def_lemma} maintaining tournament $T$ and set $F$, where $k \in [n + 1]$.
    Let $\noset{\emptydb}{F}$ be the empty set of $\noset{T}{F}$.
    If $\noset{\emptydb}{F}$ is given as an array set over $[n]$ then subprocedure $\MINFVST$ can be implemented in time $\mathcal{O}(|\noset{\emptydb}{F}| + |F|)$.
\end{lemma}

\begin{proof}
    If $n-|F| = 0$, we return $null$.
    If $|\noset{\emptydb}{F}| = 0$, we also return $null$, as then $\nopref{(\noset{T}{F})}$ is empty by Lemma~\ref{prefix_transitive}.
    Otherwise, let $d$ be the minimal element in $\noset{\emptydb}{F}$.
    Similarly as in the proof of Lemma~\ref{empty_impl_promise}, we find the smallest
    integer $d' \in [n-|F|]$ larger than $d$, but not present
    in $\noset{\emptydb}{F}$ in time $\mathcal{O}(|\noset{\emptydb}{F}|)$.
    Then, we return an arbitrary vertex $v \in \db{\noset{T}{F}}[d']$ and $d_{\nopref{(\noset{T}{F})}}(v)$.

    In order to extract a single vertex from $\db{\noset{T}{F}}[d']$ we calculate set $I_{d'} = \{i \in \mathbb{Z}$ $:$ $i - \ctok{i} = d'\} \subseteq [n]$, using Lemma~\ref{tokenized_ops_fvst}, and iterate over $i \in I_{d'}$ searching for a non-empty bucket $\reducedDegreeBucket{i}$. This takes $\mathcal{O}(|F|)$ time, since $|I_{d'}| \leq |F| + 1$.

    Let $\noset{P}{F}$ be the prefix of $\noset{T}{F}$.
    Notice that, $d_{\nopref{(\noset{T}{F})}}(v) = d_{\noset{T}{F}}(v) - |\noset{P}{F}|$, by Lemma~\ref{prefix_transitive}.
    Given $\noset{\emptydb}{F}$ as an array set, we can calculate $|\noset{P}{F}|$ in time $\mathcal{O}(|\noset{\emptydb}{F}|)$ by finding a minimum element in $\noset{\emptydb}{F}$, by Lemma~\ref{prefix_transitive}.
    In-degree $d_{\noset{T}{F}}(v)=\reducedDegree{v}-\ctok{\reducedDegree{v}}$ can be calculated in $\mathcal{O}(|F|)$ time using the information maintained by $\dsRemove{n}{k}$.
    Hence, the total running time of the described implementation of $\MINFVST$ is $\mathcal{O}(|\noset{\emptydb}{F}| + |F|)$.
\end{proof}

\begin{algorithm}
    \SetKwInOut{Input}{Algorithm}
	\SetKwInOut{Output}{Output}
    \Input{$\INCOMINGFVST{v}{l}$}
    \Output{A list of $l$ (or as many as there exist) incoming neighbours of $v \in \nopref{(\noset{T}{F})}$}
    \BlankLine

    $l_v \leftarrow \reducedDegree{v} - 2|\rmSet| - k - 1$\;
    $r_v \leftarrow \reducedDegree{v} + 2|\rmSet| + k + 1$\;
    $\noset{\emptydb}{F} \leftarrow$ the empty set of $\noset{T}{F}$ \tcp*[r]{Calculated using Lemma~\ref{tokenized_empty_fvst}}
    $|\noset{P}{F}| \leftarrow$ the length of the prefix of $\noset{T}{F}$\;
    $I \leftarrow \{i \in \mathbb{Z}$ $:$ $i - \ctok{i} = |\noset{P}{F}|\}$ \tcp*[r]{Calculated using Lemma~\ref{tokenized_ops_fvst}}
    $q \leftarrow \min(I)$\;

    \tcp{Break after finding $l$ incoming neighbours}
    Find in-neighbours $u$ of $v$ in $\nopref{(\noset{T}{F})}$, such that $\reducedDegree{u} > r_v$ by iterating over $\longArcs{v}$\;
    Find in-neighbours of $v$ in $\nopref{(\noset{T}{F})}$
    from $\reducedDegreeBucket{[l_v, r_v] \cap [n]}$\;
    Find the remaining in-neighbours of $v$ in $\nopref{(\noset{T}{F})}$ by iterating over
    $\reducedDegreeBucket{[q, \min(l_v, n))}$ \; 

    \Return At most $l$ incoming neighbours of $v$ in $\nopref{\noset{T}{F}}$
    \caption{Pseudocode for $\INCOMING{v}{l}$ subprocedure given a $\dsRemove{n}{k}$ maintaining tournament $T$ and set of removed vertices $F$}
    \label{fig:incoming_impl_fvst}
\end{algorithm}

\begin{lemma}\label{incoming_impl_fvst}
    Using $\dsRemove{n}{k}$ data structure of
    Lemma~\ref{DREM_def_lemma} maintaining tournament $T$ and set
    $F$,
    the method $\INCOMINGFVST{v}{l}$ can be implemented in
    $\mathcal{O}((|\reducedEmptydb| + 1) \cdot (|F| + 1) + |\longArcs{v}| + (k + 1) \cdot (\maxdb{T} + 1) \cdot (|F| + 1) + l)$ time.
\end{lemma}

\begin{proof}
    Let $\noset{\emptydb}{F}$ be the empty set of $\noset{T}{F}$ and $\noset{P}{F}$ be the prefix of $\noset{T}{F}$.
    Let us describe how to iterate over all incoming neighbours of $v$ in $\nopref{(\noset{T}{F})}$.
    From the algorithm it will be obvious that we can stop the procedure at any moment after finding $l$ of them.

    Let $l_v = \reducedDegree{v} - 2|\rmSet| - k - 1$ and $r_v = \reducedDegree{v} + 2|\rmSet| + k + 1$.
    First of all, we calculate $\noset{\emptydb}{F}$ and store it as a list using Lemma~\ref{tokenized_empty_fvst} in time $\mathcal{O}((|\reducedEmptydb| + 1) \cdot (|F| + 1))$.
    Secondly, we calculate $|\noset{P}{F}|$ in time $\mathcal{O}(|\noset{\emptydb}{F}|)$ by finding the minimum element in $\noset{\emptydb}{F}$, using Lemma~\ref{prefix_transitive}.
    Since $v \in V(\nopref{(\noset{T}{F})})$, $-1 < |\noset{P}{F}| < |V(\noset{T}{F})| \leq n$.
    Let $I = \{i \in \mathbb{Z}$ $:$ $i - \ctok{i} = |\noset{P}{F}|\}$.
    By Lemma~\ref{tokenized_ops_fvst}, it can be calculated in $\mathcal{O}(|F|)$ time and $\emptyset \neq I \subseteq [n]$.
    We set $q = \min(I)$.
    Notice, that $u \in V(\nopref{(\noset{T}{F})})$ if and only if $u \in V(\noset{T}{F})$ and $\reducedDegree{u} \geq q$, by Lemma~\ref{prefix_transitive}, Observation~\ref{monotonic_tokenized_degrees} and the first invariant of the $\dsRemove{n}{k}$ data structure. Thus, in the search for the in-neighbors of $v$ in  $\nopref{(\noset{T}{F})}$, it suffices to iterate over reduced degree buckets larger or equal to $q$.

    We partition the search for in-neighbours into three phases (described respectively in lines $7, 8, 9$ in Algorithm~\ref{fig:incoming_impl_fvst}).

    In the first phase (line $7$), we  iterate over $u \in \longArcs{v}$. If $u \notin \rmSet$ and $uv \in E(T)$ and $\reducedDegree{u} > r_v$, then we add $u$ to the list of in-neighbours of $v$ in $\nopref{(\noset{T}{F})}$ ($u \in V(\nopref{(\noset{T}{F})})$ by Observation~\ref{monotonic_tokenized_degrees} and the fact that $v \in V(\nopref{(\noset{T}{F})})$ implies $r_v \geq q$). This phase takes $\mathcal{O}(|\longArcs{v}|)$.

    In the second phase (line $8$), we iterate over $i \in [l_v, r_v] \cap [n]$ and for every $u \in \reducedDegreeBucket{i}$, if $uv \in E(\nopref{(\noset{T}{F})})$ then we add $u$ to the list of in-neighbors.
    Checking if $uv \in E(\nopref{(\noset{T}{F})})$ can be done in constant time by checking whether $uv \in E(T)$ and whether $\reducedDegree{u} \geq q$.
    This phase takes $\mathcal{O}((|F| + k + 1) \cdot (\maxdb{T} + 1) \cdot (|F| + 1)) = \mathcal{O}((k + 1) \cdot (\maxdb{T} + 1) \cdot (|F| + 1))$ time, by Observation~\ref{tokenized_buckets_size_fvst}.

    In the third phase (line $9$), we iterate over $i \in [q, \min(l_v, n))$. For every $u \in \reducedDegreeBucket{i}$, if $uv \in E(T)$ then we add $u$ to the list of in-neighbors. Observe that $u \in V(\nopref{(\noset{T}{F})}))$, since $\reducedDegree{u} \geq q$ and $u \in V(\noset{T}{F})$.
    If we fail to add $u$ to the list of in-neighbors, then $vu \in E(T)$, thus $u \in \longArcs{v}$, by Observation~\ref{monotonic_tokenized_degrees}.
    Hence, we fail to add $u$ in phase three only $|\longArcs{v}|$ times. Moreover, we encounter empty bucket $\reducedDegreeBucket{i}$ only $|\reducedEmptydb|$ times.
    To sum up, this phase takes a total of $\mathcal{O}(|\longArcs{v}|+|\reducedEmptydb|+l)$ time.

    Therefore, the time for executing the entire procedure totals to $\mathcal{O}((|\reducedEmptydb| + 1) \cdot (|F|+1) + |\longArcs{v}| + (k + 1) \cdot (\maxdb{T} + 1) \cdot (|F| + 1) + l)$ time.
    Notice, that the resulting list does not contain any duplicates.
    It remains to prove that we take into account all of the in-neighbours of $v$ in $\nopref{(\noset{T}{F})}$ (unless we already found $l$ of them).
    Let $u \in V(\nopref{(\noset{T}{F})})$ be such that $uv \in E(\nopref{(\noset{T}{F})})$. If $\reducedDegree{u} > r_v$ then $u \in \longArcs{v}$, by Observation~\ref{monotonic_tokenized_degrees}.
    Thus, such $u$ is found during the iteration over $\longArcs{v}$.
    On the other hand, we check all possible candidates for $u$ such that $\reducedDegree{u} \leq r_v$.
\end{proof}

The above three lemmas/theorems give the following corollary.

\begin{corollary}\label{common_query_triangles_corollary_fvst}
    Let $\dsRemove{n}{k}$ be the data structure of Lemma~\ref{DREM_def_lemma} maintaining tournament $T$ and set $F$, where $k \in [n + 1]$.
    Let $\noset{\emptydb}{F}$ be the empty set of $\noset{T}{F}$ given as an array set over $[n]$.
    Then it is possible to implement method $\FTRIANGLEFVST{}$ in time $\mathcal{O}((|\reducedEmptydb| + 1) \cdot (|F| + 1) + \max_{v \in V(\noset{T}{F})}(|\longArcs{v}|) + (k + 1) \cdot (\maxdb{T} + 1) \cdot (|F| + 1))$.
\end{corollary}

\begin{proof}
    Subprocedure $\MINFVST$ can be implemented to run in $\mathcal{O}(|\noset{\emptydb}{F}| + |F|)$ time, by Lemma~\ref{empty_impl_promise_fvst} and Lemma~\ref{tokenized_empty_fvst} ($|\noset{\emptydb}{F}| \leq |\reducedEmptydb|$).
    By Lemma~\ref{incoming_impl_fvst}, subprocedure $\INCOMINGFVST{v}{l}$ can be implemented to run in time $\mathcal{O}((|\reducedEmptydb| + 1) \cdot (|F| + 1) + \max_{v \in V(\noset{T}{F})}(|\longArcs{v}|) + (k + 1) \cdot (\maxdb{T} + 1) \cdot (|F| + 1))$, because $\mindt{\nopref{(\noset{T}{F})}} \leq |\noset{\emptydb}{F}|$ by Fact~\ref{minimal_degree}.
\end{proof}

We move on to proving our main result, which is
Theorem~\ref{thm:FvstPromise}.
Algorithm~\ref{fig:branchingFVST}
presented below is an adaptation of the standard static branching
algorithm for
$\FVST$
(Algorithm~\ref{fig:branchingFASToverview})
supported by the data structure $\dsRemove{n}{k}$.
Its correctness follows from Fact~\ref{cycles_are_triangles}.

\begin{algorithm}
    \SetKwInOut{Input}{Algorithm}
	\SetKwInOut{Output}{Output}
    \Input{$\dsRemove{n}{l}.\FINDFVST{k}$}
    \Output{Verify if $\FVST(\noset{T}{F}) \leq k - |\rmSet|$}
    \If{$k - |\rmSet| < 0$}{\Return FALSE \;}
    \If{$|\reducedEmptydb|$ is too big}{\Return FALSE \;}
    $\noset{\emptydb}{F} \gets$ the empty set of $\noset{T}{F}$ \tcp*[r]{Calculated using Lemma~\ref{tokenized_empty_fvst}}

    \If{$\noset{T}{F}$ is acyclic}{\Return TRUE \;}
    \If{$k - |\rmSet| = 0$}{ \Return FALSE \;}

    $uvw \gets \FTRIANGLEFVST{}$ \tcp*[r]{Implementation of Corollary~\ref{common_query_triangles_corollary_fvst}}

    \For{$x \in \{ v, u, w \}$}{
        $\dsRemove{n}{l}.\REMOVEVERTEX{x}$ \;
        $a \gets \dsRemove{n}{l}.\FINDFVST{k}$ \;
        $\dsRemove{n}{l}.\RESTOREVERTEX{x}$ \;

       \If{$a$}{ \Return TRUE \;}
    }
    \Return FALSE \;

    \caption{Pseudocode for $\FINDFVST{k}$ operation of the $\dsRemove{n}{l}$ maintaining tournament $T$ and set $F$, such that $F \supseteq \heavy{T}{l}$. Parameter $k$ must be from $[l + 1]$.}
    \label{fig:branchingFVST}
\end{algorithm}

\paragraph*{Proof of Theorem~\ref{thm:FvstPromise}}

\begin{proof}
    Without loss of generality, we can assume $0 \leq K < g(K) \leq n$, where $n$ is the number of vertices of input tournament $T$.
    We set up the $\dsRemovePromise{n}{g(K)}$ data structure (see Lemma~\ref{DREMP_def_lemma}) maintaining the $g(K)$-long back arcs decomposition of $T$.
    By Lemma~\ref{DREMP_def_lemma} it can be initialized with $T$ and $g(K)$ in $\mathcal{O}(n^2)$ time. In order to perform updates, we invoke $\dsRemovePromise{n}{g(K)}.\REVERSEARCDREMP{v, u}$.
    By Lemma~\ref{DREMP_def_lemma}, updates take $\mathcal{O}(g(K)^5)$ time.

    In order to perform queries, we implement a standard branching algorithm for the static $\FVST$ from Fact~\ref{folklore_branching_fvst}, which is the main part of the query algorithm.
    The branching is implemented using the $\dsRemove{n}{g(K)}$ data structure (see Lemma~\ref{DREM_def_lemma}), which is maintained by the $\dsRemovePromise{n}{g(K)}$ data structure, with the set of removed vertices set to $F=\heavy{T}{g(K)}$. We extend the $\dsRemove{n}{g(K)}$ with the $\FINDFVST{k}$ method (for $k \in [g(K) + 1]$) presented in Algorithm~\ref{fig:branchingFVST}. The method $\FINDFVST{k}$ temporarily expands the set $F$, maintaining the invariant that $\heavy{T}{g(K)} \subseteq F$.
    In each recursive call, the method $\FINDFAST{k}$ determines whether $\FVST(\noset{T}{F}) \leq k-|F|$.
    We now move on to a more detailed description of this method.

    First of all, we verify if $k-|F| \geq 0$, otherwise we
    return FALSE.
    Next we verify if the claim of Corollary~\ref{corollary_empty_fvst} holds. Recall
    that Corollary~\ref{corollary_empty_fvst} states, that
    $|\reducedEmptydb| \leq \FVST(\noset{T}{F}) (2 (k'+|F|)+ |E(G[V(\noset{T}{F})])| + 2 \FVST(\noset{T}{F})+5) + 4|F|$ for the data structure $\dsRemove{n}{k'}$, where $G$ is a $k'$-long graph of $T$. In our case $k'=g(K)$. Moreover, since $F \supseteq \heavy{T}{g(K)}$, Lemma~\ref{vertex_cover_fvst} implies that $|E(G[V(\noset{T}{F})])| \leq g(K) \cdot k$. Thus, if $\FVST(\noset{T}{F}) \leq k-|F|$, the following bound holds:
    $|\reducedEmptydb| \leq (k-|F|) \cdot (2(g(K) + |F|) +  k \cdot g(K) + 2(k-|F|) + 5) + 4|F|$. If this condition does not hold, then by counter position $\FVST(\noset{T}{F}) > k-|F|$ and we correctly return FALSE.
    If, on the other hand, both the above conditions hold then we are able to execute the next steps efficiently.
     We next compute the empty set of $\noset{T}{F}$ that we refer to as $\noset{\emptydb}{F}$. We do it using Lemma~\ref{tokenized_empty_fvst} in $\mathcal{O}((|\reducedEmptydb|+1)(|F|+1))$ time and we store\footnote{We can allocate an empty array set over $[n]$ at the initialization and clear it before and after entering the recursive call. It does not affect correctness nor time complexity.} $\noset{\emptydb}{F}$ as an array set over $[n]$.
Next, we check whether $\noset{T}{F}$ is acyclic by checking if $\noset{\emptydb}{F}$ is empty. If $\noset{\emptydb}{F}$ is empty, then $\noset{T}{F}$ is acyclic by Fact~\ref{transitive_buckets}.
If $\noset{T}{F}$ is acyclic then we return TRUE.
    If $\noset{T}{F}$ contains a cycle and $k-|F| = 0$ then
    we return FALSE.
    This concludes the preliminary part, which takes a total of $\mathcal{O}((|\reducedEmptydb|+1)(|F|+1))=\mathcal{O}((k-|F|) \cdot (2(g(K) + |F|) +  k \cdot g(K) + 2(k-|F|) + 5) + 4|F|) \cdot (|F|+1)$ $=\mathcal{O}((k+1)^3 g(K))$ time.

    In the main part of the method, we invoke procedure $\FTRIANGLEFVST{}$ implemented by Corollary~\ref{common_query_triangles_corollary_fvst}.
    Note that procedure $\FTRIANGLEFVST{}$ always returns a triangle $uvw$, because $\noset{T}{F}$ has at least one cycle.
    We branch on vertices $x \in \{ u, v, w \}$ and try to remove each of them using $\dsRemove{n}{g(K)}.\REMOVEVERTEX{x}$ method. Then we apply the recursion and store the returned value, which is either TRUE or FALSE.
    When returning from the recursive call, we invoke the $\dsRemove{n}{g(K)}.\RESTOREVERTEX{x}$ method.
    It restores the vertex removed before the recursive call, so the data structure $\dsRemove{n}{g(K)}$ maintains the same tournament and set of removed vertices as before the recursive call.
    This is a fairly standard branching algorithm, and the correctness is immediate from Fact~\ref{cycles_are_triangles}.

    We move on to the analysis of the running time of
    $\dsRemove{n}{g(K)}.\FINDFVST{k}$ method. The preliminary part takes $\mathcal{O}(g(K)(k + 1)^3)$ time.
    We can observe that $\max_{v \in V(\noset{T}{F})}(|\longArcs{v}|) \leq g(K)$, because $V(\noset{T}{F}) \subseteq V(T) \setminus \heavy{T}{g(K)}$.
    Thus, procedure $\FTRIANGLEFVST{}$ runs in time $\mathcal{O}(g(K)(k + 1)^3 + g(K)(k + 1)(\maxdb{T} + 1))$, by Corollary~\ref{common_query_triangles_corollary_fvst}.
    Both methods $\dsRemove{n}{g(K)}.\REMOVEVERTEX{x}$ and $\dsRemove{n}{g(K)}.\REMOVEVERTEX{x}$ take $\mathcal{O}(g(K) \cdot (\maxdb{T} + 1) \cdot (|F| + 1)^2 + (|\longArcs{x}| + 1) \cdot (|F| + 1)) = \mathcal{O}(g(K)(k + 1)^2(\maxdb{T} + 1))$ time, by Lemma~\ref{DREM_def_lemma}.
    To sum up, each recursive call costs us $\mathcal{O}(g(K)(k + 1)^3 + g(K)(k + 1)^2(\maxdb{T} + 1))$ time.

    We now describe how to perform queries.
    First, verify if $\maxdb{T} \leq 2k + 1$. If this fails to hold, then by
    Lemma~\ref{small_buckets_fvst} we get that $\FVST(T) > k$, so we correctly return FALSE.
    Next, call $\dsRemove{n}{g(K)}.\FINDFVST{K}$ on the $\dsRemove{n}{g(K)}$, maintaining tournament $T$ and set $F$, maintained by the $\dsRemovePromise{n}{g(K)}$.
    Thanks to the above check, such a call takes $\mathcal{O}(g(K)(K + 1)^3)$ time.
    Note, that after the call, the $\dsRemove{n}{g(K)}$ data structure maintains exactly the same tournament and set of removed vertices.
    The invoked method determines whether $\FVST(\noset{T}{F}) \leq K - |F|$ for  $F = \heavy{T}{g(K)}$.
    However, $\FVST(T) \leq K$ if and only if $\FVST(\noset{T}{F}) \leq K - |F|$, by Lemma~\ref{vertex_cover_fvst}.
    Hence, the query algorithm is correct.
    The original invocation of method $\FINDFVST{K}$ calls itself recursively at most $3^{K}$ times, every time with the same argument.
    Thus, the query takes $\mathcal{O}(3^K g(K)(K + 1)^3)$ time as claimed.
\end{proof}

%
%
%

\end{document}